\documentclass[12pt,journal,onecolumn,draftcls]{IEEEtran}

\usepackage{epsfig}
\usepackage{times}
\usepackage{float}
\usepackage{afterpage}
\usepackage{amsmath}
\usepackage{amstext}
\usepackage{amssymb,bm}
\usepackage{latexsym}
\usepackage{color}
\usepackage{graphicx}
\usepackage{amsmath}
\usepackage{amsthm}
\usepackage{graphicx}
\usepackage[center]{caption}
\usepackage{pstricks}
\usepackage{caption}
\usepackage{subcaption}
\usepackage{booktabs}
\usepackage{multicol}
\usepackage{lipsum}

\usepackage{enumitem}

\newtheorem{remark}{Remark}
\newtheorem{corollary}{Corollary}
\newtheorem{prop}{Proposition}
\newtheorem{lem}{Lemma}
\newtheorem{definition}{Definition}
\newtheorem{theorem}{Theorem}

\newcommand{\eu}{\mathrm{e}}

\newcommand{\Cov}{\mathrm{Cov}}

\newcommand{\snr}{\mathsf{snr}}

\newcommand{\sign}{\mathsf{sign}}
\newcommand{\mmse}{\mathrm{mmse}}

\newcommand{\E}{\mathbb{E}}
\newcommand{\Trc}{\mathrm{Tr}}
\newcommand{\cov}{\mathbf{Cov}}
\newcommand{\K}{\mathbf{K}}
\newcommand{\X}{\mathbf{X}}
\newcommand{\Z}{\mathbf{Z}}
\newcommand{\Y}{\mathbf{Y}}
\newcommand{\W}{\mathbf{W}}
\newcommand{\V}{\mathbf{V}}
\newcommand{\U}{\mathbf{U}}
\newcommand{\I}{\mathbf{I}}

\newcommand{\supp}{{\mathsf{supp}}}
\newcommand{\gap}{{\mathsf{gap}}}

\newcommand{\mmpe}{\mathrm{mmpe}}
\newcommand{\Err}{\mathsf{Err}}
\newcommand{\p}{{\mathsf{p}}}

\title{On the Minimum Mean $\p$-th Error in Gaussian Noise Channels and its Applications} 
\author{
\IEEEauthorblockN{Alex Dytso, Ronit Bustin, Daniela Tuninetti,  Natasha Devroye, H.Vincent Poor, Shlomo Shamai (Shitz) \\
\thanks{Alex Dytso, Daniela Tuninetti and Natasha Devroye are with the department of Electrical and Computer Engineering, University of Illinois at Chicago, IL, Chicago 60607, USA (e-mail: odytso2, danielat, devroye @ uic.edu).  
Ronit Bustin is with the department of Electrical Engineering - Systems, Tel Aviv University, Tel Aviv, Israel (email:ronitbustin@post.tau.ac.il). 
H.Vincent Poor is with the department of Electrical  Engineering, Princeton University, NJ, Princeton 08544, USA (email:poor@princeton.edu). 
S. Shamai (Shitz) is with the Department of Electrical Engineering,
Technion-Israel Institute of Technology, Technion City, Haifa 32000,
Israel (e-mail:  sshlomo@ee.technion.ac.il).  The work of Alex Dytso, Daniela Tuninetti  and Natasha Devroye was partially funded by NSF under award 1422511. The work of Ronit Bustin was supported in part by the women postdoctoral scholarship of Israel's Council for Higher Education (VATAT) 2014-2015. The work of H. Vincent Poor and Ronit Bustin was partially supported by NSF under awards CCF-1420575 and ECCS-1343210. The work of Shlomo Shamai was supported by the Israel Science Foundation (ISF).   The contents of this article are solely the responsibility of the authors and do not necessarily represent the official views of the funding agencies. This work was presented in part at the 2016 IEEE International Symposium on Information Theory, Barcelona, Spain, and in part at the 2016 IEEE Information Theory Workshop, Cambridge, UK.} 
 }
}

\begin{document}

\maketitle

\begin{abstract}
The problem of estimating an arbitrary random vector from its observation corrupted by additive white Gaussian noise, where the cost function is taken to be the Minimum Mean $\p$-th Error (MMPE), is considered.
The classical Minimum Mean Square Error (MMSE) is a special case of the MMPE. 
Several bounds, properties and applications of the MMPE are derived and discussed.

The optimal MMPE estimator is found for Gaussian and binary input distributions. Properties of the MMPE as a function of the input distribution, Signal-to-Noise-Ratio (SNR) and order $\p$ are derived. In particular, it is shown that the MMPE is a continuous function of $\p$ and SNR.  These results are possible in view of  interpolation and change of measure bounds on the MMPE.

The `Single-Crossing-Point Property'  (SCPP) that
bounds the MMSE for all SNR values {\it above} a certain value, at which the MMSE  is known,  together with the I-MMSE relationship is a powerful tool in deriving converse proofs in multi-user information theory. By studying the notion of conditional MMPE, a  unifying proof (i.e.,  for any $\p$) of the SCPP is shown.  A complementary bound to the  SCPP is then shown, which bounds the MMPE  for all SNR values {\it below} a certain value, at which the MMPE is known. 

As a first application of the MMPE, a bound on the conditional differential entropy in terms of the MMPE is provided, which then yields a  generalization of  the Ozarow-Wyner lower bound on the mutual information achieved by a discrete input on a Gaussian noise channel.

As a second application, the MMPE is shown to improve on previous characterizations of the phase transition phenomenon that manifests, in the limit as the length of the capacity achieving code goes to infinity, as a discontinuity of the MMSE as a function of SNR. 

As a final application, the MMPE is used  to show new bounds on the second derivative of mutual information, or the first derivative of the MMSE, that tighten previously known bounds important in characterizing the  bandwidth-power trade-off in the wideband regime. 

\end{abstract}

\section{Introduction}
 \label{sec:intro}

  In the Bayesian setting the  Minimum Mean Square Error (MMSE) of estimating a random variable $X$ from an observation $Y$  is understood as a cost function\footnote{  Another common term used is a risk function.}  with a quadratic loss  function (i.e., $L_2$ norm):
  \begin{subequations}
\begin{align}
\mmse(X \mid Y)&= \E \left[  \Err \left(X, \E[ X \mid Y] \right) \right], \\
\Err \left(X, \E[ X \mid Y] \right)&= |X-\E[X \mid Y] |^2.  \label{eq:def MMSE}
\end{align} 
\end{subequations}

Another commonly used cost function is the $L_1$ norm  with loss function given by the absolute value of the error (i.e., the difference between the variable of interest and its estimate). In general, cost functions with non-quadratic loss functions are not well understood  and have been considered only for special cases, such as  under the assumption of  Gaussian statistics.

The interplay between estimation theoretic and information theoretic measures has been very fruitful; for example the so called I-MMSE relationship \cite{I-MMSE}, that relates the derivative of the mutual information with respect to the Signal-to-Noise-Ratio (SNR) to the MMSE, has found numerous applications through out information theory \cite{ShamaiShannonLecture}. The goal  of this work is to show that the study of estimation problems with non-quadratic loss functions can  also offer new insights into classical information theoretic problems. \emph{The program of this paper is thus to develop the necessary theory  for a class of loss functions, and then apply the developed tools to  information theoretic problems.  } 

\subsection{Past Work} 
The popularity of the MMSE stems from its analytical tractability, which is rooted in the fact that the MMSE is defined through the $L_2$ norm in \eqref{eq:def MMSE}. The $L_2$ norm, in turn,  allows applications of the well understood Hilbert space theory  \cite{kreyszig1989introductory}.   In information theoretic applications the $L_2$ norm  is used, for example, to define an average  input power constraint. The  connection between the  power constraint and the $L_2$ norm leads to a continuous analog of Fano's  inequality that relates the conditional differential entropy and the MMSE \cite[Theorem 8.6.6]{Cover:InfoTheory}.

Recently, in view of the  I-MMSE relationship \cite{I-MMSE}, the MMSE (in an Additive White Gaussian Noise (AWGN) channel) has received considerable attention.  For example, in \cite{VerduSimpleProofEPI} the I-MMSE relationship was used to give a simple alternative proof of the Entropy Power Inequality (EPI) \cite{Shannon:1948}.  Moreover, the so called  `Single-Crossing-Point Property' (SCPP) \cite{GuoMMSEprop,BustinMMSEparallelVectorChannel} that bounds the MMSE for all SNR values {\it above} a certain value at which the MMSE is known,  together with the I-MMSE relationship, offers an alternative,   unifying framework for deriving information  theoretic converses, such as:
\cite{GuoMMSEprop} to provide an alternative proof of the converse for the Gaussian broadcast channel (BC) and show a special case of the EPI;
in~\cite{guo2013interplay} to provide a simple proof for the information combining problem and a converse for the BC with confidential messages;
in~\cite{BustinMMSEparallelVectorChannel}, by using various extensions of the SCPP,
to  prove a special case of the vector EPI, a converse for the capacity region of the parallel degraded BC under per-antenna power constraints and under an input covariance constraint, and a converse for the compound parallel degraded BC under an input covariance constraint; and
in~\cite{BustinMMSEbadCodes} to provide a converse for communication under an MMSE disturbance constraint.

In \cite{NewBoundsOnMMSE} we demonstrated a bound that complements  the SCPP, that  bounds the MMPE  for all SNR values {\it below} a certain value at which the MMSE is known,  and  allows for a finer characterization of  the \emph{phase transition} phenomenon  that manifests as a discontinuity of the MMSE as a function of SNR, as the length of the codeword goes to infinity. This plays an important role in characterizing achievable rates of the capacity achieving codes   \cite{MerhavStatisticalPhysics} and \cite{SNRevolutionOfMMSE}. One of the applications of the tools presented in this work is an improvement on the bound in \cite[Theorem 1]{NewBoundsOnMMSE}. 

Many other properties of the MMSE in relation to the I-MMSE have been studied in \cite{GuoMMSEprop,FunctionalPropMMSE,mmseDim, GuoScoreFunction}. For a comprehensive survey on results,  applications and extensions of the I-MMSE relationship we refer the reader to \cite{guo2013interplay}. 

While the MMSE has received considerable attention and is well understood, non-quadratic cost functions are only understood in  special cases, such as under the assumption of Gaussian statistics. For example, in \cite{sherman1958non} it was shown that under scalar Gaussian statistics, for a large class  of symmetric loss functions  the optimal  linear  MMSE (LMMSE) estimator is also optimal.  The result of \cite{sherman1958non} was  extended in \cite{AsymErrorBrown} to a large class of cost functions that also include asymmetric loss functions.  Other early work in this direction includes also \cite{pugachev1960method}. 

In \cite{LinftyError}, the authors  studied the expected $L_{\infty}$ norm of the error,  
when the input is assumed to be a Gaussian mixture. The authors showed that, as the dimension  of the signal goes to infinity, the optimal LMMSE estimator  minimizes the expected maximum error.

In \cite{hall1987simultaneous} and \cite{hall1991optimal} the  authors studied a class of \emph{even and nondecreasing} and \emph{even and convex}  loss functions and gave a   sufficient condition on the conditional distribution of the input  $X$ given the output $Y$, so that the conditional expectation $\E[X|Y]$ is the optimal estimator.

In \cite{akyol2012conditions},  the  authors studied a scalar additive noise channel and  an $L_p$ cost function and showed  a necessary and sufficient condition on the noise and the input distributions to guarantee that the optimal estimator is linear. Moreover,   if the  source and noise variances are the same, then the optimal estimator is linear if and only if input and the noise distributions are identical.

In \cite{WeinbergerMMalphaE} and \cite{MerhavMMalphaE}  the authors considered the problem of transmitting a modulated signal over a discrete memoryless channel where the performance criterion was taken to be the $L_p$ cost function. 
 To that end, the authors showed tight exponential bounds for very small and very large values of $p$.  

In \cite{saerens2000building} the authors focused on designing an appropriate cost function such that  the output of the trained model approximates  the desired summary statistics, such as  the conditional expectation, the geometric mean or the variance.

In non-Bayesian estimation $L_p$ cost functions have been considered in \cite{BurnashevMMalphaE} and \cite{burnashev1985minimum}, in a context of \emph{minimax} estimation, and the authors gave lower and upper bounds on the exponential behavior of the cost function.  For a non-Bayesian treatment of non-quadratic cost functions we refer the reader to \cite{lehmann2006theory}.

Looking into  non-quadratic cost functions is further motivated by the fact that often the quadratic cost function may not be the correct measure of signal fidelity for certain applications. This is especially true in image processing where  error metrics, more sensitive to structural changes of the input signal, better capture human perceptions of quality. We refer the reader to \cite{WangMSEloveorleave} for a survey of recent results in this direction.

\noindent
\subsection{Paper Outline and Main Contributions} 

In this work we are interested in studying a cost function, termed the Minimum Mean $p$-th Error (MMPE)\footnote{  The abbreviation MMPE has been used before in \cite[Chapter 8]{guo2013interplay} for the Minimum Mean Poisson Error. }, the scalar version of which is given by 
\begin{subequations}
\begin{align}
\mmpe(X \mid Y;\p)&= \inf_{f} \E \left[  \Err^{\p} \left(X, f(Y) \right) \right], \\
\Err \left(X,f(Y) \right)&= |X-f(Y) |.  
\end{align} 
\label{eq:scalar MMPE}
\end{subequations}
where the infimum is over all  estimators $f(Y)$.

Our contributions  are as follows:
\begin{enumerate}
\item In Section~\ref{sec:costFUnctDef} we formally define the  vector version of the MMPE in \eqref{eq:scalar MMPE} and introduce related definitions.
 \item In Section~\ref{sec:MMPE: properties of estimator} we study properties of the optimal MMPE estimator and show:
\begin{itemize}
\item  In Section~\ref{sec:MMPE:existanceofOPt}, Proposition~\ref{prop:existence of optimal estimator}  shows that the MPPE optimal estimator 
 indeed exists; 
\item  In Section~\ref{sec:orthogonality}, Proposition~\ref{prop:orthogonality like property}  derives an  \emph{orthogonality-like principle} that  serves as a  necessary and sufficient condition for an estimator to be MMPE optimal; 
\item Section~\ref{sec:examplesOfOptimalEstimators}   gives examples of optimal MMPE estimators. In particular, in Proposition~\ref{prop:GaussianMMPE} we find the  MMPE  for Gaussian random vectors, and  in Proposition~\ref{prop: estimator for two point} for discrete binary random variables; and
\item In Section~\ref{sec:basicpropOfEstim},  Proposition~\ref{prop:opt est} shows some basic properties of the optimal MMPE estimator in terms of input distribution, such as, \emph{linearity, stability, degradedness}, etc. Moreover, via an example it is shown that in general the MMPE optimal estimator is biased on average (i.e., the first moment of the error (bias) is not zero). However,  it is shown that the $\p$-th order estimator is unbiased on average in sense that the $\p-1$-th moment of the error is zero.

\end{itemize} 

\item In Section~\ref{sec:prop of MMPE} we study properties of the MMPE as a function of order $\p$, SNR and  the input distribution that will be useful in a number of applications: 
\begin{itemize}
\item In Section~\ref{sec: basic prop}, Proposition~\ref{prop:Scaling} shows that the MMPE is invariant under translations of the input random vector and derives basic scaling properties;
\item  In Section~\ref{sec:estimatonEquivalence},  Proposition~\ref{prop:equivalenceOfnoiseEst} shows that, as far as estimation error over the channel $\Y=\sqrt{\snr}\X+\Z$ is concerned the estimation of the input $\X$ is equivalent to the estimation of the noise $\Z$; and
\item 
In Section~\ref{sec: change of Measure}, Proposition~\ref{prop:change of measure}  gives a `change of measure' result that allows one to take the expectation in the definition of the MMPE with respect to an output at a different SNR. 
 \end{itemize}

\item In Section~\ref{sec:bound on the MMPE} we discuss basic bounds on the MMPE and show: \begin{itemize}
\item In Section~\ref{sec:basic Bounds},  Proposition~\ref{prop:higher moments bound 1}  develops basic ordering bounds between MMPE's of different orders and bounds equivalent to that of the LMMSE bound; 
\item In Section~\ref{sec:GaussianBOunds}, Proposition~\ref{prop:GaussianHardesToEstimate} shows that,  under an appropriate moment constraint on the input distribution, the  Gaussian input is asymptotically the `hardest' to estimate; 
\item In Section~\ref{sec:continuity},  Proposition~\ref{prop:logconvex}  derives interpolation bounds for the MMPE. One of the consequences of such bounds is  Proposition~\ref{prop:continuity}, which shows that the MMPE is a continuous function of order $\p$; and 
\item In Section~\ref{sec:bounds on MMPE discrete}, Proposition~\ref{prop:bound on discrete inputs} derives bounds on the MMPE with discrete vector inputs.  This in turn leads to a result in Proposition~\ref{prop:phasetrans} that shows that MMPE, similarly to the MMSE, can exhibit phase transitions (i.e., discontinuities as function of the SNR as dimension of the input goes to infinity). 
\end{itemize}

\item In Section~\ref{sec:conditional MMPE} we define the conditional MMPE  and show:
\begin{itemize}
\item  Proposition~\ref{prop:cond reduces} shows that conditioning reduces the MMPE; and
\item  Proposition~\ref{prop:max ration comb} shows that the MMPE estimation of $\X$ from  two  AWGN observations is equivalent to estimating $\X$ from a single observation with a higher SNR. This implies that the MMPE is a decreasing function of SNR.
\end{itemize} 

\item In Section~\ref{sec:SCPP} we show  applications of the developed tools:\begin{itemize}
\item In Proposition~\ref{prop: SCP generalized}, by using the tools developed for the conditional MMPE,  a simple proof of the SCPP for the MMSE is given, and extended to the MMPE;  
\item In Proposition~\ref{prop:MMSE at low snr with MMPE} we use the change of measure result in Proposition~\ref{prop:change of measure} to show a bound that complements the SCPP bound, that it is bounds the MMPE for all SNR values \emph{below} a certain  SNR value at which the MMPE is known; and
\item In Proposition~\ref{prop:continuitySNR}, by using change of measure result in Proposition~\ref{prop:MMSE at low snr with MMPE} and continuity of the MMPE in $\p$ from Proposition~\ref{prop:continuity}, we show that for any finite dimensional input the MMPE is a continuos function of SNR.
 \end{itemize}
 
 \item In Section~\ref{sec: applications} we apply the developed bounds and generalize or improve some  well known information theoretic MMSE bounds:
 \begin{itemize}
 \item In Section~\ref{sec:boundDiffEntropy}, Theorem~\ref{prop:generalization of continuous Fano's inequality} gives a general inequality that bounds the conditional differential entropy via the MMPE of which the continuous analog of Fano's from \cite[Theorem 8.6.6]{Cover:InfoTheory} is  a special case;
  \item In Section~\ref{sec:genOWbound}, Theorem~\ref{prop:OWimproved}  generalizes the Ozarow-Wyner bound \cite{PAMozarow} on the
mutual information achieved by a discrete input on an AWGN channel, to  vector discrete inputs  and yields  the sharpest known version of this bound. Moreover, in Theorem~\ref{prop: Gap in Ozarow-Wyner at n infinity}  we show how the bound behaves as the dimension of the input goes to infinity;
  \item In Section~\ref{prop:phaseTrans:bounds}, Theorem~\ref{prop: bound through MMPE}   improves  on the  previous characterizations of the width the phase transition
region of finite-length code  of length $n$ given by $O(\frac{1}{n})$ in \cite{NewBoundsOnMMSE}   to $O(\frac{1}{\sqrt{n}})$. This in turn  also improves the converse result on the communications under disturbance constrained problem studied in \cite{NewBoundsOnMMSE}; and
  \item  In Section~\ref{sec:BoundsOnDerivative}, Proposition~\ref{prop:bounds on derivative of MMSE via MMPE} we show
how the MMPE can be used to provide new lower and upper
bounds on the derivative of the MMSE. 
 \end{itemize}

\end{enumerate}

\subsection{Notation} 

Throughout the paper we adopt the following notational conventions: 
\begin{itemize}
\item
Deterministic scalar and vector quantities are denoted by lower case and bold lower case letters, respectively. Matrices are denoted by bold upper case letters;
\item Random variables  and vectors are denoted by upper case and bold upper case  letters, respectively, where r.v. is short for either random variable or random vector, which should  be clear from the context;
\item
If $A$ is a r.v. we denote the support of its distribution by $\supp(A)$;

\item
The symbol $|\cdot|$ may denote different things:  $ | {\bf A}|$ is the determinant of the matrix ${\bf A}$,
$|\mathcal{A}|$ is the cardinality of the set $\mathcal{A}$, 
$|X|$ is the cardinality of $\supp(X)$  , or
$|x|$ is the absolute value of the real-valued $x$;

\item $\E[ \cdot ]$ denotes the expectation operator;
\item  We denote the covariance of r.v. $\X$ by $ {\bf K}_{\X}$; 
\item
$\X \sim \mathcal{N}({\bf m },{\bf K}_{\X})$ denotes the density of a real-valued Gaussian r.v. $\X$ with mean vector ${ \bf m}$ and covariance matrix ${\bf K}_{\X}$;
\item The identity matrix is denoted by $\I$;
\item Reflection of the matrix $ \bf{A}$ along its main diagonal, or the transpose operation, is denoted by ${\bf A}^T$;
\item The trace operation  on the matrix $ \bf{A}$ is denoted by $\Trc( \bf{A})$; 
\item The Order notation ${ \bf A} \succeq {\bf B}$ implies that ${\bf A}-{\bf B}$ is a positive semidefinite matrix;
\item
$\log(\cdot)$ denotes logarithms in base $2$;
\item
$[n_1:n_2]$ is the set of integers from $n_1$ to $n_2 \geq n_1$;

\item
For $x\in\mathbb{R}$ we let
$\left \lfloor x \right\rfloor$  denote the largest integer not greater than $x$;

\item
For $x\in\mathbb{R}$ we let $[x]^{+} :=\max(x,0)$ and $\log^{+}(x) :=[\log(x)]^+$;

\item
Let $f(x),g(x)$ be two real-valued functions.
We use the Landau notation %
$f(x)=O(g(x))$ to mean that for {\it some}  $c>0$ there exists an $x_0$ such that $f(x)\leq c \, g(x)$ for all $x \geq x_0$,
and $f(x)=o(g(x))$ to mean that for {\it every} $c>0$ there exists an $x_0$ such that $f(x)<   c g(x)$ for all $x \geq x_0$;

\item We denote the conditional r.v.  $\X | \Y={\bf y} \sim p_{\X|\Y}(\cdot |{\bf y})$ as $\X_{\bf y}$;

\item We  denote the upper incomplete gamma function and the gamma function by
\begin{subequations}
\begin{align}
\Gamma \left( x; a \right)&:=  \int_{a}^\infty t^{x-1} e^{-t} dt,  \ x \in \mathbb{R}, a \in \mathbb{R}^{+},  \\
\Gamma \left( x\right)&:= \Gamma \left( x;0\right).
\end{align}
The generalized Q-function is denoted by
\begin{align}
\bar{Q}(x;a):=\frac{\Gamma \left( x; a \right)}{\Gamma \left( x\right)}.
\end{align}
\end{subequations}
In particular,  the generalized $Q$-function can be related to the standard $Q$-function, by  using the relationship $ Q( \sqrt{2} x ) = \frac{1}{2 \sqrt{\pi}} \Gamma \left( \frac{1}{2}; x^2 \right)$ and $\Gamma \left(\frac{1}{2}\right)=\sqrt{\pi}$, as  $\bar{Q} \left( \frac{1}{2};a^2 \right)= 2 Q(  \sqrt{2 } a )$; and
\item We define the volume of the region $S$ embedded in  $\mathbb{R}^n$ as
\begin{align}
{\rm Vol}(S):= \int_S  1 \  dx_1 dx_2 \cdot \cdot \cdot dx_n. 
\end{align} 
In particular, the volume of the $n$-dimensional ball  $B(r)$ of radius $r$ center at origin is given by \[ {\rm Vol}(B(r))= \frac{\pi^{\frac{n}{2}} r^n}{\Gamma \left( \frac{n}{2}+1\right)}.\]
\end{itemize}

\section{Cost Function Definition} 
\label{sec:costFUnctDef}
Motivated by the study of cost functions with non-quadratic error we define the following norm.

\begin{definition}
For the r.v. ${\bf U} \in \mathbb{R}^n$ and $\p>0$
\begin{align}
\| {\bf U} \|_{\p} :=  \left(\frac{1}{n } \E\left[ \Trc^{ \frac{\p}{2} } \left( {\bf U } { \bf  U}^T\right) \right] \right)^{\frac{1}{\p}}. \label{eq: defintion of the norm}
\end{align}
\end{definition} 
For $\p \ge 1$ the function  in \eqref{eq: defintion of the norm} defines a norm and obeys the triangle inequality 
\begin{align}
\| {\bf U+V} \|_{\p}  \le \| {\bf U} \|_{\p} +\| {\bf V} \|_{\p}, \label{eq: Minkowski} 
\end{align} 
as shown in Appendix~\ref{app:triangleInequality Proof of lem: Minkowski}. 
Therefore, throughout  the paper 
we define the $L_{\p}$ space, for $\p \ge 1$, as the space of r.v. on a  fixed probability space $( {\bf \Omega}, \sigma({\bf \Omega}), \mathbb{P})$ such that the norm  in~\eqref{eq: defintion of the norm} is  finite.  However, many of our results will hold for $0 \le \p <1$,  for which \eqref{eq: defintion of the norm}  is not a norm. 

In particular, for  $\Z \sim \mathcal{N}\left( 0,\I \right)$ the norm in~\eqref{eq: defintion of the norm}  is given by 
\begin{align}
n \| \Z \|_p^{\p}=\E \left[\Trc^{\frac{\p}{2}}(\Z \Z^T) \right]=\E \left[ \left(\sum_{i=1}^n Z_i^2 \right)^{\frac{\p}{2}} \right]= 2^{\frac{\p}{2}}  \frac{\Gamma \left(\frac{n}{2}+{\frac{\p}{2}} \right)}{\Gamma \left( \frac{n}{2} \right)}, \text{ for }  n \in \mathbb{N}, p \ge 0, \label{eq:moments of Gaussian}
\end{align}
and for $\V$ uniform over the $n$ dimensional  ball of radius $r$ the norm in~\eqref{eq: defintion of the norm} is given by
\begin{align}
n \| \V \|_{\p}^{\p} %
= \frac{1}{{\rm Vol}(B(r))}  \frac{\pi^{\frac{n}{2}} }{\Gamma \left( \frac{n}{2}\right)}  \int_{0}^r  \rho^{\p} \rho^{n-1} d \rho 
 =  \frac{n }{  \p+n } r^{\p}, \text{ for }  n \in \mathbb{N}, \p \ge 0. \label{eq: mommentsUniform}
\end{align} 
Note that for  $n=1$ we have that $\| U \|_{\p}^{\p} =\E \left[ |U|^{\p}\right] $ and therefore from now on we will refer to $\| \U \|_{\p}^{\p}$ as $\p$-th moment of  $\U$. Naturally, for $n>1$, there are many other ways for defining the moments, see for example \cite{lutwak2013affine}. However, in view of the information theoretic problems we are interested in, such for example from previous work \cite{NewBoundsOnMMSE}, the definition in~\eqref{eq: defintion of the norm}  arises naturally. 

\begin{definition}
\label{def:MMPE}
We define the minimum mean $\p$-th error (MMPE) %
of estimating $\X$ from $\Y$   as
\begin{subequations}
\begin{align}
\mmpe(\X | \Y;\p)
&:=\inf_{f}  n^{-1} \E \left [  \Err^{\frac{\p}{2}}\left(\X, f(\Y)\right) \right],\\
&:= \inf_{f} \| \X - f(\Y) \|^{\p}_{\p}  
\end{align}
 \label{eq: definition of MMPE} 
\end{subequations} 
and  where  the minimization is 
over all possible 
Borel measurable  functions $f(\Y)$. Whenever the optimal MMPE estimator  exists,  we shall denote it by  $f_{\p}(\X|\Y)$.\footnote{The restriction to measurable functions, in Definition~\ref{def:MMPE}, is necessary.  See \cite{wise1985note} for surprising complications that can arise without this assumption.}

\end{definition} 

We shall denote 
\begin{align}
\mmpe(\X | \Y; \p)=\mmpe(\X,\snr,\p), 
\end{align}
if  $\Y$ and $\X$ are related as
\begin{align}
\Y=\sqrt{\snr} \ \X+\Z,
\label{eq:channel model}
\end{align}
\noindent
where  $\Z,\X,\Y \in \mathbb{R}^n$,   $ \Z~\sim \mathcal{N}({\bf 0}, \I)$ is independent of $\X$,  and  $\snr \ge 0$ is the SNR. When it will be necessary to emphasize the SNR at the output $\Y$, we will denote it with $\Y_{\snr}$.  Since the distribution of the noise is fixed 
$\mmpe(\X| \Y;\p)$ is completely determined by distribution of $\X$ and $\snr$ and there is no ambiguity in using the notation $\mmpe(\X,\snr,\p)$.  Applications to the Gaussian noise channel will be the main focus of this paper.

For $\p=2$, the MMPE reduces to the MMSE,  that is, $\mmpe(\X| \Y ;2)=\mmse(\X|\Y)$ and  $f_2(\X|\Y)=\E[\X|\Y]$. 
Note that there are other ways of defining the loss function in~\eqref{eq: definition of MMPE}; our definition   in~\eqref{eq: definition of MMPE} is  motivated by:
\begin{itemize}
\item 
For $X \in \mathbb{R}^1$ the error in \eqref{eq: definition of MMPE} reduces to a natural expression with loss function  given by $\Err^{\frac{\p}{2}} (X,f(Y))= |X-f(Y)|^{\p}$;
\item  
The definition in \eqref{eq: definition of MMPE} naturally appears in applications of H{\"o}lder's or Jensen's inequalities to $\mmse(X|Y)$; %
and 
\item The norm in \eqref{eq: defintion of the norm} used in the definition of \eqref{eq: definition of MMPE}  can be related to information theoretic quantities, such as differential entropy and Reyni entropy, via the vector moment entropy inequality from \cite{MomentsEntropyInequality}.
\end{itemize}
 
We shall also look at %
the $\p$-th error achieved by the suboptimal (unless $\p=2$) estimator  $\E[\X|\Y]$,  that is,
\begin{align}
n^{-1}\E \left[ \Err^{\frac{\p}{2}} \left(\X,\E[\X| \Y] \right) \right]=\| \X-\E[\X| \Y] \|_{\p}^{\p},
\label{eq: higher order moments}
\end{align}
which represents higher order moments of the MMSE  loss function and serves (see below) as  an upper bound on~\eqref{eq: definition of MMPE}.


\section{Properties of the Optimal MMPE Estimator} 
\label{sec:MMPE: properties of estimator}

\subsection{Existence of Optimal Estimator} 
\label{sec:MMPE:existanceofOPt}
It is important to point out that $\| \X-\E[\X|\Y]  \|_{\p} $  in general is not equal to the MMPE, as  $\E[\X|\Y]$ might not be the optimal estimator under the $\p$-th norm. The first result of this section shows that for AWGN channel the optimal estimator $f_{\p}(\X|\Y={\bf y})$ indeed exists.

\begin{prop}
\label{prop:existence of optimal estimator}
For $\mmpe(\X,\snr,\p)$, $\p > 0$, $\snr>0$ the optimal estimator 
is given by the following point-wise relationship 
\begin{align}
f_{\p}( \X | \Y={\bf y}) = \arg \min_{ {\bf v} \in \mathbb{R}^n } \E \left[ \Err^{\frac{\p}{2}} (\X, {\bf v}) |\Y={\bf y}\right].\label{eq: optimal estimator}
\end{align}
Moreover, if  $ \| \X \|_{\p}< \infty$ then \eqref{eq: optimal estimator}  is also valid for $\snr=0^{+}$. 
\end{prop}
\begin{proof}
See Appendix~\ref{app:existence of optimal estimator}.  
\end{proof}

A  result similar to that in Proposition~\ref{prop:existence of optimal estimator} can be found in \cite[Theorem 4.1.1]{lehmann2006theory} where it has been shown that  for a given $\X$  an estimator $f_{\p}( \X | \Y)$ is  optimal  provided that  the minimum on the right hand side of \eqref{eq: optimal estimator} exists. In contrast to  \cite[Theorem 4.1.1]{lehmann2006theory}, Proposition~\ref{prop:existence of optimal estimator} shows that  the minimum in \eqref{eq: optimal estimator} exists for any $\X$, and $f_{\p}( \X | \Y)$ is the MMPE optimal estimator for any $\X$. 

Proposition~\ref{prop:existence of optimal estimator} immediately implies the following corollary on the interchange of the expectation and infimum which will be used in many of the following proofs. 

\begin{corollary} 
\label{lem: exchange of inf and expect}
For $\p > 0$ and $\snr>0$
\begin{align}
\mmpe(\X,\snr,\p) = \inf_{f} \frac{1}{n}\E \left[\Err^{\frac{\p}{2}}(\X,f(\Y))\right]= \frac{1}{n} \E \left[ \inf_{f} \E \left[ \Err^{\frac{\p}{2}}(\X,f(\Y))| \Y\right]\right].
\end{align}
\end{corollary}
\begin{proof}
In the proof of Proposition~\ref{prop:existence of optimal estimator} it is shown that \begin{align*}
\E \left[ \inf_{f} \E \left[ \Err^{\frac{\p}{2}}(\X,f(\Y))| \Y\right]\right]=\E \left[ \Err^{\frac{\p}{2}}(\X,f_p(\X|\Y))\right], 
\end{align*}
for $f_\p(\X|\Y)$ in \eqref{eq: optimal estimator}.
Therefore, we have the following chain of inequalities
\begin{align}
\E \left[ \Err^{\frac{\p}{2}}(\X,f_\p(\X|\Y))\right]&=\E \left[\inf_{f} \E \left[ \Err^{\frac{\p}{2}}(\X,f(\Y))| \Y\right]\right] \le \inf_{f}  \E \left[ \Err^{\frac{\p}{2}}(\X,f(\Y)) \right] \notag\\
&\le  \E \left[ \Err^{\frac{\p}{2}}(\X,f_\p(\X|\Y))\right].
\end{align}
This concludes the proof. 
\end{proof}

 \subsection{Orthogonality-like Property}
 \label{sec:orthogonality}
The MMPE for $\p \neq 2$ differs from MMSE
in a number of aspects. The main difference is that the norm defined in~\eqref{eq: defintion of the norm}  is not a Hilbert space norm in general (unless $\p=2$); as a result, there is no notion of inner product or orthogonality, and $f_\p(\X|\Y)$, unlike $\E[\X| \Y],$ can no longer be thought of as an orthogonal projection. Therefore, the orthogonality principle---an important tool in the analysis of the MMSE---is no longer available when studying the MMPE for $\p \neq 2$.  However, an orthogonality-like property can indeed be shown for the MMPE.  
\begin{prop} 
\label{prop:orthogonality like property} \emph{(Necessary and Sufficient Condition for the Optimality of $f_\p( \X |\Y)$).}
For any $\X$, any $\snr>0$, $\p \ge 1$, $f_\p( \X |\Y)$  is an optimal estimator if and only if 
\begin{subequations}
\begin{align}
\E \left[  \Err^{{\frac{\p-2}{2}}} \left( \X, f_p( \X |\Y) \right) \cdot \left( \X-f_\p( \X |\Y)\right)^T \cdot g(\Y)\right]=0, \label{eq: orthoglike}
\end{align}
for any  deterministic function $g: \mathbb{R}^n \to \mathbb{R}^n $, that is,
\begin{align}
  \E \left[  \left({\bf W}^T {\bf W}\right)^{\frac{\p-2}{2}} \cdot {\bf W}^T  \cdot g(\Y)\right]=0,
\end{align} 
\end{subequations}
 where ${\bf W}= \X-f_\p( \X |\Y) $.
 Moreover, for $0 \le \p < 1$  the condition in \eqref{eq: orthoglike} is necessary for optimality. 
\end{prop} 
\begin{proof}
See Appendix~\ref{app:prop:orthogonality like property}. 
\end{proof}
Note that   Proposition~\ref{prop:orthogonality like property} for $n=1$ and $\p \in \mathbb{R}^{+}$ reduces to 
\begin{align}
\E [ |X- f_\p(X|Y)|^{\p-2}(X- f_\p(X|Y)) g(Y)]=0,
\end{align}
which for ${\frac{\p}{2}} \in \mathbb{N}$ further reduces to 
\begin{align}
\E [ (X- f_\p(X|Y))^{\p-2}g(Y)]=0.
\end{align} 
Moreover,  for $\p=2$ Proposition~\ref{prop:orthogonality like property} reduces to the familiar orthogonality principle 
\begin{align}
\E \left[  \left( \X-f_2( \X |\Y)\right)^T \cdot g(\Y)\right]= \E \left[  \left( \X-\E[ \X |\Y]\right)^T \cdot g(\Y)\right]=0. 
\end{align} 

\begin{remark}
In the analysis of the MMSE the orthogonality property is an important tool, used for example  to show that $\E[\X| \Y]$ is the unique minimizer. The argument goes as follows: assume that we have another optimal estimator $f(\Y) \neq \E[\X| \Y]$, then by orthogonality principle, 
\begin{align}
0=\E[ (\X-\E[\X| \Y])^T g(\Y)]- \E[ (\X-f(\Y))^T g(\Y)]= \E[ (f(\Y)-E[\X|\Y])^T g(\Y)].
\end{align} 
By choosing $g(\Y)=  (f(\Y)-E[\X|\Y])$ we see that $\E[ (f(\Y)-E[\X|\Y])^T  (f(\Y)-E[\X|\Y]) ] > 0$, arriving at a contradiction.  This implies that $\E[\X| \Y]$ is the unique estimator up to a set of measure zero.

In \cite[Lemma 1]{akyol2012conditions}, by replicating the above argument and by assuming that ${\frac{\p}{2}} \in \mathbb{N}$ and $n=1$, it was shown that the optimal MMPE estimator is \emph{unique}. However,  since the proof relies heavily on the assumption that  ${\frac{\p}{2}} \in \mathbb{N}$ and $n=1$,   this argument cannot be extended in a straightforward way to $\p\in \mathbb{R}^{+}$ or $n>1$. 

 However, uniqueness of the MMPE optimal estimator can be shown for $\p>1$ (i.e., strictly convex loss functions) by using Proposition~\ref{prop:existence of optimal estimator} in  conjunction with \cite[Corollary 4.1.4]{lehmann2006theory}.
\end{remark} 

\subsection{Examples of Optimal MMPE Estimators} 
\label{sec:examplesOfOptimalEstimators}
 In general we do not have a closed form solution for the MMPE optimal estimator in~\eqref{eq: optimal estimator}. 
Interestingly, the optimal estimator for  Gaussian inputs can be found and   is linear for all $\p \ge 1$. Note that similar results have been  demonstrated in \cite{sherman1958non} and \cite{akyol2012conditions} for scalar Gaussian inputs.    Next we extend this result to vector inputs and give two alternative proofs of the linearity of the optimal MMPE estimator for Gaussian  inputs, via Proposition~\ref{prop:existence of optimal estimator} and via Proposition~\ref{prop:orthogonality like property}.

\begin{prop}
\label{prop:GaussianMMPE}
For input $\X_G \sim \mathcal{N}(0,\I)$  and $\p \ge 1$
\begin{subequations}
\begin{align}
\mmpe(\X_G,\snr,\p)=  \frac{ \| \Z \|_{\p}^{\p}}{(1+\snr)^\frac{\p}{2}}, 
\end{align}
with optimal estimator given by 
\begin{align}
f_p(\X_G|\Y={\bf y})=\frac{\sqrt{\snr} \ {\bf y}}{1+\snr}.
\end{align}
\end{subequations}
\end{prop}
\begin{proof}
The proof follows by observing that  ${\bf W}=\X_G-\frac{\sqrt{\snr}  \Y }{1+\snr}$ has a Gaussian distribution  and  is independent of $\Y$. So, for any  two functions $f(\cdot)$ and $g(\cdot)$  we have
\begin{align}
\E[f( {\bf W}) g(\Y)]=0. \label{eq: Gaussian Independent for any function}
\end{align}
Therefore, by using \eqref{eq: Gaussian Independent for any function}  for estimator $f_p(\X_G|\Y={\bf y})=\frac{\sqrt{\snr} }{1+\snr} \ {\bf y}$ the necessary and sufficient conditions in Proposition~\ref{prop:orthogonality like property} hold and thus the linear estimator must be an optimal one. Finally observe that 
\begin{align}
 \left\| \X_G-\frac{\sqrt{\snr}}{1+\snr}  \Y \right\|_{\p}^{\p} = \|\hat{\Z}\|_{\p}^{\p} = \frac{ \| \Z \|_{\p}^{\p}}{(1+\snr)^{\frac{\p}{2}}},
 \end{align}
 where  we have used   $\hat{\Z}=\X_G-\frac{\sqrt{\snr} }{1+\snr} \Y \sim \mathcal{N} \left(0, \frac{1}{1+\snr} \I\right)$.

For a proof that uses only Proposition~\ref{prop:existence of optimal estimator} see Appendix~\ref{app:proof: prop:GaussianMMPE}. 
\end{proof} 

The optimal MMPE estimator is in general a function of $p$ as shown next. 

\begin{prop}
\label{prop: estimator for two point} 
For $X=\{ x_1,x_2\}$ with $\mathbb{P}[X=x_1]=1-\mathbb{P}[X=x_2]=q \in (0,1)$ and for $\p \ge 1$ we have that
\begin{subequations}
\begin{align}
f_\p(X|Y=y)&= \frac{x_1 \cdot q^{{\frac{1}{\p-1}}} \cdot\eu^{-\frac{(y-\sqrt{\snr}x_1)^2}{2(\p-1)}} +x_2 \cdot(1-q)^{{\frac{1}{\p-1}}}\cdot \eu^{-\frac{(y-\sqrt{\snr}x_2)^2}{2(\p-1)}}}{  q^{{\frac{1}{\p-1}}}\cdot\eu^{-\frac{(y-\sqrt{\snr}x_1)^2}{2(\p-1)}} + (1-q)^{{\frac{1}{\p-1}}}\cdot\eu^{-\frac{(y-\sqrt{\snr}x_2)^2}{2(\p-1)}}}. \label{eq: optimal est for binary}
\end{align} 
In particular,  for  $\p=1$, we have that \begin{align}
f_{\p=1}(X|Y=y)=\left \{  \begin{array}{ll} x_1,  & a \ge 1 \\
x_2, & a <1
\end{array}\right. ,
\end{align} 
\end{subequations}
where $a=\frac{q \ \eu^{- \frac{(y-\sqrt{\snr}x_1)^2}{2}} }{(1-q) \ \eu^{- \frac{(y-\sqrt{\snr}x_2)^2}{2}}}=\frac{q}{q-1} \eu^{- \snr \frac{(x_1-x_2)(x_1+x_2)}{2}+\sqrt{\snr} y (x_1-x_2)}$.
\end{prop}
\begin{proof}
See Appendix~\ref{app: proof of prop: estimator for two point}.
\end{proof}

Proposition~\ref{prop: estimator for two point} will be useful in demonstrating several examples and   counter examples in the following sections. Note that for the practically relevant case of BPSK modulation, or $x_1=-x_2=1, q=\frac{1}{2}$  the optimal estimator in \eqref{eq: optimal est for binary} reduces to 
\begin{subequations}
\begin{align}
 f_\p(X|Y=y)= \tanh \left(\frac{y \sqrt{\snr}}{\p-1} \right), \label{eq:optimal est for BPSK}
 \end{align}
 which   for $\p=1$   is the hard decision decoder
\begin{align}
f_{\p=1}(X|Y=y)= \left \{  \begin{array}{ll} -1,  & y  \le 1 \\
+ 1, & y > 1
\end{array}\right. .
\end{align}
\end{subequations}

By Proposition~\ref{prop: estimator for two point} we can show that the orthogonality principle  only holds for $\p=2$ (when MMPE corresponds to MMSE)  as shown in Fig.~\ref{fig:Orthonognality}, where we plot $h(p) := \E[(X-f_\p(X|Y)) Y]$ vs. $\p$ for BPSK input and observe it is zero only for $\p=2$.

%
%

    \begin{figure*}
          \begin{subfigure}[t]{0.5\textwidth}
           \centering
                \includegraphics[width=7cm]{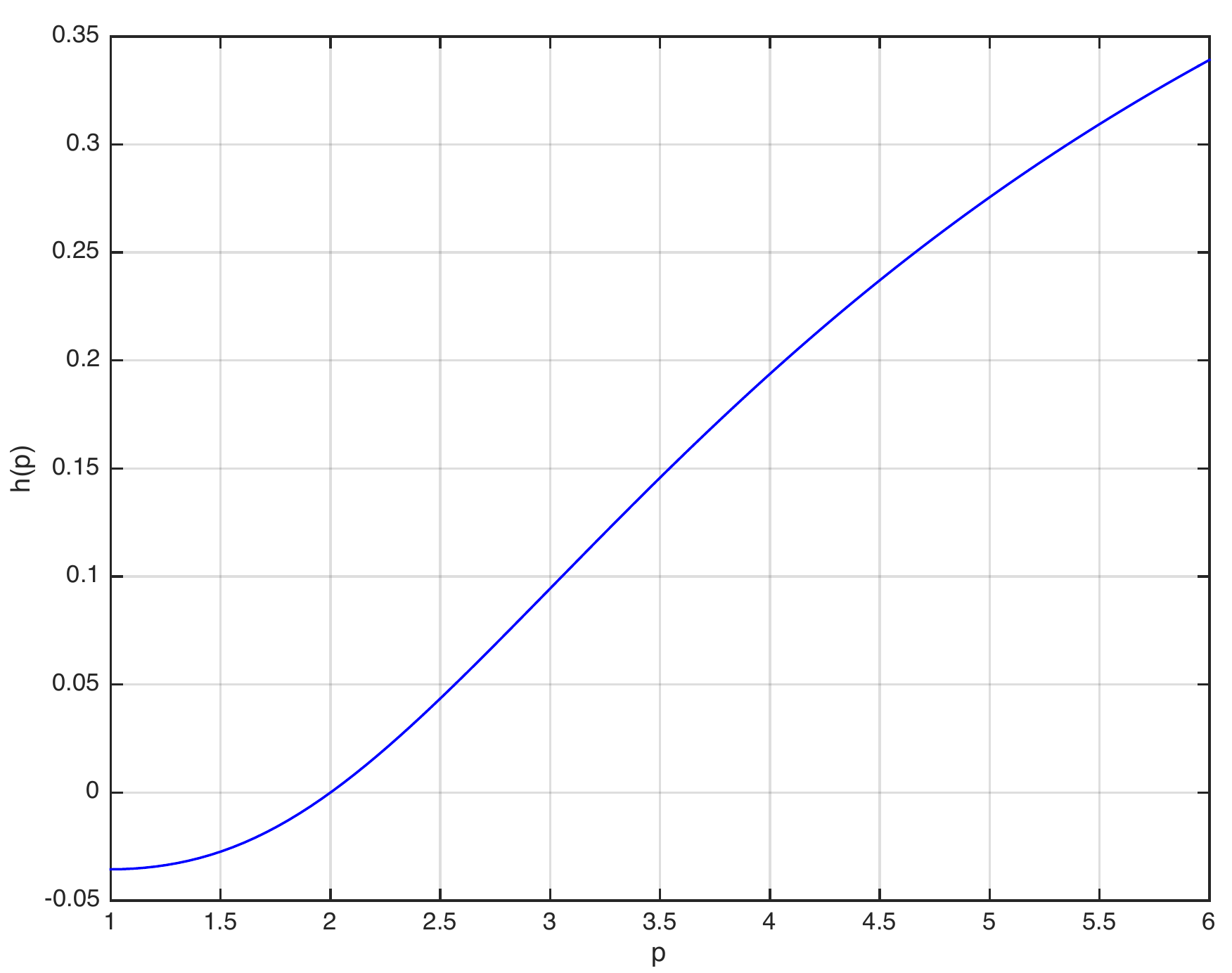}
                \caption{Plot of   $h(\p) := \E[(X-f_\p(X|Y)) Y]$ vs. $\p$, for $X\in \{ \pm 1 \}, \mathbb{P}[X=1]=\frac{1}{2}$ and $\snr=1$ and $f_\p(X|Y)$  given in \eqref{eq:optimal est for BPSK}. }
                \label{fig:Orthonognality}
        \end{subfigure}%
           \begin{subfigure}[t]{0.5\textwidth}
           \centering
                \includegraphics[width=7cm]{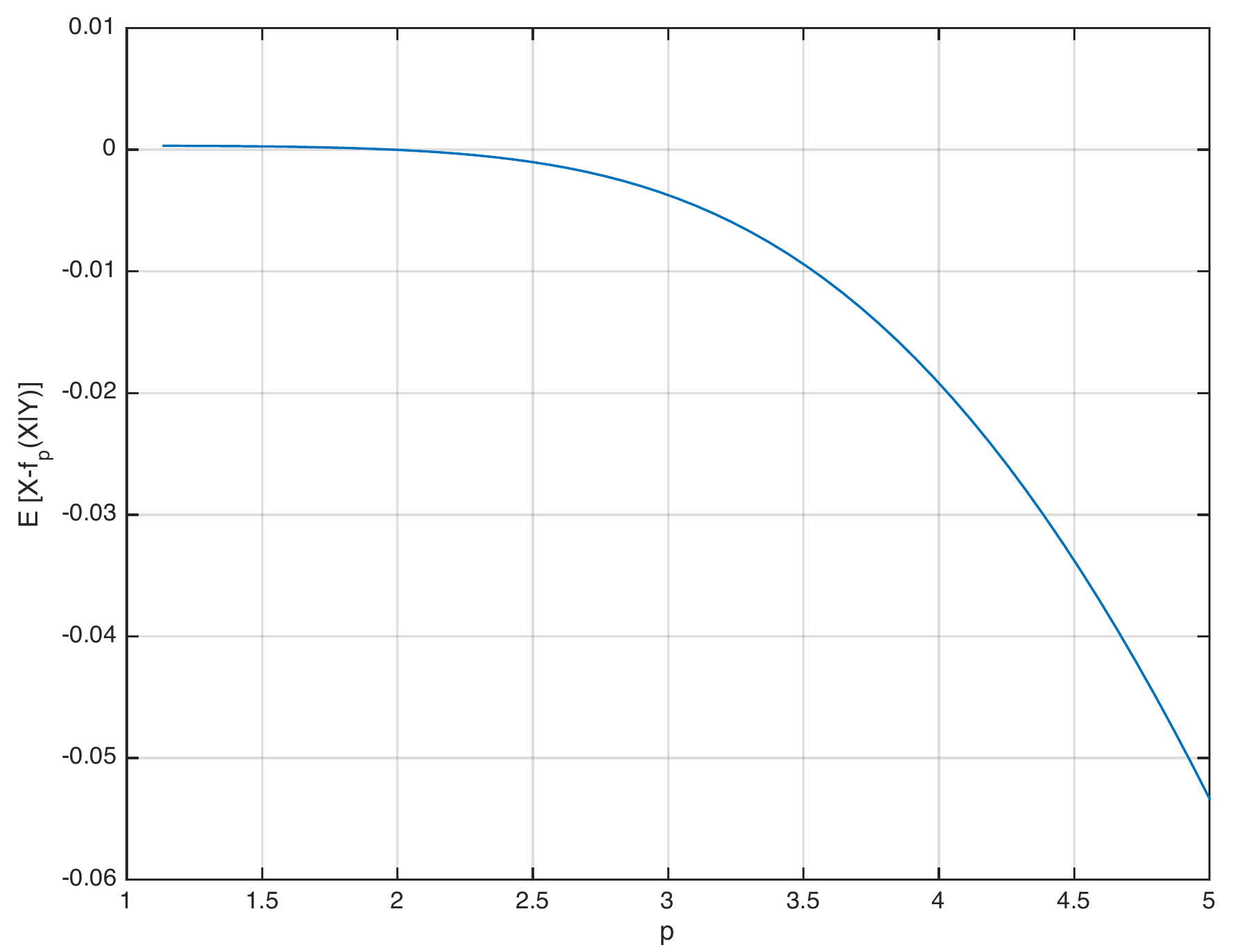}
                \caption{Plot of $\E[(X- f_\p(X|Y)]$ vs. $\p$, for $X\in\{-3,1\}$, $ \mathbb{P}[X=-3]=0.01$ and $\snr=1$ with $f_\p(X|Y)$ given in \eqref{eq: optimal est for binary}.}
                \label{fig:BiasExample}
        \end{subfigure}
       \caption{Counter examples for the orthogonality principle and the bias of the MMPE optimal estimator.}
        \label{fig: Propes of optimal estimator}
\end{figure*}

\subsection{Basic Properties of the Optimal MMPE Estimator} 
\label{sec:basicpropOfEstim}

Interestingly many of the  known properties of $f_2(\X|\Y)=\E[\X| \Y]$ for MMSE  are still exhibited by 
 $f_\p(\X| \Y)$ for any $\p>0$.
\begin{prop} 
\label{prop:opt est}
For any $\p>0$ the  optimal MMPE estimator  has the following properties:
\begin{enumerate} 
\item if $ 0 \le X \in \mathbb{R}^1$ then $ 0\le f_\p(X|Y)$, 
\item (Linearity) $f_\p(a \X+b| \Y)=af_\p( \X| \Y)+b$ for $a,b \in \mathbb{R}$,
\item (Stability) $f_\p( g(\Y)  | \Y) =g(\Y)$ for any deterministic function $g( \cdot)$,
\item (Idempotent) $f_\p( f_\p(\X| \Y) | \Y) =f_\p(\X| \Y)$,
\item (Degradedness) $f_\p \left(\X |\Y_{\snr_0}, \Y_{\snr} \right)= f_\p \left( \X |\Y_{\snr_0} \right)$ for a Markov chain $\X  \to \Y _{\snr_0} \to \Y_{\snr}$,
\item (Orthogonality-like Principle) See Proposition~\ref{prop:orthogonality like property}.
\end{enumerate}
\end{prop}
\begin{proof}
See Appendix~\ref{app:prop:opt est}.
\end{proof}

\begin{remark}(Average Bias of the MMPE optimal Estimator)
\label{sec:Bias}
An estimator $f_\p(\X|\Y)$ is said to be  unbiased on average if
$\E[ \X-f_\p(\X|\Y) ]={\bf 0}.$
In general $f_\p(\X|\Y)$ is unbiased on average only for $\p=2$, since
\begin{align}
\E[ f_{\p=2}(\X|\Y) ]=\E[ \E[\X| \Y] ]=\E[\X].
\end{align}
%
Fig.~\ref{fig:BiasExample} shows that in general the optimal MMPE estimator is biased  on average;
it plots $\E[X- f_\p(X|Y)]$ vs. $\p$ for $X\in\{-3,1\}: \mathbb{P}[X=-3]=0.01$ and $\snr=1$, with $f_\p(X|Y)$ as in Proposition~\ref{prop: estimator for two point}. This comes as no surprise as it is very  common in Bayesian estimation that the optimal estimator is biased  \cite{kay1993fundamentals}.

However, the optimal MMPE estimator 
is unbiased in the  sense that the $(\p-1)$-th moment of the bias is zero. This can be seen from the orthogonality like property in Proposition~\ref{prop:orthogonality like property} by taking $g(\Y)$ to be the vector of all one's 
\begin{align}
\E \left[  \Err^{\frac{\p-2}{2}} \left( \X, f_\p( \X |\Y) \right) \cdot \left( \X-f_\p( \X |\Y)\right)^T \right]=0.
\end{align}

\end{remark}

%
%
%
%
%
%
%

\section{ Properties of the MMPE}
\label{sec:prop of MMPE}
In this section we explore properties of the MMPE as a function of SNR and of the input distribution. 

\subsection{Basic Properties}
\label{sec: basic prop}
The next two properties of the MMPE  directly follow from the properties of $f_\p(\X |\Y)$  in Proposition~\ref{prop:opt est}. 

\begin{prop}
\label{prop:Scaling} For  any $\p>0$
\begin{subequations}
\begin{align}
 \mmpe(\X+a,\snr,\p)&=\mmpe(\X,\snr,\p), \\
 \mmpe(a \X,\snr,\p)&= a^{\p} \mmpe(\X, a^2 \snr,\p).
 \end{align}
 \end{subequations}
\end{prop}

Proposition~\ref{prop:Scaling} implies that the MMPE, like the MMSE, is invariant under translations, and that  scaling  the input results in   scaling the SNR and the error.

\subsection{ Estimation of the Input is Equivalent to Estimation of the Noise} 
\label{sec:estimatonEquivalence}
The following lemma is commonly applied  in the analysis of the  MMSE. 

\begin{lem} 
\label{lem: equivalence of errors}
For $\X,\Z,\Y$ given in \eqref{eq:channel model}
\begin{subequations}
\begin{align} 
\sqrt{\snr}(\X-\E[\X|\Y])=-(\Z-\E[\Z|\Y]).
\end{align}
Moreover,
\begin{align}
\sqrt{\snr} \ \| \X-  \E[\X|\Y] \|_\p =   \| \Z-  \E[\Z|\Y] \|_\p.
\end{align}
\end{subequations}
\end{lem}
Lemma~\ref{lem: equivalence of errors} states that estimating the noise is equivalent to estimating the input signal  if one  uses the conditional expectation as an estimator.

Next we show that an equivalent  statement holds for the MMPE.

\begin{prop} 
\label{prop:equivalenceOfnoiseEst}
For any $\X$, $\p>0$ and $\snr>0$, we have 
\begin{subequations}
\begin{align}
\sqrt{\snr} \ \mmpe^{\frac{1}{\p}}(\X| \Y; \p) =   \mmpe^{\frac{1}{\p}}(\Z|\Y; \p). \label{eq: equaivalence: noise-input}
\end{align} 
Moreover, 
\begin{align}
\sqrt{\snr}\  (\X-f_\p(\X|\Y) )= -(\Z-f_\p( \Z | \Y)).  \label{eq: equivalence estimators}
\end{align}
\end{subequations}
\end{prop}
\begin{proof}
From the definition of the MMPE in \eqref{eq: definition of MMPE}
\begin{align}
\mmpe^{\frac{1}{\p}}(\X,\snr, \p)&=  \inf_{f({\bf y})} \| \X-f(\Y) \|_\p \label{eq: infimimu first step} \\
&=  \inf_{f({\bf y})} \left\| \frac{1}{\sqrt{\snr}}(\Y-\Z)-f(\Y) \right\|_\p \notag \\
&=\frac{1}{\sqrt{\snr}}  \inf_{f({\bf y})} \left\| \Z-\left( \sqrt{\snr}f(\Y)-\Y\right) \right\|_\p \notag\\
&=\frac{1}{\sqrt{\snr}}  \inf_{g({\bf y}): \ g({\bf y})= {\bf y}-\sqrt{\snr} f({\bf y})} \left\| \Z-g(\Y) \right\|_\p  \label{eq: infimum transformation}\\
&= \frac{1}{\sqrt{\snr}}  \mmpe^{\frac{1}{\p}}(\Z| \Y; \p).
\end{align}
This  shows  the equality in  \eqref{eq: equaivalence: noise-input}. 
Moreover, since $f_\p(\X|\Y)$ exists and the infimum in \eqref{eq: infimimu first step} is attainable by Proposition~\ref{prop:existence of optimal estimator}, so is the infimum in \eqref{eq: infimum transformation}. Therefore, from  \eqref{eq: infimum transformation}  we have that  $f_\p(\Z | \Y)$ exists  and is given by 
\begin{align}
f_\p(\Z | \Y)= \Y-\sqrt{\snr} f_\p(\X|\Y)=  \sqrt{\snr}\X +\Z-\sqrt{\snr} f_\p(\X|\Y),
\end{align}
which leads to \eqref{eq: equivalence estimators}.
This concludes the proof. 
\end{proof} 

\subsection{Change of Measure}
\label{sec: change of Measure}
The next result  enables us to change the expectation from $\Y_{\snr}$ to $\Y_{\snr_0}$ in \eqref{eq: definition of MMPE} whenever $\snr \le \snr_0$. This is particularly useful when we know the MMPE, or the structure of the optimal MMPE estimator, at one SNR value but not at another smaller SNR value. 
\begin{prop} 
\label{prop:change of measure}
For any $\X$,   $\snr \in (0, \snr_0]$ and  $\p > 0$, we have 
\begin{align}
\mmpe(\X,\snr,\p) = \inf_{f}\frac{1}{n}\E \left[ \Err^\frac{\p}{2}(\X,f(\Y_{\snr_0}))  \sqrt{\frac{\snr}{\snr_0 }} \eu^{\frac{\snr_0-\snr}{2\snr_0}  \sum_{i=1}^n Z_i^2 }\right].  \label{eq:change Measure equat} 
\end{align}
\end{prop}
\begin{proof}
See Appendix~\ref{app: prop:change of measure}.
\end{proof}
One must be careful when evaluating Proposition~\ref{prop:change of measure}. 
For example, since we have that
\begin{align*}
\lim_{\snr \to 0^{+}}  \sqrt{\frac{\snr}{\snr_0 }} \eu^{\frac{\snr_0-\snr}{2\snr_0} Z^2 } =0,
\end{align*}
at first glance it appears that the expectation on the right of \eqref{eq:change Measure equat} is zero while $\mmpe(X, 0,\p)$ is not, thus 
violating the equality. 
However, a more careful examination shows that when  $\snr \to 0$ the limit and expectation  in \eqref{eq:change Measure equat} cannot be exchanged; indeed we have that
\begin{align*}
& \lim_{\snr \to 0^{+}} \E\left[\sqrt{\frac{\snr}{\snr_0 }} \eu^{\frac{\snr_0-\snr}{2\snr_0} Z^2 } \right]
 = \lim_{\snr \to 0^{+}} \sqrt{\frac{\snr}{\snr_0 }}  \E\left[\eu^{\frac{\snr_0-\snr}{2\snr_0} Z^2 } \right] \notag\\
&  = \lim_{\snr \to 0^{+}}  \sqrt{\frac{\snr}{\snr_0 }}   \frac{1}{\sqrt{1-\frac{\snr_0-\snr}{\snr_0}}}=1,
\end{align*}
where in the last equality we used the 
moment generating function of  the Cauchy r.v. $Z^2$.
As an example, Proposition~\ref{prop:change of measure} for $X\sim \mathcal{N}(0,1)$ 
with the optimal linear estimator from Proposition~\ref{prop:GaussianMMPE}, i.e., $f(y)=ay$ for some $a$,
evaluates to 
\begin{align*}
&\E \left[  \Err(X,f(Y_{\snr_0}))  \sqrt{\frac{\snr}{\snr_0 }} \eu^{\frac{\snr_0-\snr}{2\snr_0}  Z^2 }\right] \notag\\
&\stackrel{a)}{=}(1-\sqrt{\snr_0} a)^2 \sqrt{\frac{\snr}{\snr_0 }} \E[X^2] \E \left[ \eu^{\frac{\snr_0-\snr}{2\snr_0}  Z^2 }\right]  +a^2 \sqrt{\frac{\snr}{\snr_0 }} \E \left[Z^2\eu^{\frac{\snr_0-\snr}{2\snr_0}  Z^2 } \right] \\ 
&\stackrel{b)}{=}\frac{1}{1+\snr}, 
\end{align*}
where the equalities follow from:
a) linearity of expectation and the fact that $Z$ and $X$ are independent; and
b) since $\E \left[ \eu^{\frac{\snr_0-\snr}{2\snr_0}  Z^2 }\right] =\sqrt{\frac{\snr_0}{\snr }}$ and  $\E \left[Z^2\eu^{\frac{\snr_0-\snr}{2\snr_0}  Z^2 } \right]=\left(\frac{\snr_0}{\snr}\right)^{3/2}$ and by choosing $a=\frac{\snr}{\sqrt{\snr_0}(1+\snr)}$ in order to minimize the expression in a).\footnote{Note that this optimal $a$ is evident from the specific change of measure that we have used. Instead of having the estimator according to Proposition~\ref{prop:GaussianMMPE} as $\frac{ \sqrt{\snr}}{1 + \snr}$ we get it with the normalization  by $\sqrt{\frac{\snr}{\snr_0}}$. }

\section{Bounds on The MMPE }
\label{sec:bound on the MMPE}
In this section we develop bounds on the MMPE, many of which generalize well known MMSE bounds. However, we also show bounds that are unique to the MMPE and emphasize the usefulness  of the MMPE.
\subsection{Extension of Basic MMSE Bounds}
\label{sec:basic Bounds}
An important upper bound on the  MMSE often used in practice is the LMMSE.
\begin{prop}
\label{prop:LMMSE bounds} 
(\emph{LMMSE}  \cite{GuoMMSEprop}.) For any input $\X$ and $\snr>0$
\begin{subequations}
\begin{align}
\mmse(\X,\snr) \le \frac{1}{\snr}. \label{eq: LMMSE decorrolator} 
\end{align}
If  $ \| \X \|_2^2  = \sigma^2< \infty$, then for any $\snr \ge 0$ 
\begin{align}
\mmse(\X,\snr) \le \frac{\sigma^2}{1+\sigma^2 \snr},  \label{eq: LMMSE power}
\end{align}
\end{subequations}
where equality in~\eqref{eq: LMMSE power}  is achieved  iff $\X \sim \mathcal{N}(0, \sigma^2 \I)$.
\end{prop}

The next bound generalizes  Proposition~\ref{prop:LMMSE bounds}   to higher order errors.
\begin{prop}
\label{prop:higher moments bound 1}
For $\snr \ge 0$,  $0< {\rm q} \le \p$, and input $\X$
\begin{subequations}
\begin{align}
n^{\frac{\p}{ \rm q}-1} \mmpe^\frac{\p}{\rm q}(\X,\snr,{\rm q}) \le \mmpe(\X,\snr,\p) \le \| \X-\E[\X| \Y] \|_{\p}^{\p},   \label{eq: chain bound on MMPE}
\end{align}
and where 
\begin{align} 
\| \X-\E[\X| \Y] \|_{\p}^{\p} &\stackrel{\text{for  $\p \ge 2$}}{\le} 2^{\p} \min   \left( \frac{  \| \Z \|_{\p}^{\p}  }{\snr^{\frac{\p}{2}}} ,  \| \X \|_{\p}^{\p}  \right),  \label{eq: 1/snr bound for higher orders}\\
\| \X-\E[\X| \Y] \|_{\p}^{\p} &\stackrel{\text{for  $1 \le \p \le 2$}}{\le}  \min   \left( \frac{  \left( \|\Z \|_{\p}+n^{\frac{1}{2}-\frac{1}{\p}}\| \Z \|_2 \right)^{\p}  }{\snr^{\frac{\p}{2}}} ,   \left( \|\X \|_p+n^{\frac{1}{2}-\frac{1}{\p}}\| \X \|_2 \right)^{\p}  \right),  \label{eq: 1/snr bound for higher orders: 1/2<p<1}\\
 \mmpe(\X,\snr,\p)&\stackrel{\text{for  $\p \ge 0$}}{\le} \min   \left( \frac{  \| \Z \|_{\p}^{\p}}{\snr^{\frac{\p}{2}}} ,   \| \X \|_{\p}^{\p}  \right), \label{eq: p-th error bounds} 
\end{align}

\end{subequations}
where  $  \| \Z \|_{\p}^{\p}$ is given in \eqref{eq:moments of Gaussian}. 
\end{prop}
\begin{proof}
See Appendix~\ref{app:bounds}.
\end{proof}

It is interesting to point out that in the derivation of the bounds in Proposition~\ref{prop:higher moments bound 1} no assumption is put on the distribution of $\Z$, and thus the bounds hold in great generality. 
 If $\Z$ is composed of independent identically distributed (i.i.d.) Gaussian elements, then the moment $ \| \Z \|_{\p}^{\p}$ in Proposition~\ref{prop:higher moments bound 1} can be tightly approximated in terms of factorials as
\begin{align}
\frac{\Gamma \left( \frac{n}{2}+\frac{\p}{2}\right)}{\Gamma \left( \frac{n}{2}\right)} \le \frac{\Gamma \left( \left \lceil \frac{n}{2}+\frac{\p}{2} \right \rceil \right) }{\Gamma \left( \left \lfloor \frac{n}{2} \right \rfloor \right)}=\frac{\left( \left \lceil \frac{n}{2}+\frac{\p}{2} \right \rceil -1 \right) !}{\left( \left \lfloor \frac{n}{2} \right \rfloor -1\right)!} =O(n^\frac{\p}{2}),
\end{align}
which  is tight for even $n$ and integer $\frac{\p}{2}$.

It is not difficult to check that for $\p=2$ 
Proposition~\ref{prop:higher moments bound 1} reduces to Proposition~\ref{prop:LMMSE bounds}.
The reason that the bounds on $\| \X-\E[\X| \Y] \|_\p$ are only available for $\p \ge 2$, while the bounds on  $\mmpe(\X,\snr,p)$ are available for $\p \ge 0$, is because the proof of the bound in~\eqref{eq: 1/snr bound for higher orders}  uses Jensen's inequality, which requires $\p \ge 2$, while the proof of the bound in~\eqref{eq: p-th error bounds} does not.

\subsection{Gaussian Inputs are the Hardest to Estimate}
\label{sec:GaussianBOunds}
Note that the bounds in  Proposition~\ref{prop:higher moments bound 1} are similar to the bound in \eqref{eq: LMMSE decorrolator} and blow up at $\snr=0^{+}$. Therefore, it is  desirable to have bounds  as in \eqref{eq: LMMSE power}. The next result demonstrates such a bound and shows that Gaussian inputs are asymptotically the hardest to estimate. 
\begin{prop} 
\label{prop:GaussianHardesToEstimate}
For  $\snr \ge 0$, $\p \ge 1$, and  a  random variable $\X$ such that  $\| \X \|_{\p}^{\p} \le \sigma^{\p} \| \Z \|_{\p}^{\p}$, we have
\begin{subequations}
\begin{align}
\mmpe(\X, \snr, \p) &\le \kappa_{\p,\sigma^2 \snr} \cdot   \frac{\sigma^{\p}  \| \Z \|_{\p}^{\p} }{ (1+ \snr \sigma^2)^\frac{\p}{2}},  \label{eq: 1/(1+snr)  bound} 
\end{align}
where 
\begin{align}
   \text{ for $\p=2$: }&  \kappa_{\p, \sigma^2 \snr}^\frac{1}{\p} =1,\\
   \text{  for $\p \neq 2$: }&  1 \le  \kappa_{\p, \sigma^2\snr}^\frac{1}{\p} = \frac{1+\sqrt{\sigma^2 \snr} }{ \sqrt{1+\sigma^2 \snr}}   \le 1+\frac{1}{\sqrt{1+\sigma^2 \snr}}.
\end{align}
\end{subequations} 
Moreover,  a  Gaussian  $\X$ with  per-dimension variance $\sigma^2$ (i.e., $\X \sim \mathcal{N}( {\bf 0}, \sigma^2 \I)$)  asymptotically achieves the  bound in \eqref{eq: 1/(1+snr)  bound}, since  $ \lim_{\snr \to \infty} \kappa_{\p, \sigma^2  \snr}=1$.
\end{prop}
\begin{proof}
See Appendix~\ref{app: GaussianHardesToEstimate}.
\end{proof}

\subsection{Interpolation Bounds and Continuity}
\label{sec:continuity}
One of the key advantages of using the MMPE is that  the MMPE of order $q$ can be tightly predicted based on  the knowledge of the MMPE at lower orders $\p$ and higher orders ${\rm r}$.
At the heart of this analysis is the  interpolation result of  $L_\p$ spaces~\cite{folland2013real}:
given $0 < \p \le {\rm q} \le {\rm r} $ and $\alpha \in (0,1)$ such that $\frac{1}{\rm q}=\frac{\alpha}{\p}+\frac{1-\alpha}{\rm r}$, the ${\rm q}$-th norm  can be bounded as
\begin{align}
\| V \|_{\rm q} \le \| V \|_{\p}^{\alpha}   \| V \|_{\rm r} ^{(1-\alpha)},
 \label{eq: interpollation of the norm}
\end{align}
%
which implies that the norm is log-convex
and thus a continuous function of $\p$~\cite[Theorem 5.1.1]{webster1994convexity}.
Next, we present several interpolation results for the MMPE.
\begin{prop}\label{prop:logconvex}
\emph{(Log-Convexity and Interpolation.)}
For any $0 < \p \le {\rm q} \le {\rm r}   \le \infty$ and  $\alpha  \in (0,1)$    such that 
\begin{subequations}
\begin{align}
\frac{1}{\rm q}=\frac{\alpha}{\p}+\frac{\bar{\alpha}}{\rm r} 
\Longleftrightarrow
\alpha=\frac{{\rm q}^{-1}-{\rm r}^{-1}}{\p^{-1}-{\rm r}^{-1}},  \label{eq: interpolation parameters} 
\end{align}
where $\bar{\alpha}=1-\alpha$,  we have   for any $f(\Y)$
\begin{align}
 \| \X -f(\Y)\|_{\rm q} &\le  \| \X -f(\Y)\|_{\p}^\alpha\  \| \X -f(\Y) \|_{\rm r}^{\bar{\alpha}}. \label{eq:general interpolation} 
 \end{align}
In particular,
   \begin{align}
    \| \X -\E[\X | \Y]\|_{\rm q} &\le  \| \X -\E[\X | \Y] \|_{\p}^\alpha \  \| \X -\E[\X | \Y] \|_{\rm r}^{\bar{\alpha}}. \label{eq: higher moments log-convex} 
    \end{align}
    Moreover,
    \begin{align}
    \mmpe^{\frac{1}{ \rm q}}(\X,\snr,{\rm q}) \le  \inf_f \| \X -f(\Y) \|_{\p}^\alpha  \| \X -f(\Y) \|_{\rm r}^{\bar{\alpha}}  \label{eq: interpolation 3}.
    \end{align}
    In particular,
    \begin{align}
        \mmpe^{\frac{1}{\rm q}}(\X,\snr,{\rm q})  &\le   \| \X -f_r(\X|\Y) \|_{\p}^{\alpha}  \ \mmpe^{\frac{\bar{\alpha}}{\rm r}}(\X,\snr,{\rm r}), \label{eq: interpolation 4}\\
        \mmpe^{\frac{1}{{\rm q}}}(\X,\snr,{\rm q})  &\le  \mmpe^{\frac{\alpha}{\p}}(\X,\snr,\p)      \| \X -f_\p(\X|\Y) \|_{{\rm r}}^{ \bar{\alpha}}.\label{eq: interpolation 5}
    \end{align}
\end{subequations} 
\end{prop}
\begin{proof}
The bound in \eqref{eq:general interpolation} follows by applying \eqref{eq: interpollation of the norm} 
with $V=\Err(\X,f(\Y)) \in \mathbb{R}$. The bounds in \eqref{eq: higher moments log-convex}  follow by choosing $f(\Y)=\E[\X|\Y]$.

The bound in  \eqref{eq: interpolation 3} follows by 
\begin{align}
 \mmpe^{\frac{1}{ \rm q}}(\X,\snr,{\rm q})=\inf_f \| \X-f(\Y) \|_{\rm q} & \le \inf_f \| \X -f(\Y) \|_{\p}^\alpha  \| \X -f(\Y) \|_{\rm r}^{\bar{\alpha}} ,
 \end{align}
where  the last  inequality follows from \eqref{eq: interpollation of the norm} by choosing $V=\Err(\X,f(\Y)) \in \mathbb{R}$. 

Finally, the bounds in \eqref{eq: interpolation 4} and  \eqref{eq: interpolation 5} follow by choosing $f(\Y)$ in \eqref{eq: interpolation 3} equal to $f_{\rm r}(\X|\Y)$ and  $f_{\p}(\X|\Y)$ respectively. This concludes the proof. 
 \end{proof}

From log-convexity we can deduce continuity.
\begin{prop}
\label{prop:continuity}
\emph{(Continuity.)}
For any $\X$ and $\snr>0$, $\mmpe(\X,\snr,\p)$ and $\| \X -\E[\X | \Y]\|_{\p}$ are continuous functions of $\p >0$.
\end{prop}
\begin{proof}
Continuity of $\| \X -\E[\X | \Y]\|_{\p} $ follows from log-convexity in \eqref{eq: higher moments log-convex} while the continuity of MMPE follows from  
\begin{align*}
&  \lim_{{\rm q} \to \p} \left|\mmpe(\X,\snr,\p)-\mmpe(\X,\snr,{\rm q})  \right|  \\
  & \le \lim_{{\rm q} \to \p}
 \max \left(  \| \X -f_{\rm q}(\X| \Y) \|_{\p}^{\p}-\| \X -f_{\rm q}(\X| \Y) \|_{\rm q}^{\rm q},   \| \X -f_{\p}(\X| \Y) \|_{\rm q}^{\rm q}-\| \X -f_\p(\X| \Y) \|_{\p}^{\p} \right) \\
 &=0,
\end{align*}
where the last inequality is due to the continuity of the norm. 
\end{proof}

    An interesting question  is whether  the following interpolation inequality holds:
      \begin{align}
      \mmpe^{\frac{1}{{\rm q}}}(\X,\snr,{\rm q}) \le   \mmpe^{\frac{\alpha}{\p}}(\X,\snr,\p) \  \mmpe^{\frac{\bar{\alpha}}{{\rm r}}}(\X,\snr,{\rm r})  \label{eq: conjecture Interp}
      \end{align}
      instead of \eqref{eq: interpolation 4} and \eqref{eq: interpolation 5}. 
A counter example to the interpolation inequality in \eqref{eq: conjecture Interp} is shown in Fig.~\ref{fig:InterpolationCounterExample} where we take a binary input  $X \in \{\pm 1\}$ equality likely, $\p=2,{\rm r}=8$,  and $\snr=1$  and show:
\begin{itemize}
\item  The MMPE of $X$  of order ${\rm q}$ versus $\alpha \in [0,1]$ where  ${\rm q}$ is computed according to \eqref{eq: interpolation parameters} (blue-solid line); 
\item  The interpolation bound in \eqref{eq: interpolation 4} (purple dashed-dotted line);
\item The interpolation bound in  \eqref{eq: interpolation 5} (yellow solid-dotted line);
\item The interpolation bound in \eqref{eq: interpolation 3} with $f(Y)=f_{ \frac{\p+{\rm r}}{2} }(X|Y)$ (green dashed line); and
\item The right-hand side of the conjectured inequality in \eqref{eq: conjecture Interp} (red-dotted line). \end{itemize} 
This shows that \eqref{eq: conjecture Interp} is not true in general. 
\begin{figure}
        \centering
        \includegraphics[width=7.5cm]{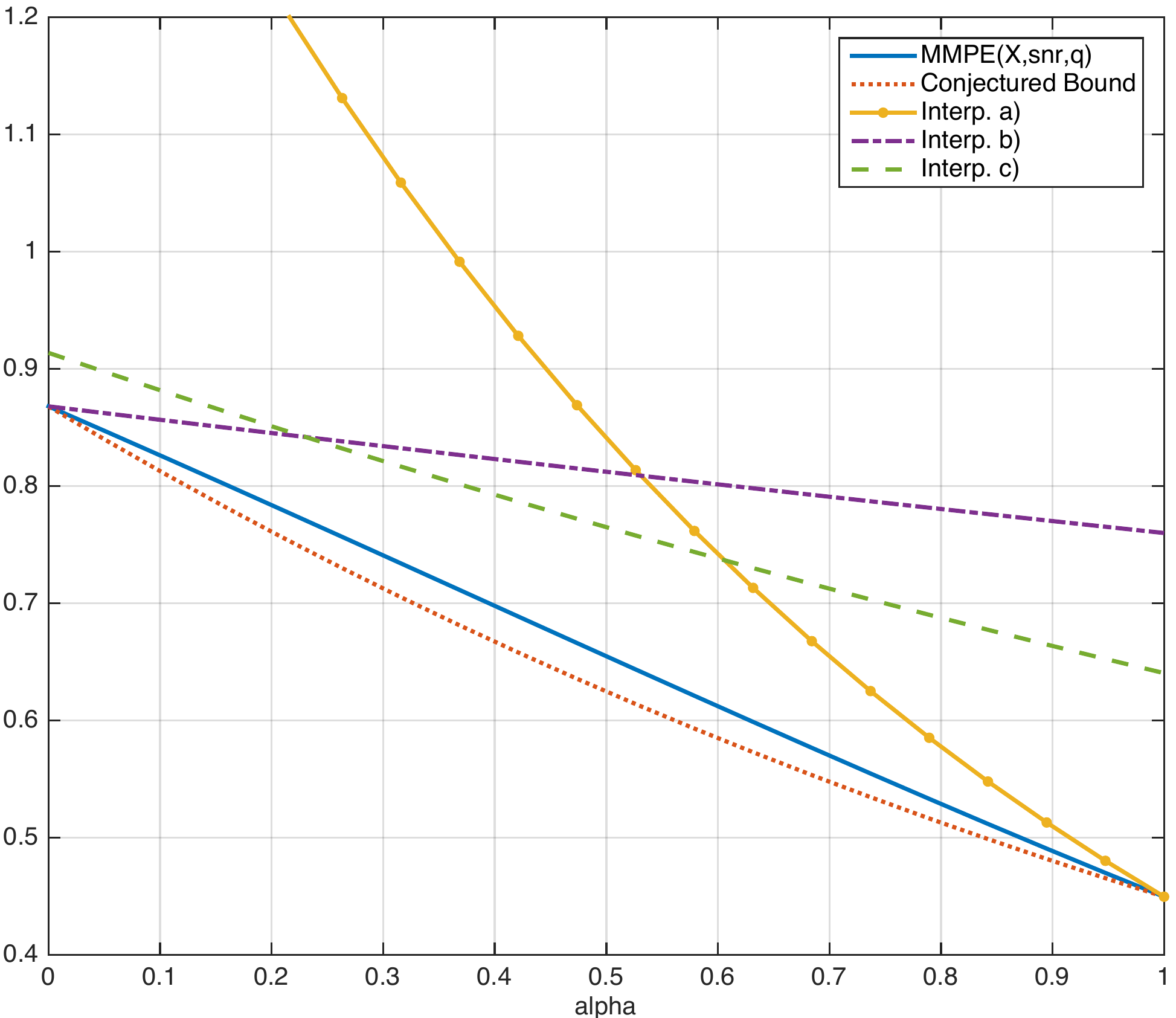}
        \caption{ Interpolation bounds from Proposition~\ref{prop:logconvex} and the conjectured bound in \eqref{eq: conjecture Interp}  versus $\alpha$. Clearly the conjectured bound is below the true MMPE thus \eqref{eq: conjecture Interp} cannot be true.  }
       \label{fig:InterpolationCounterExample}
\end{figure}

\subsection{Bounds on Discrete Inputs}
\label{sec:bounds on MMPE discrete}

So far, by using Proposition~\ref{prop:higher moments bound 1}, we have shown that the  MMPE as a function of $\snr$ decreases  as $O \left(\frac{ 2^{\frac{\p}{2}} n^{\frac{\p}{2}-1}}{\snr^\frac{\p}{2}}\right)$. Next we show that the MMPE can decrease exponentially in $\snr$. Such a behavior has been already observed for the MMSE in \cite{OptimalPowerAllocationOfParChan} and \cite{ mmseDim}. 

\begin{prop}
\label{prop:bound on discrete inputs}
Let $\X_D$ be a discrete r.v. with $|\supp(\X_D)|=N$ and $\mathbb{P}[\X_D= {\bf x}_i ]=p_i$ for ${\bf x}_i \in  \supp(\X_D)$ then
\begin{subequations}
\begin{align}
\mmpe(\X_D,\snr,\p)& \le  d_{\max}^{\p}(\X_D)  \frac{ \sum_{i=1}^N p_i \bar{Q} \left( \frac{n}{2}; \frac{\snr \ d_{ {\bf x}_i}^2(\X_D) }{8}\right)}{ n }  \\
&\le  d_{\max}^{\p}(\X_D) \frac{ \bar{Q} \left( \frac{n}{2}; \frac{\snr \ d^2_{ \min}(\X_D) }{8}\right)}{n },
\end{align} 
where 
\begin{align}
d_{ {\bf x}_i}(\X_D)&= \min_{{\bf x}_j \in \supp(\X_D): i \neq j}  \|  {\bf x}_j-{\bf x}_i \|, \\
d_{ \min}(\X_D)&= \min_{{\bf x}_i \in \supp(\X_D)}  d_{ {\bf x}_i}(\X_D),\\
d_{ \max}(\X_D)&= \max_{{\bf x}_j \in \supp(\X_D)}  \|  {\bf x}_j-{\bf x}_i \|.
\end{align}
\end{subequations}
\end{prop}

\begin{proof}
See Appendix~\ref{app:prop:bound on discrete inputs}.
\end{proof}
The exponential behavior of the MMPE of discrete inputs can be clearly seen for the case  $n=1$ as follows. By using $\bar{Q} \left( \frac{1}{2};a^2 \right)= 2 Q(  \sqrt{2 } a )$  we have that 
\begin{align}
\mmpe(X_D,\snr,p) \le  2 d_{\max}^{\p}(X_D)Q\left(   \sqrt{ \frac{ \snr d^2_{ \min}(X_D)   }{4} }\right) \le 2 d_{\max}^{\p}(X_D) \eu^{- \frac{   \snr     d^2_{ \min}(X_D) }{8}},
\end{align}
where the last inequality follows from  the Chernoff's bound $Q(x) \le \eu^{-\frac{x^2}{2}}$. 

Having developed bounds on the MMPE of discrete inputs, we are now in the position to demonstrate a phase transition phenomenon, that is, we show that as $n \to \infty$ the MMPE becomes a discontinuous function of the SNR.

\begin{prop} \label{prop:phasetrans}
For $\X_D  \in\{  \pm {\bf 1}\}$ (i.e., vector of all one's or all minus ones) equally likely  then 
\begin{align} 
\lim_{n \to \infty} \mmpe(\X_D,\snr,\p)  \le   \left \{\begin{array}{ll}  \lim_{n \to \infty} 4^\frac{\p}{2}  n^{\frac{\p}{2}-1} , &  \snr  < 1   \\
0, & \snr > 1  
\end{array}  \right.  . \label{eq:phasetrans}
\end{align}
\end{prop}
\begin{proof}
For $\X_D  \in \{  \pm {\bf 1}\}$ we have that 
\begin{align}
d_{\max}(\X_D)&=  d_{\min}(\X_D) =\sqrt{4 n}.
\end{align}

With the following well known limits \cite{gautschi1998incomplete}:
\begin{subequations}
\begin{align}
&\lim_{x \to \infty} x^p \bar{Q}\left(x;(1+\epsilon)x \right)=\lim_{x \to \infty} x^p \frac{ \Gamma \left(x;(1+\epsilon)x \right)}{\Gamma \left(x \right)}=0, \\
&\lim_{x \to \infty}  \bar{Q}\left(x;(1-\epsilon)x \right)=\lim_{x \to \infty} \frac{ \Gamma \left(x;(1-\epsilon)x \right)}{\Gamma \left(x \right)}=1,
\end{align} 
\label{eq:limits gamma function}
\end{subequations}
for any $\epsilon>0$ and $p \ge 0$, by using Proposition~\ref{prop:bound on discrete inputs} we have that  
\begin{align}
 \mmpe(\X_D,\snr,p) \le 4^\frac{\p}{2} n^{\frac{\p}{2}}  \frac{  \bar{Q}\left( \frac{n}{2}; \frac{\snr \ d^2_{ \min}(\X_D) }{8}\right)}{n } \le  4^\frac{\p}{2} n^{\frac{\p}{2}}  \frac{  \bar{Q} \left( \frac{n}{2}; \frac{n}{2} \snr\right)}{n } ,
 \end{align}
 and, in light of the limit in \eqref{eq:limits gamma function}, 
 we have that the bound in \eqref{eq:phasetrans} holds.
This concludes the proof. 

\end{proof}

\section{Conditional MMPE} 
\label{sec:conditional MMPE}
We define the conditional MMPE as follows.
\begin{definition} For any  $\X$ and  $\U$,  the conditional MMPE of $\X$ given $\U$ is defined as 
\begin{align}
\mmpe(\X,\snr,\p | \U) := \| \X- f_\p(\X| \Y_{\snr},\U ) \|_{\p}^{\p}.  \label{eq:def cond MMPE}
\end{align}
\end{definition}
The conditional MMPE in \eqref{eq:def cond MMPE} reflects the fact  that the optimal estimator has been given additional information in the form of $\U$. 
Note that when $ \Z$ is 
independent of $(\X,\U)$ we can write the conditional MMPE  for $\X_{\bf u} \sim P_{ \X | \U }(\cdot | {\bf u} )$ as
\begin{align}
\mmpe(\X,\snr,\p | \U) = \int  \mmpe(\X_{\bf u},\snr,\p ) \ dP_\U({\bf u}).
\end{align}

Since giving extra information does not increase the estimation  error,  we have the following result.
\begin{prop} 
\label{prop:cond reduces} (Conditioning Reduces the MMPE.)
For every $\snr \ge 0$, and  random variable $\X$, we have
\begin{align}
\mmpe(\X,\snr,\p ) \ge \mmpe(\X,\snr,\p | \U).
\end{align}
\end{prop}

Finally,  the following Proposition generalizes \cite[Proposition 3.4]{guo2013interplay} and states that the MMPE estimation of $\X$ from two observations is equivalent to estimating $\X$ from a single observation with a higher SNR.
\begin{prop}
\label{prop:max ration comb}
For every $\X$ and $\p  \ge 0$, let $\U=\sqrt{\Delta} \cdot \X+\Z_{\Delta} $ where $\Z_{\Delta} \sim \mathcal{N}(0, \I)$ and where $(\X, \Z, \Z_{\Delta})$ are mutually independent. Then 
\begin{align}
\mmpe(\X,\snr_0,\p | \U ) = \mmpe(\X,\snr_0+\Delta,\p).
\end{align}
\end{prop}
\begin{proof}
For two independent observations 
$\Y_{\snr_0} =\sqrt{\snr_0} \X+\Z$ and  $ \Y_{\Delta} =\sqrt{\Delta} \X+\Z_\Delta$
where $\Z_\Delta$ and $\Z$ are independent, by using maximal ratio combining, we have that 
\begin{align*}
\Y_\snr &=\frac{\sqrt{\Delta}}{ \sqrt{\snr_0 +\Delta}} \Y_{\Delta} +\frac{\sqrt{\snr_0}}{ \sqrt{\snr_0 +\Delta}} \Y_{\snr_0}\\
& =  \sqrt{\snr_0+\Delta} \X +\W,
\end{align*}
where $\W \sim \mathcal{N}(0,\I)$. Next by using the same argument as in \cite[Proposition 3.4]{guo2013interplay}, we have that the conditional probabilities are
\begin{align}
p_{ \X | \Y_{\snr_0},\Y_{\Delta} }({\bf x}| {\bf y}_{\snr_0},  {\bf y}_{\Delta}) = p_{\X| \Y } ({\bf x} | {\bf y}_{\snr} )
\end{align}
for  ${\bf y}_{\snr}= \frac{\sqrt{\Delta}}{ \sqrt{\snr_0 +\Delta}} {\bf y}_{\Delta} +\frac{\sqrt{\snr_0}}{ \sqrt{\snr_0 +\Delta}} {\bf y}_{\snr_0}$. The equivalence of the posterior probabilities implies that the estimation of $\X$ from $\Y_{\snr}$ is as good as the estimation of $\X$ from $(\Y_{\snr_0},\Y_{\Delta})$.  This concludes the proof. 
\end{proof}

Propositions~\ref{prop:max ration comb} and Proposition~\ref{prop:cond reduces} imply that, for fixed $\X$ and $p$\begin{align}
\mmpe(\X,\snr,\p) \ge  \mmpe(\X,\snr,\p | \sqrt{\Delta}\X+\Z')  =\mmpe(\X,\snr+\Delta,\p),
\end{align}
and we have the following:
\begin{corollary}
$\mmpe(\X,\snr,\p)$ is a  \emph{non-increasing} function of $\snr$. 
\end{corollary}

\section{Advance Bounds: SCPP Bound  and Its Complement}
\label{sec:SCPP}
The SCPP is a powerful too that can be used to show the advantage of Gaussian inputs over arbitrary inputs in certain channels with Gaussian noise. In conjunction with the I-MMSE relationship, the SCPP provides simple and insightful converse proofs to the capacity of multi-user AWGN channels.  The original proof of  the SCPP in \cite{GuoMMSEprop} and \cite{BustinMMSEparallelVectorChannel} relied on  bounding the MMSE. Next we give a simpler proof of the SCPP that does not require knowledge of the derivative the MMSE and can easily be extended to the MMPE of any order $\p$. 
 
 First observe that, in light of the bound in \eqref{eq: p-th error bounds}, for any $\snr>0$ we can always find  a $\beta \ge 0$ such that 
 \begin{align}
 \mmpe^{\frac{2}{\p}}(\X,\snr,\p)=\frac{ \beta \| \Z \|_{\p}^{2}  }{ 1+\beta \snr}, \text{ since } \lim_{\beta \to \infty}\frac{ \beta \| \Z \|_{\p}^{2}  }{ 1+\beta \snr}=  \frac{ \| \Z \|_{\p}^{2}  }{  \snr} .
 \end{align}
 
 Next we generalize the SCPP bound to the MMPE.
\begin{prop}
\label{prop: SCP generalized}
Let $\mmpe^{\frac{2}{\p}}(\X,\snr_0,\p)=\frac{ \beta \| \Z \|_{\p}^{2}  }{ 1+\beta \snr_0}$ for some $\beta \ge 0$. Then
\begin{subequations}
\begin{align}
\mmpe^{\frac{2}{\p}}(\X,\snr,p) \le  c_{\p} \cdot \frac{ \beta \| \Z \|_{\p}^{2}  }{ 1+\beta \ \snr},  \text{ for } \snr \ge \snr_0, \label{eq: SCP generalized} 
\end{align}
where 
\begin{align}
c_{\p}= 
\begin{cases}
2 & \p \ge 2 \\
1 & \p=2
\end{cases}.
\end{align} 
\end{subequations}
\end{prop}
\begin{proof}
Let  $\snr=\snr_0+\Delta$ for $\Delta\geq 0,$ and let 
$\Y_{\Delta} =\sqrt{\Delta} \X+\Z_\Delta.$ 
Then
\begin{align*}
\Y_\snr &=\frac{\sqrt{\Delta}}{ \sqrt{\snr_0 +\Delta}} \Y_{\Delta} +\frac{\sqrt{\snr_0}}{ \sqrt{\snr_0 +\Delta}} \Y_{\snr_0}\\
& =  \sqrt{\snr_0+\Delta} \X +\W,
\end{align*}
where $\W \sim \mathcal{N}(0,\I)$. Next, let
\begin{align}
m:=  \mmpe^{\frac{2}{\p}}(\X,\snr_0,p)=\| \X -f_p (\X| \Y_{\snr_0}) \|_{\p}^2, \label{eq: m definition}
\end{align}
and define a suboptimal estimator  given $(\Y_\Delta, \Y_{\snr_0})$ as 
\begin{align}
\hat{\X}= \frac{(1-\gamma)}{\sqrt{\Delta}} \Y_\Delta +\gamma f_p(\X| \Y_{\snr_0}), \label{eq:sub-optimalEstimator}
\end{align}
for some $\gamma \in \mathbb{R}$ to be determined later. Then 
\begin{align*}
\X-\hat{\X}= \gamma( \X-f_p(\X| \Y_{\snr_0}) )-\frac{(1-\gamma)}{\sqrt{\Delta}}\Z_{\Delta},
\end{align*}
and
\begin{align}
\mmpe^{\frac{1}{\p}}(\X,\snr,p) & = \| \X- f_{\p}(\X| \Y_{\snr} ) \|_{\p} & \notag\\
& \stackrel{a)}{=} \| \X- f_{\p}(\X| \Y_\Delta, \Y_{\snr_0} ) \|_{\p}  \notag\\
& \stackrel{b)}{\le}  \| \X-\hat{\X}\|_{\p} = \left\| \gamma( \X-f_{\p}(\X| \Y_{\snr_0}) )-\frac{(1-\gamma)}{\sqrt{\Delta}}\Z_{\Delta} \right\|_{\p} \notag\\
& \stackrel{c)}{=}  \frac{ \left\| \| \Z \|_{\p}^2 ( \X-f_{\p}(\X| \Y_{\snr_0}) )-\sqrt{\Delta} \cdot m \cdot \Z_{\Delta} \right\|_{\p} }{\| \Z \|_{\p}^2+\Delta \cdot  m}, \label{eq: proof up to here SCPP}
\end{align}
where the (in)-equalities follow from: 
a) Proposition~\ref{prop:max ration comb};
b) by using the sub-optimal estimator in \eqref{eq:sub-optimalEstimator}; and
c) by choosing $\gamma=\frac{\| \Z \|_{\p}^2 }{\| \Z \|_{\p}^2+\Delta \cdot  m}$ for $m$ defined in \eqref{eq: m definition}.

Next, by applying the triangle inequality to \eqref{eq: proof up to here SCPP} we get 
\begin{align}
\mmpe^{\frac{1}{\p}}(\X,\snr,\p) 
&\le  \frac{ \left\| \| \Z \|_{\p}^2 ( \X-f_{\p}(\X| \Y_{\snr_0}) ) \right\|_{\p} +\left\|\sqrt{\Delta} \cdot m \cdot \Z_{\Delta} \right\|_{\p} }{\| \Z \|_{\p}^2+\Delta \cdot  m} \label{eq: triangle inequality}\\
&=  \frac{ \sqrt{m}\| \Z \|_{\p}  \cdot (\|\Z\|_{\p}  +\sqrt{\Delta} \cdot \sqrt{m}) }{\| \Z \|_{\p}^2+\Delta \cdot  m}  \\
& \le  \sqrt{2}  \frac{ \sqrt{m}  \| \Z \|_\p  }{ \sqrt{\| \Z \|_{\p}^2+\Delta \cdot  m}},
\end{align}
where in the last step we  used $(a+b) \le \sqrt{2} \sqrt{a^2+b^2}$. 

Note that for the case $\p=2$, instead of using the triangular inequality in~\eqref{eq: triangle inequality}, the term in \eqref{eq: proof up to here SCPP} can  be expanded into a quadratic equation for which it is not hard to see that the choice of $\gamma=\frac{\| \Z \|_{\p}^2 }{\| \Z \|_{\p}^2+\Delta \cdot  m}$ is optimal and leads to the bound
\begin{align*}
\mmpe^{\frac{1}{\p}}(\X,\snr,\p) \le \frac{ \sqrt{m}  \| \Z \|_{\p}  }{ \sqrt{\| \Z \|_{\p}^2+\Delta \cdot  m}}.
\end{align*}
The proof is concluded by noting that $\beta=   \frac{m}{ \| \Z \|_{\p}^2 -\snr_0 m}$.
\end{proof} 

\begin{remark}
We conjecture that the multiplicative constant $c_\p$  can be sharpened to $1$ for all $\p\ge 1$.  However, in order to make such a claim  one must solve the following optimization problem
\begin{align}
\min_{\gamma \in [0,1]}  \| (1-\gamma)\W + \gamma \Z  \|_\p, \label{eq:optimization problem}
\end{align} 
where $\W$ and $\Z$ are independent and $\Z \sim \mathcal{N}( {\bf 0},\I)$. 
Because it is not clear how to solve  \eqref{eq:optimization problem}  for $\p \neq 2$ and thus we leave it for the future work.  
\end{remark}

\begin{remark}
Note that the proof of Proposition~\ref{prop: SCP generalized} does not require the assumption that $\Z$ is  Gaussian and only requires the assumptions of Proposition~\ref{prop:max ration comb}. That is, we only require that a channel is such that the estimation of $\X$ from two observations is equivalent to estimating $\X$ from a single observation with a higher SNR. 
\end{remark}

\subsection{Complementary SCPP bound}
In this section we give a bound that complements the SCPP bound, that is, while  the SCPP bounds the MMPE for all $\snr \ge  \snr_0$,  we give a bound that bounds the MMPE for all $\snr \le \snr_0$ where it is assumed that the MMPE is known at $\snr_0$.

The next result enables us to bound the MMPE  at $\snr$ with values of the MMPE  at $\snr_0$ while varying  the order. 
\begin{prop} 
\label{prop:MMSE at low snr with MMPE}
For $ 0 < \snr \le \snr_0$, $\X$ and  $\p \ge 0$, we have
\begin{align*}
 & \mmpe(\X,\snr,\p) \le  \kappa _{n,t} \ \mmpe^{\frac{1-t}{1+t}} \left(\X,\snr_0,\frac{1+t}{1-t} \cdot \p \right) , \notag \\
& \text{where }  \kappa_{n,t}:=\left( \frac{2^n}{n^2} \right)^{\frac{t}{t+1}}\left( \frac{1}{1-t}\right)^{\frac{nt}{t+1}-\frac{1}{2}}, \;   t=\frac{\snr_0-\snr}{\snr_0}. 
  \end{align*}
\end{prop} 
\begin{proof}
From Proposition~\ref{prop:change of measure} we have that 
\begin{align}
\mmpe(\X,\snr,\p)&=  \inf_{f} \frac{1}{n}\E \left[ \Err^\frac{\p}{2}(\X,f(\Y_{\snr_0}))  \sqrt{\frac{\snr}{\snr_0 }} \eu^{\frac{\snr_0-\snr}{2\snr_0}  \sum_{i=1}^n Z_i^2 }\right] \notag\\
&\stackrel{a)}{\le} \inf_{f}  \sqrt{\frac{\snr}{\snr_0 }}  \frac{1}{n}\left(\E \left[   \Err^{m\cdot \frac{\p}{2}}(\X,f(\Y_{\snr_0})) \right] \right)^{\frac{1}{m}}  \left( \E \left[ \eu^{\frac{r(\snr_0-\snr)}{2\snr_0}  \sum_{i=1}^n Z_i^2 }\right] \right)^{\frac{1}{r}} \notag\\
&\stackrel{b)}{=} \sqrt{\frac{\snr}{\snr_0 }} n^{\frac{1}{m}-1} \mmpe^{\frac{1}{m}}(\X,\snr_0,m\cdot \p) \left(1-\frac{r(\snr_0-\snr)}{\snr_0} \right)^{-\frac{n}{2r}},\label{eq:ChangeMeasureANYn}
\end{align}
where the (in)-equalities follow from: a) H{\"o}lder's inequality with conjugate exponents $1 \le m,r$ such that $\frac{1}{m}+\frac{1}{r}=1$; and b) by recognizing that the expectation of the exponential is the moment generating function of a Chi-square distribution of degree $n$, which exists only if $\frac{r  (\snr_0-\snr)}{2 \ \snr_0}  <\frac{1}{2}$.

Next, we let $t=\frac{\snr_0-\snr}{\snr_0}$ and let $r=\frac{t+1}{2t}$ in \eqref{eq:ChangeMeasureANYn}, so that 
$m=\frac{1+t}{1-t}$. Observe that  now  the bound  in \eqref{eq:ChangeMeasureANYn} holds for all values of $\snr \le \snr_0$ since 
\begin{align}
\frac{r(\snr_0-\snr)}{\snr_0} =r t= \frac{(t+1)t}{2t}=\frac{t+1}{2} <1.
\end{align}
With the choice of $m=\frac{1+t}{1-t}$ the bound in  \eqref{eq:ChangeMeasureANYn} becomes
\begin{align*}
 \mmpe(\X,\snr,\p) & \le \sqrt{\frac{\snr}{\snr_0}} \left(1-\frac{r(\snr_0-\snr)}{\snr_0} \right)^{-\frac{n}{2r}}  n^{\frac{1}{m}-1} \left(\mmpe(\X,\snr_0,m \cdot \p) \right)^{\frac{1}{m}}\\
 & = \sqrt{\frac{\snr}{\snr_0}} \left(1-r t \right)^{-\frac{nt}{t+1}} n^{\frac{-2t}{1+t}}  \left(\mmpe \left(\X,\snr_0, \frac{1+t}{1-t} \cdot \p\right) \right)^{\frac{1-t}{1+t}}\\
  &=  \left( \frac{2^n}{n^2} \right)^{\frac{t}{t+1}}\left( \frac{\snr_0}{\snr}\right)^{\frac{nt}{t+1}-\frac{1}{2}} \left(\mmpe \left(\X,\snr_0,\frac{1+t}{1-t} \cdot \p \right) \right)^{\frac{1-t}{1+t}} .
  \end{align*}
  This concludes the proof.  
\end{proof}
The bound in Proposition~\ref{prop:MMSE at low snr with MMPE} is the key in showing  new bounds on the phase transitions region for the MMSE, presented in the next section.

As an application of Proposition~\ref{prop:MMSE at low snr with MMPE} we show that the  MMPE is a continuous function of SNR. 

\begin{prop} 
\label{prop:continuitySNR}
For fixed $\X$ and $\p$, $\mmpe(\X,\snr,\p)$ is a continuous function of $\snr>0$. 
\end{prop}

\begin{proof}
Assume without loss of generality that $\snr_0 \ge \snr$ 
\begin{align*}
&\lim_{\snr \to \snr_0}  |\mmpe(\X,\snr,\p)-\mmpe(\X,\snr_0,\p)|  \\
&\stackrel{a)}{=} \lim_{\snr \to \snr_0}  \mmpe(\X,\snr,\p)-\mmpe(\X,\snr_0,\p) \\
& \stackrel{b)}{\le}    \lim_{\snr \to \snr_0}  \kappa _{n,t} \ \mmpe^{\frac{1-t}{1+t}} \left(\X,\snr_0,\frac{1+t}{1-t} \cdot \p \right)-\mmse(\X,\snr_0,\p) \\
& \stackrel{c)}{=}\mmse(\X,\snr_0,\p)-\mmse(\X,\snr_0,\p)=0,
\end{align*}
where the (in)-equalities follow from: a) since the MMPE is a decreasing function of SNR and since $\snr_0 \ge \snr$; b) by using Proposition~\ref{prop:MMSE at low snr with MMPE}; and c)  by definition of $t$ in  Proposition~\ref{prop:MMSE at low snr with MMPE} we have that $ \lim_{\snr \to \snr_0} t=0$ and $ \lim_{\snr \to \snr_0} k_{n,t} =1$, and by continuity of the MMPE in $\p$ from Proposition~\ref{prop:continuity}.
This concludes the proof. 
\end{proof}

\section{Applications}
\label{sec: applications}

We next show how the MMPE  can be used to derive tighter versions of some well known  bounds. It is important to point out that even though the focus of this paper is on the AWGN setting, the results that follow (Theorem~\ref{prop:generalization of continuous Fano's inequality}, Theorem~\ref{prop:OWimproved} and Theorem~\ref{prop: Gap in Ozarow-Wyner at n infinity}) apply to any additive channel model in which the noise is an absolutely continuous random variable, without the need for the i.i.d. assumption.

\subsection{Bounds on the Differential Entropy}
\label{sec:boundDiffEntropy}
For any random vector $\U$ such that  $|\K_{\U}| < \infty, \ h( \U) < \infty,$ and any random vector $\V$, 
the following inequality is  considered to be a continuous analog of Fano's inequality~\cite{Cover:InfoTheory}:
\begin{align}  
 h(\U|\V) &\le \frac{n}{2} \log(2 \pi \eu  \ |\K_{\U| \V}|^{\frac{1}{n}} )  \\
 &\le  \frac{n}{2} \log(2 \pi \eu  \ \mmse(\U|\V) ),
\label{eq: mmse and entropy cov}
\end{align}
where  the inequality in \eqref{eq: mmse and entropy cov} is a consequence of the arithmetic-mean geometric-mean inequality, that is,  for any $0 \preceq {\bf A}$ we have used 
$| {\bf A} |^\frac{1}{n} = \left( \prod_{i=1}^n \lambda_i \right)^{\frac{1}{n}} \le \frac{\sum_{i=1}^n \lambda_i}{n}= \frac{\Trc({\bf A})}{n}$ 
where $\lambda_i$'s are the eigenvalues of ${\bf A}$. 

By applying \eqref{eq: mmse and entropy cov} to the AWGN setting, for any $\X$  such that $|\K_{\X}| < \infty, \ h( \X) < \infty,$ by using Proposition~\ref{prop:higher moments bound 1} with $q=1$, we can arrive at the trivial bound: for any $\p \ge 2$ 
\begin{align}
h(\X|\Y) \le  \frac{n}{2} \log\left(2 \pi \eu \cdot  n^ \frac{2-\p}{\p} \cdot \mmpe^ \frac{1}{p}(\X,\snr,\p) \right). 
\label{eq: trivial MMPE and entropy}
\end{align}

\noindent 
Next, we show that 
the inequality in \eqref{eq: mmse and entropy cov} can be generalized in terms of the norm in \eqref{eq: defintion of the norm}, and 
the trivial bound in \eqref{eq: trivial MMPE and entropy} can be improved. 
\begin{theorem}
\label{prop:generalization of continuous Fano's inequality}
For any $\U \in \mathbb{R}^n$ such that $  h(\U) < \infty$ and $ \| {\U} \|_\p < \infty$  for some $\p \in (0,\infty)$,  and for any $\V \in  \mathbb{R}^n$, we have
\begin{align}
h( \U | \V) \le \frac{n}{2}  \log \left( k_{n,\p}^2  \cdot n^{\frac{2}{\p}} \cdot  \mmpe^{\frac{2}{\p}}(\U|\V;p) \right),
\end{align} 
  where
\begin{align}
k_{n,\p} &:= \frac{ \sqrt{\pi}\left( \frac{\p}{n} \right)^{\frac{1}{\p}}\eu^{\frac{1}{\p}} \Gamma^{\frac{1}{n}} \left( \frac{n}{\p}+1\right)}{ \Gamma^{\frac{1}{n}}\left( \frac{n}{2}+1 \right)}= \sqrt{2 \pi \eu}  \frac{1}{n^{\frac{1}{2}} \left( \frac{\p}{2}\right)^{\frac{1}{2n}}} + o\left( \frac{n}{\p} \right).\label{eq: constant for moment entropy inequality}
\end{align}
\end{theorem}
\begin{proof}
See Appendix~\ref{proof:generalization of continuous Fano's inequality}.
\end{proof}

{Note that the result in  Theorem~\ref{prop:generalization of continuous Fano's inequality}   holds in great generality, i.e., the AWGN assumption is not necessary. As an application of Theorem~\ref{prop:generalization of continuous Fano's inequality} to the AWGN setting we have the following stronger version of the inequality in \eqref{eq: trivial MMPE and entropy}. }
\begin{corollary} For any $\X$ such that $h(\X) < \infty$ and $\| \X \|_p < \infty$ for some $\p \in (0,\infty)$, we have that 
 \begin{align*}
h(\X| \Y) \le \frac{n}{2}  \log \left( k_{n,\p}^2  \cdot n^{\frac{2}{\p}} \cdot \mmpe^{\frac{2}{\p}}(\X,\snr,\p) \right),
 \end{align*}
 where  $k_{n,\p}^2 $ is defined in \eqref{eq: constant for moment entropy inequality}.
 \end{corollary}
\begin{proof}
The proof follows by setting $\U=\X$ and   $\V=\Y$ in the statement of Theorem~\ref{prop:generalization of continuous Fano's inequality}.
\end{proof}

\subsection{Generalized Ozarow-Wyner Bound} 
\label{sec:genOWbound}
In~\cite{PAMozarow} the following ``Ozarow-Wyner lower bound'' on  the mutual information achieved by a discrete input $X_D$ 
transmitted over an AWGN channel was shown:
\begin{subequations}
\begin{align}
 &[H(X_D) -\gap]^{+} \le I(X_D;Y) \le H(X_D),  \label{eq:gap}
 \\
 &\gap \le \frac{1}{2} \log \left( \frac{\pi \eu}{6}\right) +\frac{1}{2} \log \left(1+\frac{ {\rm lmmse}(X,\snr)}{d_{\min}(X_D)^2} \right), 
\end{align} 
\label{eq:gap LMMSE}
\end{subequations}
where ${\rm lmmse}(X|Y)$ is the LMMSE. 
The advantage of the bound in~\eqref{eq:gap LMMSE} compared to existing bounds is its computational simplicity. 
The bound  on the $\gap$ in \eqref{eq:gap LMMSE} has been sharpened in \cite[Remark~2]{DytsoTINPublished} to 
\begin{align}
 \gap \le \frac{1}{2} \log \left( \frac{\pi \eu}{6}\right) +\frac{1}{2} \log \left(1+\frac{ \mmse(X,\snr)}{d_{\min}(X_D)^2} \right),
\end{align} 
since ${\rm lmmse}(X,\snr) \ge  \mmse(X,\snr)$.

Next, we generalize the bound in~\eqref{eq:gap LMMSE} to  discrete vector inputs and give the sharpest known bound on the gap term. 
\begin{theorem} \label{prop:OWimproved} (\emph{Generalized Ozarow-Wyner Bound})
Let $\X_D$ be  a discrete random vector with finite entropy, such that  $p_i=\mathbb{P}[\X_D= {\bf x}_i ]$, and ${\bf x}_i \in \supp(\X_D)$,  and let $\mathcal{K}_p $ be a set of  continuous random vectors, independent of $\X_D$, such that for every $\U \in\mathcal{K} $,   $ h(\U), \| \U \|_{\p} < \infty$,  and 
\begin{subequations}
\begin{align}
&\supp (\U+{\bf  x}_i)  \cap  \supp (\U+{\bf  x}_j) =\emptyset, \notag \\
& \ \forall \ {\bf x}_i, {\bf x}_j \in \supp(\X_D), i \neq j. \label{eq: assumption on U}
\end{align}
 Then for any $p > 0$
\begin{align}
 [H(\X_D) -\mathsf{gap}_\p]^{+} \le I(\X_D;\Y) \le H(\X_D), \label{eq: OWbound}
\end{align} 
where 
\begin{align}
 n^{-1} \mathsf{gap}_\p& \le  \inf_{\U \in \mathcal{K}_\p }  \left( G_{1,\p}(\U,\X_D) + G_{2,\p}(\U)  \right),  \notag\\
G_{1,\p}(\U,\X_D)&=  \log \left(   \frac{  \|  \U+\X_D-f_\p(\X_D|\Y) \|_\p}{ \| \U\|_\p} \right)  \\
& \stackrel{ \text{ for  $\p \ge 1$} }{\le}     \log \left( 1+\frac{ \mmpe^{\frac{1}{\p}}(\X_D,\snr,\p) }{\| \U \|_\p}  \right), \\
G_{2,\p}(\U)&= \log \left(  \frac{  k_{n,\p}  \cdot n^{\frac{1}{\p}} \cdot  \| \U \|_\p}{\eu^{ \frac{1}{n} h_{\eu}(\U)} }\right).
\end{align}
\label{eq:gapOWgeneral}
\end{subequations} 
\end{theorem} 
\begin{proof}
See Appendix~\ref{app:prop:OWimproved}.
\end{proof} 

It is interesting to note that the lower bound in \eqref{eq: OWbound} resembles the bound for lattice codes in \cite[Theorem 1]{Forney:ShannonWinener:2003}, where $\U$ can be thought of as dither,  $G_{2,\p}$ corresponds to the log of the normalized $\p$-moment of a compact region in $\mathbb{R}^n$,  $G_{1,\p}$  corresponds  to the log  of the normalized MMSE term, and $H(\X_D)$ corresponds with the capacity $C$.  

In order to show the advantage of Theorem~\ref{prop:OWimproved} over the original Ozarow-Wyner bound (case of $n=1$ and with LMMSE instead of MMPE),  we consider $X_D$ uniformly distributed with the number of points equal to $N= \lfloor\sqrt{1+\snr} \rfloor$, that is, we choose the number of points such that $H(X_D) \approx \frac{1}{2} \log(1+\snr)$.  
Fig.~\ref{fig:GapU} shows:
\begin{itemize} 
\item 
The solid cyan line is the ``shaping loss'' $\frac{1}{2}\log \left( \frac{\pi \eu}{6} \right)$ for a one-dimensional infinite lattice and is the limiting gap if the number of points $N$ grows faster than $\sqrt{\snr}$;
\item 
The solid magenta line is the gap in the original Ozarow-Wyner bound  in \eqref{eq:gap LMMSE}; and
\item 
The dashed purple,  dashed-dotted blue and dotted green lines are the new gap due to Theorem~\ref{prop:OWimproved} for value of $\p=2,4,6$, respectively, and where we chose  $U \sim \mathcal{U} \left[-\frac{d_{\min(X_D)}}{2},\frac{d_{\min(X_D)}}{2} \right]$.
\end{itemize} 
 We note that the  version of the Ozarow-Wyner bound in Theorem~\ref{prop:OWimproved} provides the sharpest bound for the gap term.
An open question, for $n=1$, is what value of $\p$  provides the smallest gap and if that  coincide with the ultimate ``shaping loss''.

\begin{figure}
        \centering
        \includegraphics[width=7.5cm]{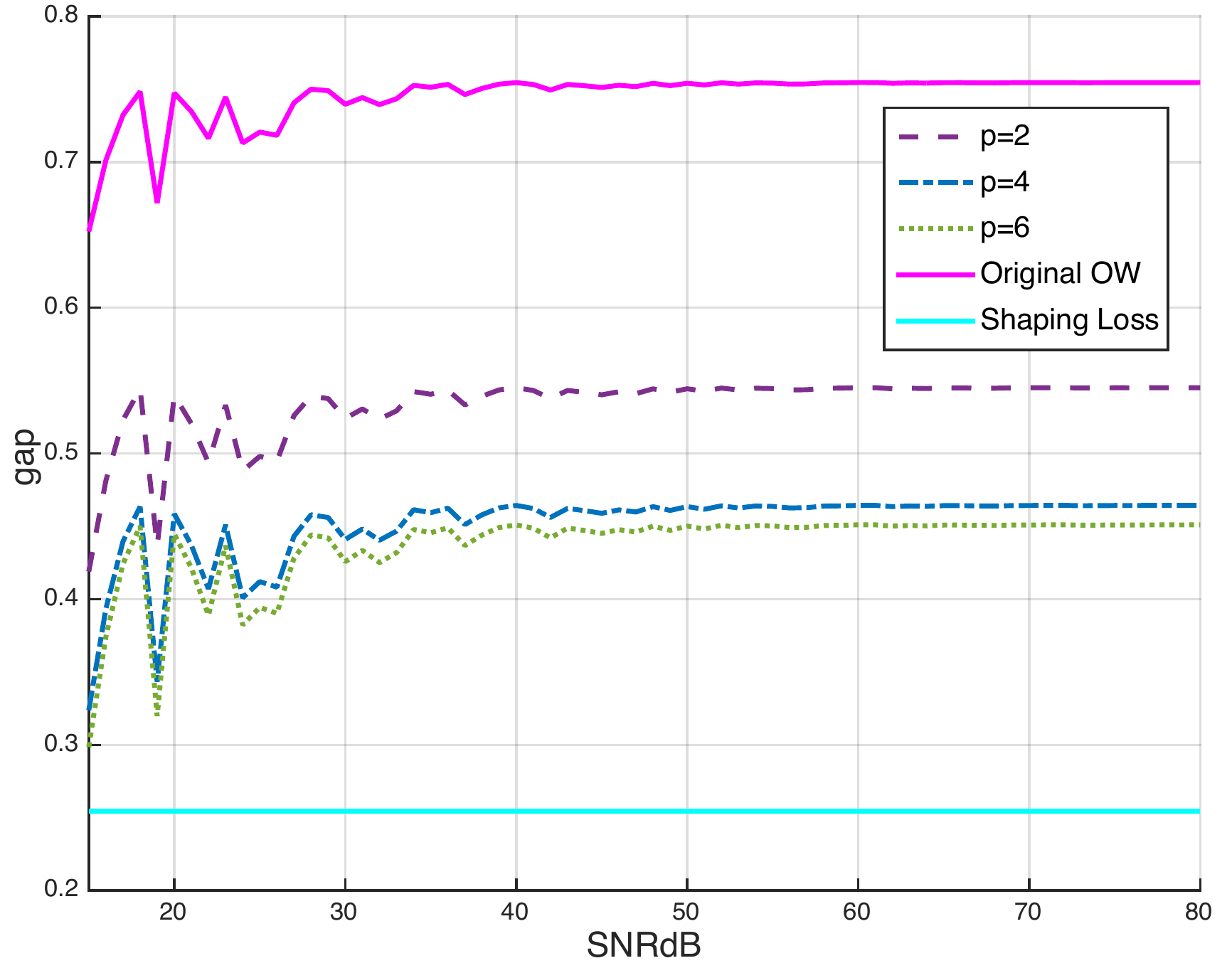}
        \caption{ Gap  in equation \eqref{eq:gap}  and \eqref{eq:gapOWgeneral} vs. $\snr$.}
       \label{fig:GapU}
\end{figure}

Next we turn our attention to the case of $n>1$. Another interesting question is how the gap behaves as  $n \to \infty$. 

\begin{theorem}
\label{prop: Gap in Ozarow-Wyner at n infinity}
Let $\U$  be uniform  over the ball of radius $ r=\frac{d_{\min}(\X_D)}{2}$ then  for any $\p > 0$ 
\begin{subequations}
\begin{align}
G_{2,\p} (\U) = O \left( \frac{1}{n}\log \left( \frac{n}{\p} \right)\right). 
\end{align}
and therefore  $ \lim_{n \to \infty} G_{2,\p} (\U)=0$. 
Therefore, 
\begin{align}
  \frac{1}{n} H(\X_D) \ge \frac{1}{n} I(\X_D; \Y) \ge    \frac{1}{n} H(\X_D)-    G_{1,\p} (\U, \X_D) - O \left( \frac{1}{n}\log \left( \frac{n}{\p} \right)\right),
\end{align} 
where 
\begin{align}
\eu^{ G_{1,\p} (\U, \X_D)} \stackrel{\text{ for $\p \ge 1$}}{ \le } 1+2 \frac{ d_{\max}(\X_D)}{ d_{\min}(\X_D)} \sqrt[\p]{   \frac{ (\p+n)     }{ n  } \bar{Q}\left( \frac{n}{2}; \frac{\snr d^2_{ \min}(\X_D) }{8}\right)}.
\end{align}
\end{subequations}
\end{theorem}
\begin{proof}
See Appendix \ref{app:prop: Gap in Ozarow-Wyner at n infinity}.
\end{proof}

\subsection{New bounds on  the MMSE and Phase Transitions}
\label{prop:phaseTrans:bounds}

The  SCPP is  instrumental in showing the behavior of the MMSE of capacity achieving codes. For example, 
as the length of any capacity achieving code goes to infinity, the MMSE  behaves as follows:
 \begin{align}
\limsup_{ n \to \infty} \mmse(\X,\snr) = \left \{  \begin{array}{ll} 
\frac{1}{1+\snr},  & 0 \le \snr \le \snr_0\\
\frac{\beta}{1+\beta \snr}, & \snr_0  \le \snr \le \snr_1\\
\frac{\gamma}{1+ \gamma \snr},  &\snr \ge \snr_1
\end{array} \right.  ,
\label{eq: optimal codes MMSE}
\end{align} 
as shown:
in~\cite{MerhavStatisticalPhysics}, for the  Gaussian point-to-point channel with the output $Y_{\snr_0}$ with $\beta=\gamma=0$;
in \cite{SNRevolutionOfMMSE}, for the Gaussian BC  with outputs $Y_{\snr_1}$ and $Y_{\snr_0}$, where $\snr_0 \le \snr_1$ and rate pair $(R_1,R_2)= \left(\frac{1}{2} \log(1+\beta \snr_1), \frac{1}{2} \log \left(\frac{1+\snr_0}{1+\beta \snr_0} \right) \right)$ for some $\beta \in [0,1]$, with $\gamma=0$;
in~\cite{SNRevolutionOfMMSE}, for the Gaussian wiretap channel with  outputs $Y_{\snr_0}$ (primary)  and $Y_{\snr_1}$ (eavesdropper) with maximum equivocation $d_{\max}$ and rate $R \ge d_{\max}$, for $\beta=\gamma=0$;
and in~\cite{BustinMMSEbadCodes}, for the Gaussian point-to-point channel with output $Y_{\snr_1}$  and  an MMSE disturbance constraint at $Y_{\snr_0}$ measured 
by $\mmse(\X,\snr_0) \le \frac{\beta}{1+\beta \snr_0}$ for some $\beta \in [0,1]$ with $\gamma=\beta$.
%
The jump discontinuities in~\eqref{eq: optimal codes MMSE} at $\snr=\snr_0$ and $\snr=\snr_1$ are referred to as the \emph{phase transitions}. 

Based on the above, an interesting question is how the MMSE in~\eqref{eq: optimal codes MMSE} behaves for codes of finite length.  In \cite{NewBoundsOnMMSE}, in order to study the  phase transition phenomenon for inputs of finite length, the following optimization problem was proposed:
\begin{definition}
\begin{subequations}
\begin{align}
&\mathrm{M}_n(\snr,\snr_0,\beta) := \sup_{\X}
\mmse(\X,\snr), 
\label{eq:MMSE obj for MMSE}
\\
\text{ s.t. }& \| \X \|_2^2 \le 1, \text{ and } \mmse(\X,\snr_0) \le \frac{\beta}{1+\beta \snr_0},\label{eq:MMSE constr: for MMSE}
\end{align}
\label{eq: mmse constrained MMSE} 
\end{subequations}
for some $\beta \in [0,1]$. 
\end{definition} 
Investigation in~\cite{NewBoundsOnMMSE} revealed that $\mathrm{M}_{n}(\snr,\snr_0,\beta)$ in~\eqref{eq:MMSE obj for MMSE} must be of the following form:
\begin{align*}
\mathrm{M}_n(\snr,\snr_0,\beta)  =  \left \{  \begin{array}{ll} 
\frac{1}{1+\snr}, & \snr \le \snr_L \\
T_n(\snr,\snr_0,\beta), & \snr_L \le  \snr \le \snr_0 \\
\frac{\beta}{1+\beta \snr},& \snr_0 \le \snr
 \end{array} \right.  ,
\end{align*}
for some $\snr_L$ and some function $T_n(\snr,\snr_0,\beta)$, where the region $\snr_L \le \snr \le \snr_0$ is referred to as the \emph{phase transition region} and its width is defined as
$W(n) :=\snr_0-\snr_L.$
In \cite{NewBoundsOnMMSE} the following 
was established for $T_n(\snr,\snr_0,\beta)$ and $W(n)$:
\begin{align}
\mmse(\X,\snr)  &\le  \mmse(\X,\snr_0)+\kappa_n \left( \frac{1}{\snr}-\frac{1}{\snr_0} \right),  
\notag\\
&\text{\rm where }   \kappa_n \le n+2,
\label{eq:mainBound}
\end{align}
and  the width of phase transition region scales as $W(n)=O \left(n^{-1}\right).$

The main result of this subsection is shown next. It uses Propositions~\ref{prop:MMSE at low snr with MMPE}  and Proposition~\ref{prop:logconvex}. 
\begin{theorem} 
\label{prop: bound through MMPE} 
For $ 0 < \snr \le \snr_0$, 
\begin{subequations}
\begin{align}
&\mmse(\X,\snr) \le  \min_{ r>\frac{2}{\gamma}}   \kappa(r,\gamma,n)\left(  \frac{\beta}{1+\beta \snr_0}\right)^{ \frac{ \gamma r-2}{r-2} }, \label{eq: bound in Thm 1}
\end{align}
 where  $\gamma:=\frac{\snr}{2\snr_0-\snr} \in (0,1],$ and 
\begin{align}
 &\kappa(r,\gamma,n):=\frac{\sqrt{2}}{n^{1-\gamma}} \left(\frac{1+\gamma}{\gamma}\right)^{\frac{n(1- \gamma)-1}{2}}   M_r^{\frac{2(1-\gamma)}{r-2}}, \\
&M_r :=  \left\| \X-\E\left[\X | \Y_{\snr_0} \right] \right \|_r^{r} \le 2^{r} \min \left(  \frac{\|  \Z\|_r^{r}}{\snr_0^\frac{r}{2}}, \| \X \|_r^{r} \right),
\end{align}
and where the minimizing $r$ in~\eqref{eq: bound in Thm 1} can be approximated by
\begin{align}
 r_{\text{opt}} \approx  \left \{\begin{array}{ll} 2 \ln \left( \frac{ 4 \eu  }{ \snr_0 \mmse(\X,\snr_0)   } \right), & \frac{2}{\gamma} \le  \ln \left( \frac{ 4 \eu  }{ \snr_0 \mmse(\X,\snr_0)   } \right)\\  
 \frac{2}{\gamma}, & \frac{2}{\gamma} >  \ln \left( \frac{ 4 }{ \snr_0 \mmse(\X,\snr_0)   } \right)\\ 
  \end{array} \right. .  \label{eq: approx optimal r}
\end{align}
Moreover, the width of the phase transition region scales as
\begin{align}
W(n)=O \left( n^{-\frac{1}{2}}\right).
\end{align}
\end{subequations}
\end{theorem} 

\begin{proof}
From the SCPP complementary bound in  Proposition~\ref{prop:MMSE at low snr with MMPE} with $\p=1$  we have that
\begin{align}
\mmse(\X,\snr) \le  \kappa _{n,t} \ \mmpe^{\frac{1-t}{1+t}} \left(\X,\snr_0,\frac{1+t}{1-t}  \cdot 2 \right). 
\end{align}
  From the interpolation result in Proposition~\ref{prop:logconvex}   letting $q=\frac{1+t}{1-t}\cdot 2$, $p=2$  we have that  for some $r$ such that $2 \le 2\frac{1+t}{1-t}=q <r$ and 
 \begin{align}
 \alpha&=\frac{  \frac{1-t}{2(1+t)}-\frac{1}{r}}{\frac{1}{2}-\frac{1}{r}} \Rightarrow 1-\alpha= \frac{\frac{2t}{1+t} r}{r-2},
 \end{align}
and thus the MMPE term can be bounded as 
 \begin{align*}
\mmpe^{\frac{1-t}{1+t}} \left(\X,\snr_0,\frac{1+t}{1-t}  \cdot 2\right) 
  & \le \mmpe^\alpha(\X,\snr_0,2)   \| \X -\E[\X|\Y_{\snr_0}] \|^{2 (1-\alpha)}_{r } \\
  &=  \mmse^\alpha(\X,\snr_0)  \| \X -\E[\X|\Y_{\snr_0}] \|^{2(1-\alpha)}_{r }\\
   &=  \mmse^\alpha(\X,\snr_0)  \| \X -\E[\X|\Y_{\snr_0}] \|^{2r\frac{\frac{2t}{1+t} }{r-2}}_{r }.
 \end{align*}

By  Proposition~\ref{prop:higher moments bound 1} we can  bound  $ \| \X -\E[\X|\Y_{\snr_0}] \|^{2 r\frac{\frac{2t}{1+t} }{r-2}}_{r }$  as follows:

\begin{align*}
 \| \X -\E[\X|\Y_{\snr_0}] \|^{r\frac{\frac{4t}{1+t} }{r-2}}_{r } \le   \left(2^{r} \min \left(  \frac{\|  \Z\|_r^{r}}{\snr^\frac{r}{2}}, \| \X \|_r^{r} \right)\right)^{\frac{4t}{(1+t)(r-2)}}.
\end{align*}

By putting  all  of the  bounds together, letting $\gamma=\frac{1-t}{1+t}$ and observing that 
\begin{align*}
1-\gamma&=\frac{2t}{1+t},\\
\gamma&=\frac{\snr}{2\snr_0-\snr},\\
\frac{\snr_0}{\snr}&=\frac{1+\gamma}{2 \gamma},
\end{align*}
we get the bound in \eqref{eq: bound in Thm 1}.
 Finally, the proof of approximately optimal $r$  in \eqref{eq: approx optimal r} is given in  Appendix~\ref{app: approximate r}. 
\end{proof}

    \begin{figure*}
          \begin{subfigure}[t]{0.5\textwidth}
           \centering
                \includegraphics[width=7cm]{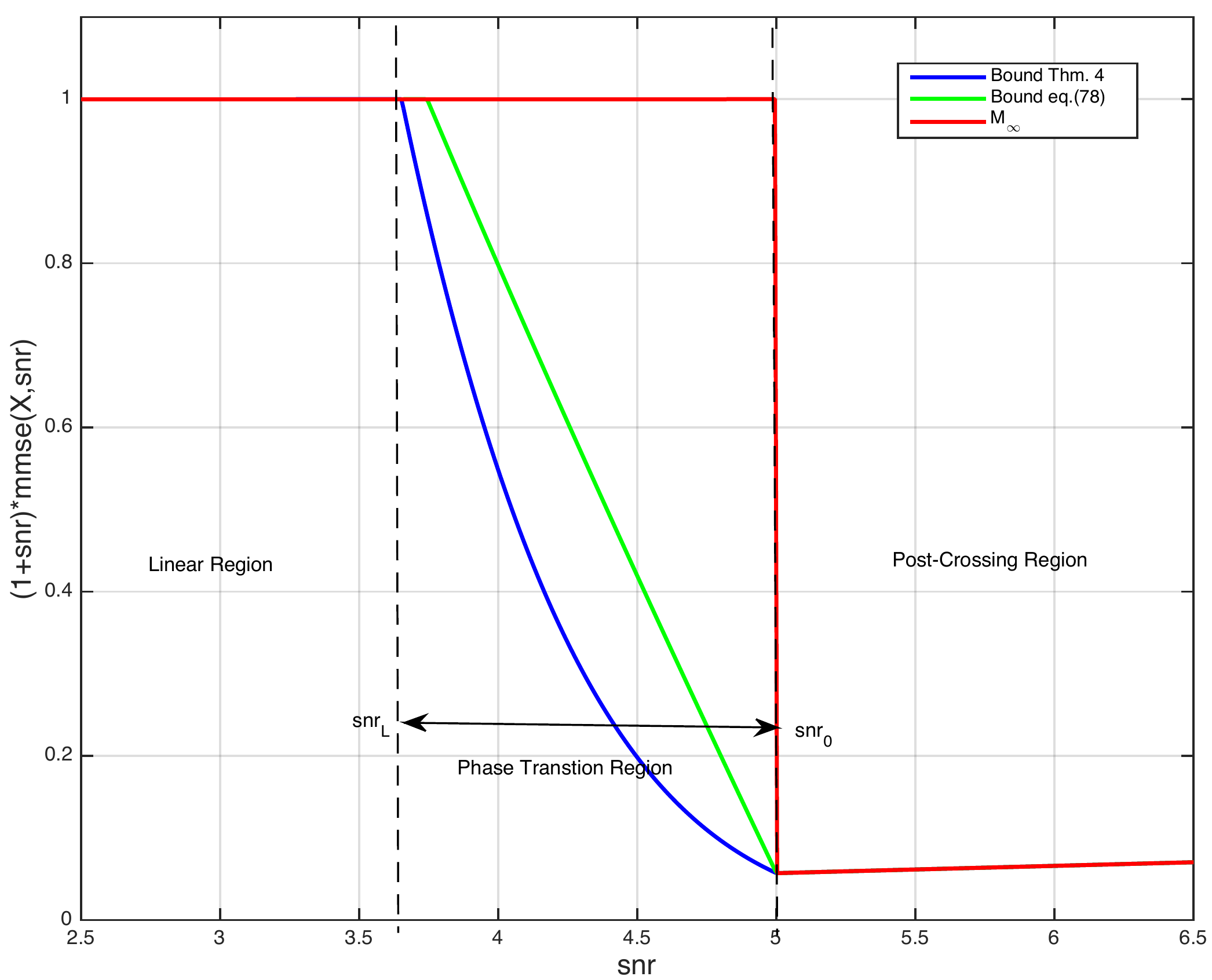}
                \caption{
                For  $\snr_0=5$ and $\beta=0.01$.  Here $n=1$.}
                \label{fig:Example of Phase}
        \end{subfigure}%
           \begin{subfigure}[t]{0.5\textwidth}
           \centering
                \includegraphics[width=7cm]{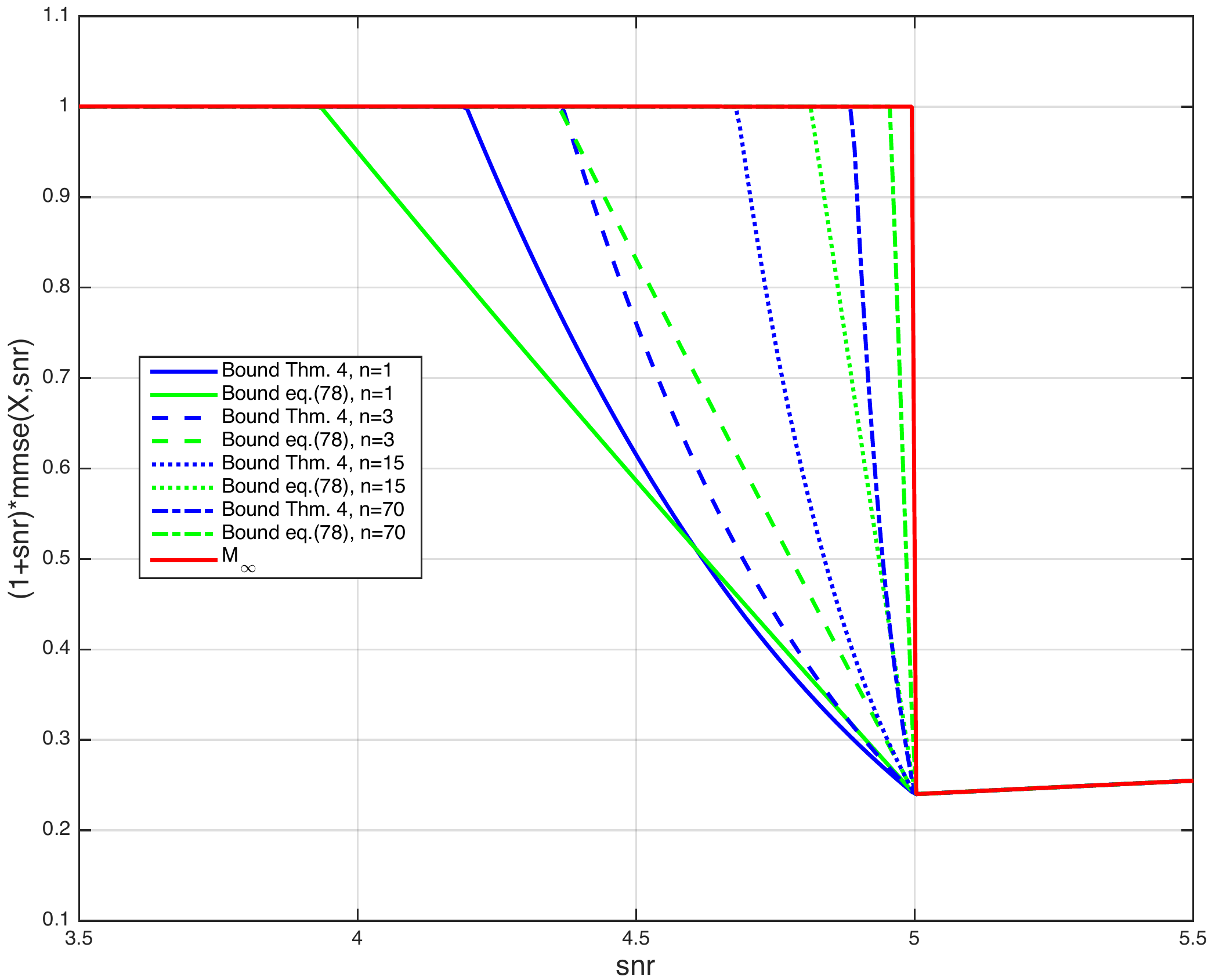}
                \caption{
                For $\snr_0=5$ and $\beta=0.05$.  Several values of $n$.}
                \label{fig:ExampleWith- n}
        \end{subfigure}
       \caption{Bounds on  $\mathrm{M}_n(\snr,\snr_0,\beta)$ vs  $\snr$.}
        \label{fig: Bounds on Mn}
\end{figure*}
The bounds in Theorem~\ref{prop: bound through MMPE} and in \eqref{eq:mainBound} are shown in Fig.~\ref{fig: Bounds on Mn}. The bound in Theorem~\ref{prop: bound through MMPE} is asymptotically tighter than the one in \eqref{eq:mainBound}. This follows since  the phase transition region shrinks as  $O \left( \frac{1}{\sqrt{n}} \right)$ for Theorem~\ref{prop: bound through MMPE}, and as $O \left( \frac{1}{{n}} \right)$ for the bound in \eqref{eq:mainBound}.
It is not possible in general to assert that Theorem~\ref{prop: bound through MMPE} is tighter than  \eqref{eq:mainBound}.
In fact, for small values of $n$, the bound in \eqref{eq:mainBound} can  offer advantages, as seen for the case $n=1$ shown in Fig.~\ref{fig:ExampleWith- n}. Another advantage of the bound in \eqref{eq:mainBound} is its analytical simplicity.

\subsection{Bounds on the  derivative of the MMSE}
\label{sec:BoundsOnDerivative}

The MMPE can be used to study  the second derivative of mutual information (or first derivative of MMSE), as initiated 
for $n=1$ in~\cite{GuoMMSEprop} and for $n \ge 1$ in~\cite{BustinMMSEparallelVectorChannel}, namely,
\begin{align}
&\frac{d^2 I(\X,\Y)}{d \snr^2}
= n \frac{d \ \mmse(\X,\snr)}{d \snr} 
= -\Trc \left(\E \left[\cov^2(\X|\Y) \right] \right), 
\notag\\
& 
\cov(\X|\Y)\!:=\!\E\left[(\X-\E[\X|\Y])(\X-\E[\X|\Y])^T|\Y\right].
\label{eq: derivative}
\end{align}
The second derivative of mutual information is  important in characterizing the  bandwidth-power trade-off in the wideband regime~\cite{MIsecondOrder} and \cite{SpectralEfficiencyWideband}, and has also been used in the proof of the SCPP in \cite{GuoMMSEprop} and \cite{BustinMMSEparallelVectorChannel}.  Moreover, in \cite{GuoMMSEprop} it has been shown that the derivative of the MMSE and the quantity in~\eqref{eq: higher order moments} are related by the following bound for $n=1$:
\begin{align}
\E \left[\Cov^2(X|Y) \right] \le \| X -\E[X|Y] \|_4^4 \le \frac{3 \cdot 2^4}{\snr^2}. \label{eq: bound on der}
\end{align}
The main result of this subsection is the next bound.
\begin{prop}
\label{prop:bounds on derivative of MMSE via MMPE}
 For any input $\X$
\begin{align}
& \mmse^2(\X,\snr) =\mmpe^2(\X,\snr,2) \notag\\
 &\le \frac{1}{n}\Trc \left(\E \left[\cov^2(\X|\Y)\right] \right) \le n \ \mmpe(\X,\snr, 4).
\end{align}
\end{prop}
\begin{proof}
See Appendix~\ref{app:prop:bounds on derivative of MMSE via MMPE}.
\end{proof}

It can be observed that, for the case $n=1$, by using the bound in~\eqref{eq: 1/snr bound for higher orders} from Proposition~\ref{prop:higher moments bound 1} we have that
\begin{align}
\E \left[\Cov^2(X|Y)\right] \le \mmpe(X,\snr,4) \le \frac{3}{ \snr^2}, 
\label{eq: n=1 MMPE bounds on derivative}
\end{align}
which significantly reduces the constant in~\eqref{eq: bound on der} from $ 3 \cdot 2^4$ to $3$.
For a similar but slightly different bound than that in~\eqref{eq: n=1 MMPE bounds on derivative}  on $\E \left[\Cov^2(X|Y)\right]$ please see \cite{NewBoundsOnMMSE}.

\section{Concluding Remarks}
This paper has considered the  problem of estimating a random variable from a noisy observation under a general cost function, termed the MMPE. We have show that many properties of the MMSE and the conditional expectation (i.e., optimal MMSE estimator)  are identical or have a natural generalization to the MMPE and the MMPE optimal estimator.    

We have also provided a new simpler proof of the SCPP for the MMSE and generalized it to the MMPE. We have shown that the new framework of the MMPE also permits the  development of bounds that are complementary to the SCPP which in turn allows for new tighter characterizations of the phase transition phenomena that manifest, in the limit
as the length of the capacity achieving code goes to infinity, as a discontinuity of the MMSE as a function of SNR.

We  have also shown connections between the MMPE and the conditional differential entropy by generalizing  a well know continuous analog of Fano's inequality. 
 The MMPE was further used to refine bounds on the conditional entropy and improve the gap term in the Ozarow-Wyner bound.  

 Currently, we are investigating the connections between bounds on the MMPE provided in this work and the rate distortion problem with the MMPE distortion measure. Possible future applications of the sharpened version of the Ozarow-Wyner bound include sharpening  the bounds on discrete inputs in \cite{DiscreteWu} and \cite{shamai1991information}. Another interesting future direction is to consider a modified `information bottleneck problem' \cite{chechik2005information} where the constraint on the mutual information is replaced by a constraint on the MMPE.

\begin{appendices}
\section{Proof of the Triangle Inequality in~\eqref{eq: Minkowski}}
\label{app:triangleInequality Proof of lem: Minkowski}
It is well know that the trace operator is an inner product  in the space of matrices and since the inner product induces a norm we have 
\begin{align*}
\Trc \left( {\bf v}{\bf v}^T\right) = \langle {\bf v}, {\bf v} \rangle= \| {\bf v}\|^2.
\end{align*}
Therefore, we have that 
\begin{align*}
\E^{\frac{1}{\p}} \left[ \Trc^\frac{\p}{2} \left( \left({\bf V}-\U \right)\left({\bf V}-\U \right)^T \right) \right]&=\E^{\frac{1}{\p}} \left[\left \|{\bf V}-\U \right\|^{\p} \right]\\
& \stackrel{a)}{\le} \E^{\frac{1}{\p}} \left[\left( \|{\bf V} \|+ \|\U \| \right)^{\p} \right]\\
& \stackrel{b)}{\le}  \E^{\frac{1}{\p}} \left[\|{\bf V} \|^{\p} \right]+\E^{\frac{1}{\p}} \left[ \|\U \|^{\p} \right]\\
&= \E^{\frac{1}{\p}} \left[ \Trc^\frac{\p}{2} \left( {\bf V}{\bf V}^T \right) \right]+\E^{\frac{1}{\p}} \left[ \Trc^\frac{\p}{2} \left( \U \U^T \right) \right],
\end{align*}
where the inequalities follow from:  a) triangle inequality for inner product induce  norm for $\p\ge 1$; and  b) Minkowski inequality for the expectation which holds for $\p\ge 1$.
This concludes the proof. 


\section{ Proof of Proposition~\ref{prop:existence of optimal estimator}}
\label{app:existence of optimal estimator}
For simplicity, we look at the case $n=1$. The case for $n>1$ follows similarly. We first assume that $\snr>0$.
The first direction follows trivially:
\begin{align}
\inf_f \E \left[ |X-f(Y)|^{\p}\right] \le \E \left[ |X-f_p(X|Y)|^{\p}\right]. 
\end{align}
The other direction follows by using
\begin{align}
\inf_f \E \left[ |X-f(Y)|^{\p}\right] \ge \E \left[ \inf_f  \E \left[ \left|X-f(Y)|^{\p} \right | Y\right]\right],
\end{align}
where  we focus on the inner expectation $\inf_f  \E \left[ \left|X-f(Y)|^{\p} \right | Y=y\right]$ and show that the infimum is achieved by $f(y)=f_{\p}(X|Y=y)$ given in \eqref{eq: optimal estimator}. Since  $y$ is now given, we are simply looking for an optimal solution to the more general problem 
\begin{align}
 \inf_{v \in\mathbb{R}}  \E \left[ \left|X_y-v\right|^{\p} \right], \label{eq: infimumum of Opt Prob}
\end{align}
where $X_y \sim p_{X|Y}(\cdot|y)$. The goal is to show that the infimum in \eqref{eq: infimumum of Opt Prob} is achievable.  Clearly, the infimum exists since
\begin{align}
0 \le \inf_{v \in\mathbb{R}}  \E \left[ \left|X_y-v\right|^{\p} \right] &\le    \E \left[ \left|X_y-0 \right|^{\p} \right]  \\
&=\E[ |X_y|^{\p}] < \infty, \label{eq: inf existence}
\end{align}
where the last inequality follows from \cite[Proposition 6]{GuoMMSEprop} which asserts that for any $\p<\infty$,  $X_y$ is a sub-Gaussian random variable and hence all conditional moments are finite.

Next, we show that $ g(v)= \E \left[ \left|X_y-v\right|^{\p} \right]$  is a continuous function of $v$. Recall, that any given function $h(x)$ is continuous if   $x_n \to x$ implies $h(x_n) \to h(x)$ as $n \to \infty$.

For  arbitrary $|v| < \infty$  take a sequence $v_n$ such that  $v_n \to v$, we want to show that
\begin{align}
\lim_{n \to \infty} g(v_n)=\lim_{n \to \infty}  \E \left[ \left|X_y-v_n\right|^{\p} \right]=   \E \left[ \lim_{n \to \infty} \left|X_y-v_n\right|^{\p} \right]=g(v).
\end{align}
This can be done with the help of the dominated convergence theorem.
We must find an integrable random variable $ \theta$ such that $ \left|X_y-v_n\right|^{\p}  \le  \theta $ for all $n$; this is found as  
\begin{align}
\left|X_y-v_n\right|^{\p}   & \stackrel{a)}{\le} 2^{\p}  \left(|X_y|^{\p}+ |v_n|^{\p}\right)\\
& \stackrel{b)}{\le}  2^{\p}  \left(|X_y|^{\p}+K\right)=\theta,
\end{align}
where the inequalities follow from: a)  $ \left|X_y-v_n\right|^{\p} \le (2 \max(|X_y|,|v_n|))^{\p} \le 2^{\p} \left(|X_y|^{\p}+|v_n|^{\p}\right)$ which holds for any $\p\ge0$; and b) recall that every convergent sequence is bounded and since the sequence $v_n$ converges to $v$ it is also bounded by some finite $K$ for every $n$.
The  integrability of $\theta=2^{\p}  \left(|X_y|^{\p}+K\right)$ follows  again by the sub-Gaussian argument from \cite[Proposition 6]{GuoMMSEprop}.
Therefore, we conclude that the function $g(v)$ is continuous. 

 Next, we show that the infimum is attainable by some $|v_0| <\infty$.
By definition of the infimum there exists some $v_n$ (not necessarily convergent) such that 
\begin{align}
\liminf_{n \to \infty} \E[|X_y-v_n|^\p ]= \inf_{v \in\mathbb{R}}  \E \left[ \left|X_y-v\right|^{\p} \right].
\end{align}
Towards a contradiction, assume that $v_n \to \infty$. Then by Fatou's lemma
\begin{align}
\liminf_{n \to \infty} \E[|X_y-v_n|^\p ]\ge \E[  \liminf_{n \to \infty} |X_y-v_n|^\p ] =\infty.
\end{align}
However, this contradicts the result in \eqref{eq: inf existence} and therefore sequence $v_n$ must be bounded. This, together with the fact that $g(v)$ is continuos, implies that the infimum is attainable and thus
\begin{align}
 \inf_{v \in\mathbb{R}}  \E \left[ \left|X_y-v\right|^{\p} \right]= \min_{v \in\mathbb{R}}  \E \left[ \left|X_y-v\right|^{\p} \right].
\end{align}
Therefore, for each $y\in \mathbb{R}$ there exists $|v|<K$ that minimize the expression $\min_{v \in\mathbb{R}}  \E \left[ \left|X_y-v\right|^{\p} \right]$ and the optimal estimator defined point-wise is given by
\begin{align}
f(y)= \arg \min_{v \in\mathbb{R}}  \E \left[ \left|X_y-v\right|^{\p} \right]. 
\end{align}

For the case of $\snr=0^{+}$  the problem reduces to 
\begin{align}
\inf_ { {\bf v} \in \mathbb{R}^n} \| \X - {\bf v} \|_\p,
\end{align} 
which is bounded if and only if  $\| \X \|_\p < \infty$. 
This concludes the proof.

\section{Proof of Proposition~\ref{prop:orthogonality like property}}
\label{app:prop:orthogonality like property}
We take a classical approach used in estimation theory to find an optimal estimator by using tools from calculus of variations \cite[Ch.7 Thm.1]{luenberger1997optimization}. A necessary condition for $f$ to be a minimizer in \eqref{eq: optimal estimator} is expressed through a functional derivative as
\begin{align}
\nabla_g  \E \left[\Err^{\frac{\p}{2}}  ( \X, f(\Y) ) \right]=\lim_{\epsilon \to 0}  \E \left[\frac{ \Err^{\frac{\p}{2}} ( \X, f(\Y)+ \epsilon g(\Y) )-\Err^{\frac{\p}{2}} ( \X,  f(\Y))}{\epsilon} \right]=0,
\end{align}
for all admissible $g(\Y)$.

Therefore, we focus on the following limit:
\begin{align}
\lim_{\epsilon \to 0}  \E \left[\frac{ \Err^{\frac{\p}{2}} ( \X, f(\Y)+ \epsilon g(\Y) )-\Err^{\frac{\p}{2}} ( \X,  f(\Y))}{\epsilon} \right]. \label{eq: limit that we want}
\end{align}
We seek to apply the dominated convergence theorem to  \eqref{eq: limit that we want} in  order to interchange the order of the limit and the expectation.  To that end  we let ${\bf v}= {\bf x}-f({\bf y})$ and
\begin{align}
\Err^{\frac{\p}{2}} ( {\bf x},  f({\bf y})) &= \left( {\bf v}^T {\bf v} \right)^{\frac{\p}{2}}\\
 \Err^{\frac{\p}{2}} ( {\bf x}, f({\bf y})+ \epsilon g({\bf y}) ) &= \left( ({\bf v}- \epsilon g({\bf y}))^T ( {\bf v}- \epsilon g({\bf y})) \right)^{\frac{\p}{2}}\\
 &=  \left( {\bf v}^T{\bf v} -\epsilon  g({\bf y})^T  {\bf v} -\epsilon {\bf v}^T g({\bf y})+\epsilon^2 g({\bf y})^T g({\bf y})\right)^{\frac{\p}{2}} 
\end{align}

Next  for the integrant 
\begin{align}
&\frac{ \Err^\p ( {\bf x}, f({\bf y})+ \epsilon g({\bf y}) )-\Err^\p ( {\bf x},  f({\bf y}))}{\epsilon}
 \label{eq: to dominate} 
\end{align} 
we observe that all the terms in \eqref{eq: to dominate} are of order no more than $\p$, and since all of the terms are in $L_{\p}$ (or $\p$ integrable)  the quantity in \eqref{eq: to dominate} is integrable for any $\epsilon$. Therefore, the dominated convergence theorem applies and we can interchange the order of limit and  expectation in \eqref{eq: limit that we want}.

Next, observe that we can re-write the limit as a derivative, that is,
\begin{align}
 \lim_{\epsilon \to 0}  \frac{ \Err^{\frac{\p}{2}} ( {\bf x}, f({\bf y})+ \epsilon g({\bf y}) )-\Err^{\frac{\p}{2}} ( {\bf x},  f({\bf y}))}{\epsilon} =  \frac{d}{d \epsilon}   \Err^{\frac{\p}{2}} ( {\bf x}, f({\bf y})+ \epsilon g({\bf y}) ) \Big |_{\epsilon=0}. \label{eq: derivative and limit}
\end{align}

By using chain rules of differentiation of matrix calculus we arrive at 
\begin{align}
\frac{d}{d \epsilon}   \Err^{\frac{\p}{2}} ( {\bf x}, f({\bf y})+ \epsilon g({\bf y}) )  \Big |_{\epsilon=0}
&=  \frac{d}{d \epsilon}   \left( ({\bf v}-\epsilon g({\bf y}) )^T  ({\bf v}-\epsilon g({\bf y}) ) \right)^{\frac{\p}{2}} \Big |_{\epsilon=0} \notag \\
&= -\p \left(({\bf v}-\epsilon g({\bf y}) )^T ({\bf v}-\epsilon g({\bf y}) )\right)^{{\frac{\p}{2}}-1}  ( {\bf v} -\epsilon g({\bf y}))^T g({\bf y})  \Big |_{\epsilon=0} \notag \\
&= -\p \Trc^{{\frac{\p}{2}}-1} \left[ ({\bf v}-\epsilon g({\bf y}) ) ({\bf v}-\epsilon g({\bf y}) )^T \right]  ( {\bf v} -\epsilon g({\bf y}))^T g({\bf y})  \Big |_{\epsilon=0}  \label{eq:first derivative before evalutation} \\
&= -\p \Trc^{{\frac{\p}{2}}-1} \left[ {\bf v}  {\bf v} ^T \right]   {\bf v}^T g({\bf y}).  \label{eq:first derivative}
\end{align}

Therefore, the function derivative is given by 
\begin{align*}
&\lim_{\epsilon \to 0}  \E \left[\frac{ \Err^{\frac{\p}{2}} ( \X, f(\Y)+ \epsilon g(\Y) )-\Err^{\frac{\p}{2}} ( \X,  f(\Y))}{\epsilon} \right] \\
&=  \E \left[-\p \cdot  \Err^{{\frac{\p}{2}}-1} \left( \X, f(\Y) \right) ( \X-f(\Y))^T g(\Y) \right].
\end{align*}

Finally, for $f_\p(\X|\Y)$ to be optimal it must satisfy 
\begin{align}
 \E \left[  \Err^{\frac{\p-2}{2}} \left( \X, f_\p(\X|\Y) \right) ( \X-f_\p(\X|\Y))^T g(\Y) \right]=0,
\end{align}
for any  admissible $g(\Y)$. This verifies the necessary condition for optimality for $\p> 0$. 

To verify that this is a sufficient condition for optimality we take the second variational derivative of $ \E \left[\Err^{\frac{\p}{2}} ( \X, f(\Y) ) \right]$ and demonstrated that it is always positive for $\p \ge 1$. The fact that 
\begin{align}
  \frac{d^2}{d \epsilon^2}   \Err^{\frac{\p}{2}} ( {\bf x}, f({\bf y})+ \epsilon g({\bf y}) ) \Big |_{\epsilon=0} \ge 0 \text{ for } \p \ge 1,
\end{align} 
 follows since $\Err^{\frac{\p}{2}} ( {\bf x}, f({\bf y})+ \epsilon g({\bf y}) )$ is  a convex function of $\epsilon$ for $\p \ge 1$.  
%

This verifies the sufficient condition  for $\p \ge 1$ and concludes the proof. 

\section{ Proof of Proposition~\ref{prop:GaussianMMPE}}
\label{app:proof: prop:GaussianMMPE}
In  Proposition~\ref{prop:existence of optimal estimator}  we let $X_y \sim p_{X|Y}(\cdot |y)$ and 
therefore have to solve for all $y$
\begin{align}
\min_{v \in \mathbb{R}}\E \left[|X_y -v|^\p \right]. \label{eq:opt problem for cond G}
\end{align}
We know that  $X_y$ is Gaussian with $X_y \sim \mathcal{N} \left(\frac{\sqrt{\snr}y}{1+\snr},\frac{1}{1+\snr} \right)$. The optimization problem in  \eqref{eq:opt problem for cond G} can be transformed into 
\begin{align}
\min_{v \in \mathbb{R}}\E \left[ \left| \frac{Z}{\sqrt{1+\snr}}+\frac{\sqrt{\snr}y}{1+\snr} -v\right|^{\p} \right]
&= \frac{1}{(1+\snr)^{p}}\min_{a \in \mathbb{R}}\E \left[ \left| Z-a\right|^{\p} \right] \label{eq: transformed opt problem} \\
\text{ where } a&=\sqrt{1+\snr} \ v-\frac{\sqrt{\snr}y}{\sqrt{1+\snr}}, \label{eq: opt to opt}
\end{align}
 and where $Z \sim \mathcal{N}(0,1)$.
 Next, by taking the derivative with respect to $a$ in \eqref{eq: transformed opt problem}
\begin{align}
f'(a)=\frac{d}{d a}\E \left[ \left| Z-a\right|^{\p} \right]&=\E \left[ \frac{d}{d a} \left| Z-a\right|^{\p} \right] \label{eq: switch E and d}\\
&=\E \left[ -\p \sign(Z-a) \left| Z-a\right|^{\p-1} \right],
\end{align}
where the interchange of  the order of differentiation and expectation  in \eqref{eq: switch E and d} is possible by Leibniz integral rule \cite{widder1989advanced} which requires verifying that for 
\begin{align}
g(a,z)=\frac{d}{da} |z-a|^{\p}=-\p \ \sign(z-a) \left| z-a\right|^{\p-1}, \label{eq: function g}
\end{align}
we have that $|g(a,z)| \le \theta(z)$ where $\theta(z)$ is integrable. This is indeed the case since
\begin{align}
|\p \ \sign(z-a) \left| z-a\right|^{\p-1}|  \le \p 2^{\p} \left( |z|^{\p-1}+|a|^{\p-1}\right) =\theta(z).
\end{align}
Clearly, $\theta(z)$ is integrable, so the change of the order of differentiation and expectation in \eqref{eq: transformed opt problem} is justified.

Next, observe that for a fixed $a$  the function $g(z,a)$ in \eqref{eq: function g} is a decreasing function of $z$ for any $\p \ge 1$ and  in addition $g(z,a)$ is an odd function around $z=a$. Since $f'(a)$ is an average value of $g(a,z)$ this means that the sign of $f'(a)$ is the same as the sign of $a$, that is, $f'(a)>0$ if $a>0$ and $f'(a)<0$ if $a<0$. Moreover,  if $a=0$ 
\begin{align}
f'(a=0)=\E \left[ -\p \ \sign(Z) \left| Z\right|^{\p-1} \right]=0.
\end{align}
All this implies that $a=0$ is a critical and a minimum point. Therefore, the optimal $\hat{a}=0$ for the optimization problem in \eqref{eq: transformed opt problem} and the optimal $\hat{v}$ for the original optimization problem is  found through \eqref{eq: opt to opt} to be
\begin{align}
\hat{v}= \frac{\sqrt{\snr} \ y}{1+\snr}.
\end{align}

Finally, we compute the $\mmpe(X,\snr,\p)$ for $X \sim \mathcal{N}(0,1)$
\begin{align}
\mmpe(X,\snr,\p)&=\E\left[ \left| X-\frac{\sqrt{\snr} }{1+\snr} Y\right |^{\p} \right]\\
&=\E\left[ \left| \frac{X}{1+\snr} -\frac{\sqrt{\snr}  Z}{1+\snr} \right |^{\p} \right]\\
&\stackrel{a)}{=}\E \left[  \left| \frac{\hat{Z}}{\sqrt{1+\snr}} \right|^{\p}\right|\\
&\stackrel{b)}{=}\frac{2^{{\frac{\p}{2}}} \Gamma \left( \frac{\p+1}{2}\right)}{\sqrt{\pi}(1+\snr)^{{\frac{\p}{2}}}},
\end{align}
where the equalities follow from: a) follows since $X$ and $Z$ are independent Gaussian r.v.'s and have an equivalent distribution given by $\frac{\hat{Z}}{\sqrt{1+\snr}}$ where $\hat{Z}\sim \mathcal{N}(0,1)$; and  b) follows from \eqref{eq:moments of Gaussian} by setting $n=1$.  
This concludes the proof.

\section{ Proof of Proposition~\ref{prop: estimator for two point} }
\label{app: proof of prop: estimator for two point} 
From Proposition~\ref{prop:existence of optimal estimator} we have to minimize $\E \left[|X_y -v|^{\p}\right]$ where $X_y \sim p_{X|Y}(\cdot |y)$.
We have that the joint probability density function of $(X,Y)$ is given by 
\begin{align}
p_{X,Y}(x,y)= q \frac{1}{\sqrt{2 \pi}} \eu^{- \frac{(y-\sqrt{\snr}x_1)^2}{2}} \delta(x-x_1)+(1-q) \frac{1}{\sqrt{2 \pi}} \eu^{- \frac{(y-\sqrt{\snr}x_2)^2}{2}}\delta(x-x_2).
\end{align}
Without loss of generality we assume that $x_1 \le x_2$. By using Bayes' formula we have that 

\begin{align}
\E \left[|X_y -v|^\p \right] =  \frac{|x_1-v|^\p q \frac{1}{\sqrt{2 \pi}} \eu^{- \frac{(y-\sqrt{\snr}x_1)^2}{2}} + |x_2-v|^\p (1-q) \frac{1}{\sqrt{2 \pi}} \eu^{- \frac{(y-\sqrt{\snr}x_2)^2}{2}}}{p_Y(y)}.  \label{eq: to minimize for BPSK}
\end{align}

The minimization of  \eqref{eq: to minimize for BPSK} with respect to $v$ is equivalent to minimizing 
\begin{align}
g(v)=a |x_1-v|^\p+|x_2-v|^\p,\\
\text{ where } a=\frac{q \ \eu^{- \frac{(y-\sqrt{\snr}x_1)^2}{2}} }{(1-q) \ \eu^{- \frac{(y-\sqrt{\snr}x_2)^2}{2}}}.
\end{align}

In piecewise form we can write $g(v)$ as
\begin{align}
g(v) = \left \{  \begin{array}{ll} a(v-x_1)^\p+(v-x_2)^\p, &  x_2 \le v \\
 a(v-x_1)^\p+(x_2-v)^\p, & x_1 < v < x_2 \\
a (x_1-v)^\p+(x_2-v)^\p ,&  v \le x_1 \\
\end{array}\right. ,
\end{align}

with the derivative of $g(v)$ given by 

\begin{align}
g'(v) = \left \{  \begin{array}{ll} ar(v-x_1)^{\p-1}+r(v-x_2)^{\p-1}, &  x_2 \le v \\
 ar(v-x_1)^{\p-1}-r(x_2-v)^{\p-1}, & x_1 < v < x_2 \\
-ar (x_1-v)^{\p-1}-r(x_2-v)^{\p-1} ,&  v \le x_1 \\
\end{array}\right. , \label{eq: derivative g'(v)}
\end{align}

From \eqref{eq: derivative g'(v)} we see that  for the regime $x_2 \le v$  the derivative is positive and therefore the minimum occurs at $v=x_2$. For the regime $v \le x_1$ we have that the derivative is always negative so the minimum occurs at $v=x_1$.  For the regime $x_1 < v < x_2$   the optimal $v$ soves
\begin{align}
g'(v)= a\p(v-x_1)^{\p-1}-\p(x_2-v)^{\p-1}=0,
\end{align} 
that is,
\begin{align}
v= \frac{a^{\frac{1}{\p-1}} x_1+x_2}{a^{\frac{1}{\p-1}}+1}.
\end{align}
Next, by comparing the three candidates for the minimizing $v$,  we have that
\begin{align}
g(v=x_2)&=  a |x_2-x_1|^{\p},\\
g(v=x_1)&=   |x_2-x_1|^{\p},\\
g\left(v=\frac{a^{\frac{1}{\p-1}} x_1+x_2}{a^{\frac{1}{\p-1}}+1}\right)&=  a  \left| x_1-\frac{a^{\frac{1}{\p-1}} x_1+x_2}{a^{\frac{1}{\p-1}}+1} \right|^\p +\left|x_2- \frac{a^{\frac{1}{\p-1}} x_1+x_2}{a^{\frac{1}{\p-1}}+1}  \right|^\p  \notag  \\
&=  \frac{a}{(a^{\frac{1}{\p-1}}+1)^\p}  \left| x_1-x_2 \right|^\p +  \frac{a^{\frac{\p}{\p-1}}}{(a^{\frac{1}{\p-1}}+1)^\p}  \left|x_2- x_1 \right|^\p   \notag \\
&=  \left|x_2- x_1 \right|^\p    \frac{a}{(a^{\frac{1}{\p-1}}+1)^{\p-1}}.
\end{align}

Since $\frac{a}{(a^{\frac{1}{\p-1}}+1)^{\p-1}} \le \min(1,a)$, we have that the minimum  of $g(v)$ occurs at

\begin{align}
v=  \frac{a^{\frac{1}{\p-1}} x_1+x_2}{a^{\frac{1}{\p-1}}+1}=\frac{q^{{\frac{1}{\p-1}}}  \eu^{- \frac{(y-\sqrt{\snr}x_1)^2}{2(r-1)}} \cdot x_1+(1-q)^{{\frac{1}{\p-1}}}  \eu^{- \frac{(y-\sqrt{\snr}x_2)^2}{2(\p-1)}} \cdot x_2}{q^{{\frac{1}{\p-1}}} \ \eu^{- \frac{(y-\sqrt{\snr}x_1)^2}{2(\p-1)}}+(1-q)^{{\frac{1}{\p-1}}}  \eu^{- \frac{(y-\sqrt{\snr}x_2)^2}{2(\p-1)}}}. \label{eq: opt of binary}
\end{align}

Therefore, the optimal estimator is given by the RHS of \eqref{eq: opt of binary}.

Note, that for the case of $\p=1$ the function $g(v)$ reduces to 
\begin{align}
g(v)= a |x_1-v| + |x_2-v|,
\end{align} 
and the minimum occurs at 
\begin{align}
v=\left \{  \begin{array}{ll} x_1,  & a \ge 1 \\
x_2, & a <1
\end{array}\right. .
\end{align}
This implies that for  $\p=1$ the optimal estimator is 
\begin{align}
f_\p(X|Y=y)=\left \{  \begin{array}{ll} x_1,  & a \ge 1 \\
x_2, & a <1
\end{array}\right. ,
\end{align} 
where $a=\frac{q \ \eu^{- \frac{(y-\sqrt{\snr}x_1)^2}{2}} }{(1-q) \ \eu^{- \frac{(y-\sqrt{\snr}x_2)^2}{2}}}$.
This concludes the proof.

\section{ Proof of Proposition~\ref{prop:opt est}}
\label{app:prop:opt est}
The key to deriving all of the claimed properties is the expression of the optimal estimator in Proposition~\ref{prop:existence of optimal estimator}. We prove next all the properties. 
\begin{enumerate}
\item For $ 0\le X \in \mathbb{R}$ suppose that 
\begin{align}
 0 > v_y =f_\p(X|Y=y)= \arg \min_{v}  \E \left[ | X-v|^{\p} | Y= y \right],
\end{align}
then 
\begin{align}
\min_{v}  \E \left[ | X-v|^{\p} | Y= y \right] &= \E \left[ | X-v_y|^{\p} | Y= y \right] \notag\\
&\stackrel{a)}{=} \E \left[ ( X-v_y)^{\p} | Y= y \right] \notag\\
& \stackrel{b)}{ \ge}  \E \left[ X^{\p} | Y= y \right], \label{eq: contradic equation}
\end{align}
where the (in)-equalities follow from:  a) using the assumption that  $X \ge 0$ and $v_y <0$ so $X-v_y >0$ and the absolute value is redundant; and  b) by using  the assumption that  $X \ge 0$ and $v_y <0$ then   $X-v_y \ge X$.
The expression in \eqref{eq: contradic equation} leads to a contradiction since it implies that $v_y=0$ but by assumption $v_y <0$.     Therefore, $v_y= f_\p(X|Y=y) \ge 0$.
This concludes the proof of property 1). 
\item Next we show that $f_\p( a\X+b| \Y)=a f_\p( \X| \Y)+b$. Let 
\begin{align}
{\bf  v_{\bf y}} =f_\p(\X| \Y={\bf y})&=\arg \min_{{\bf v}}  \E \left[  \Trc^\frac{\p}{2} ( \X -{\bf v} ) (\X -{\bf v})^T| \Y={\bf y} \right],
\end{align} 
then 
\begin{align}
f_\p( a\X+b| \Y={\bf y})&=\arg \min_{{\bf v}}  \E \left[  \Trc^\frac{\p}{2} (a \X+b -{\bf v} ) (a\X+b -{\bf v})^T \Big| \Y={\bf y} \right]\notag\\
&=\arg \min_{{\bf v}} a^{\p} \ \E \left[  \Trc^\frac{\p}{2} \left( \X - \frac{{\bf v}-b}{a} \right) \left(\X -\frac{ {\bf v}-b}{a}\right)^T \Big| \Y={\bf y} \right]\notag\\
&\stackrel{a)}{=} \arg \min_{{\bf v}}  \E \left[  \Trc^\frac{\p}{2} \left( \X - \frac{{\bf v}-b}{a} \right) \left(\X -\frac{{\bf v}-b}{a}\right)^T \Big| \Y={\bf y} \right]\notag\\
&\stackrel{b)}{=} a {\bf  v_{\bf y}}+b\notag\\
&=a  f_\p(\X| \Y={\bf y})+b,
\end{align}
where the equalities follow from: a) since scaling the objective function does not change the optimizer; and  b) since the minimum is attained at $\frac{{\bf v}-b}{a}={\bf v_y}$. This concludes the proof of property 2).

\item Next, we show that $f_\p(g(\Y)   | \Y={\bf y}) =g(\Y)$.
Since,
\begin{align}
f_\p(g(\Y)   | \Y={\bf y}) &=\arg \min_{{\bf v}}  \E \left[ \Err^\frac{\p}{2} (g(\Y), {\bf v}) | \Y={\bf y} \right] \\
&=\arg \min_{{\bf v}} \int  \Err^\frac{\p}{2} (g({\bf y}), {\bf v}) p_{\X | \Y} ({\bf x}| {\bf y}) d{\bf x}\\
&=\arg \min_{{\bf v}} \Err^\frac{\p}{2} (g({\bf y}), {\bf v}) \\
&=g({\bf y}).
\end{align}

This concludes the proof of property 3).

\item Follows from property 3) by taking $g(\Y)=f_\p(\X| \Y)$.

\item Observe that  for the Markov chain  $\X \to  \Y_{\snr_0}  \to \Y_{\snr}$  we have 
\begin{align}
 p_{\X| \Y_{\snr_0} , \Y_{\snr}}({\bf x|y_{\snr_0},y_{\snr}})= p_{\X| \Y_{\snr_0}}({ \bf x|y_{\snr_0}}). \label{eq: Marko chain}
 \end{align}
 By using Proposition~\ref{prop:existence of optimal estimator} we have that
 \begin{align*}
f_\p \left(\X|\Y_{\snr_0}={\bf y_{\snr_0}}, \Y_{\snr} ={\bf y_{\snr}}\right)&=  \arg \min_{{\bf v} \in \mathbb{R}^n} \E \left[ \Err^\frac{\p}{2}(\X,{\bf v}) |\Y_{\snr_0}={\bf y_{\snr_0}}, \Y_{\snr} ={\bf y_{\snr}} \right] \\
 &= \arg \min_{{\bf v} \in \mathbb{R}^n} \int \Err^\frac{\p}{2}( { \bf x},{\bf v})  p_{\X| \Y_{\snr_0} , \Y_{\snr}}({\bf x|y_{\snr_0},y_{\snr}}) {\bf dx}\\
  &\stackrel{a)}{=}  \arg \min_{{\bf v} \in \mathbb{R}^n} \int \Err^\frac{\p}{2}( { \bf x},{\bf v})  p_{\X| \Y_{\snr_0} }({\bf x|y_{\snr_0}})  {\bf dx}\\
 &=  \arg \min_{{\bf v} \in \mathbb{R}^n} \E \left[ \Err^\frac{\p}{2}(\X,{\bf v}) |\Y_{\snr_0}={\bf y_{\snr_0}} \right] \\&=f_\p \left( \X | \Y_{\snr_0}={\bf y_{\snr_0}}\right)
  \end{align*}
  where the equality in a) follows from \eqref{eq: Marko chain}.
  \item See Fig.~\ref{fig:Orthonognality} for the counter example. 
\end{enumerate}
This concludes the proof.

\section{Proof of the bound in Proposition~\ref{prop:change of measure}}
\label{app: prop:change of measure}

We define 
\begin{align}
\hat{\Y}_{\snr}=\Y_{\snr_0}+\Z',
\end{align}
where $\Z' \sim \mathcal{N}(0, \sigma^2 \I )$ whith $\sigma^2=\frac{\snr_0-\snr}{\snr}$ is independent of $\Y_{\snr_0}$, $\X$ and $\Z$. Observe that $\hat{\Y}_{\snr}$ and $\Y_{\snr}$ have  the same SNR's and therefore
\begin{align}
\mmpe(\X|\Y_{\snr};\p)=\mmpe(\X|\hat{\Y}_{\snr};\p).
\end{align}

By performing a change of measure we have 
\begin{align}
 n \ \mmpe(\X|\hat{\Y}_{\snr};\p) &= \inf_{f} \E \left[ \Err^\frac{\p}{2} \left(\X,f(\hat{\Y}_{\snr})\right) \right] \\
 &= \inf_{f} \E \left[\Err^\frac{\p}{2} \left(\X,f(\Y_{\snr_0})\right)  L(\X,\Y_{\snr_0})\right],
\end{align}

where $L({ \bf x},{ \bf y})$ is given by
\begin{align}
L({ \bf x},{ \bf y})&=\frac{  p_{\hat{\Y}_{\snr}|\X}({\bf y}| {\bf x})  }{ p_{\Y_{\snr_0}|\X}({\bf y}| {\bf x})}= \frac{ \frac{1}{\sqrt{(2 \pi)^n (1+\sigma^2)} } \eu^{-\frac{1}{2} ({ \bf y}-\sqrt{\snr_0} {\bf x})^T  \frac{1}{1+\sigma^2}\I  ({ \bf y}-\sqrt{\snr_0} {\bf x})}}{ \frac{1}{\sqrt{(2 \pi)^n } } \eu^{-\frac{1}{2} ({ \bf y}-\sqrt{\snr_0} {\bf x})^T \I  ({ \bf y}-\sqrt{\snr_0} {\bf x})}},
\end{align}

and thus 
\begin{align*}
 n \ \mmpe(\X|\hat{\Y}_{\snr};\p) &= \inf_{f} 
\E \left[\Err^\frac{\p}{2} \left(\X,f(\Y_{\snr_0})\right)  L(\X,\Y_{\snr_0})\right]\notag\\
&= \inf_{f} \frac{1}{ \sqrt{1+\sigma^2}} \E \left[ \Err^\frac{\p}{2} \left(\X,f(\Y_{\snr_0})\right) \eu^{\frac{1}{2} \Z^T \I  \Z-\frac{1}{2} \Z^T  \frac{1}{1+\sigma^2}\I  \Z }\right]  \notag\\
&=  \inf_{f} \sqrt{\frac{\snr}{\snr_0 }} \E \left[ \Err^\frac{\p}{2} \left(\X,f(\Y_{\snr_0})\right)  \eu^{\frac{\snr_0-\snr}{2\snr_0} \Z^T  \Z }\right]  \notag\\
&=  \inf_{f} \sqrt{\frac{\snr}{\snr_0 }}\E \left[ \Err^\frac{\p}{2} \left(\X,f(\Y_{\snr_0})\right)  \eu^{\frac{\snr_0-\snr}{2\snr_0}  \sum_{i=1}^n Z_i^2 }\right].
\end{align*}

This concludes the proof.

\section{Proof of Proposition~\ref{prop:higher moments bound 1}}
\label{app:bounds}

\subsection{Proof of the bound in \eqref{eq: chain bound on MMPE}}

The upper bound in \eqref{eq: chain bound on MMPE} follows from the fact that $\E[\X |\Y]$ is in general a suboptimal estimator for a given $p$ thus
\begin{align}
\mmpe(\X,\snr,\p) \le   \frac{1}{n}\E \left[\Err^\frac{\p}{2} \left(\X,\E[\X |\Y]\right)\right].
\end{align}

The lower  bound  in \eqref{eq: chain bound on MMPE} for $\p\ge {\rm q}$ follows  by
\begin{align*}
\mmpe(\X,\snr,{\rm q})=\inf _{f}   \frac{1}{n}\E \left[\Err^\frac{{\rm	q }}{2}(\X,f(\X|\Y))\right]  &= \inf _{f}   \frac{1}{n}\E \left[\Err^{\frac{\p {\rm q}}{2 \p} }(\X,f(\X|\Y)) \right]  \\
&  \stackrel{a)}{\le} \inf _{f}   \frac{1}{n} \left(\E \left[ \Err^\frac{\p}{2}(\X,f(\X|\Y))  \right] \right)^{\frac{\rm q}{\p}}\\
&=   \left(  \frac{1}{n^{\frac{p}{q}}} \inf _{f} \E \left[ \Err^\frac{\p}{2}(\X,f(\X|\Y)) \right] \right)^{\frac{\rm q}{\p}}\\
&=   \left(  \frac{1}{n^{\frac{\p}{\rm q}-1}} \mmpe(\X,\snr,\p) \right)^{\frac{\rm q}{\p}},
\end{align*}
where the inequality in a) follows from  Jensen's inequality and the concavity of $(\cdot)^{\frac{\rm q}{\p}}$.

\subsection{ Proof of the bounds in \eqref{eq: 1/snr bound for higher orders} and \eqref{eq: 1/snr bound for higher orders: 1/2<p<1}}

We now proceed to the proof of the upper bounds in \eqref{eq: 1/snr bound for higher orders} and \eqref{eq: 1/snr bound for higher orders: 1/2<p<1}. We have
\begin{align}
\| \X-\E[\X|\Y]  \|_{\p}&\stackrel{a)}{=} \frac{1}{ \sqrt{\snr}} \|\Z-\E[\Z|\Y] \|_{\p} \notag\\
&\stackrel{b)}{\le}  \frac{1}{ \sqrt{\snr}}  \left(\| \Z \|_{\p}+  \|\E[\Z| \Y] \|_{\p}  \right)  \label{eq: bound on high order error with Z},
 \end{align}
where the (in)-equalities follow from: a) by using Lemma~\ref{lem: equivalence of errors}, b) by using the triangle inequality which holds for $\p \ge 1$. 

Next, the term  $  \|\E[\Z| \Y] \|_p$ can be further bound as follows:
\begin{align}
n^{\frac{1}{\p}}  \|\E[\Z| \Y] \|_\p&=\E^{\frac{1}{\p}} \left[ \Trc^\frac{\p}{2}  \left(\E[\Z|\Y] \E^T[\Z|\Y]  \right)\right] \notag \\&=\E^{\frac{1}{\p}} \left[  \left( \sum_{i=1}^n \E^2[Z_i|\Y] \right)^\frac{\p}{2} \right] \notag\\
  & \stackrel{a)}{\le} \E^{\frac{1}{\p}} \left[  \left( \sum_{i=1}^n \E[Z_i^2|\Y] \right)^\frac{\p}{2} \right] \notag\\
&=\E^{\frac{1}{\p}} \left[  \E^\frac{\p}{2} \left[ \sum_{i=1}^nZ_i^2 | \Y\right] \right] \notag\\
  &=\E^{\frac{1}{\p}} \left[   \E^\frac{\p}{2} \left[ \Trc ( \Z \Z^T) | \Y\right] \right], \label{eq: bound on conditional expectation}
\end{align}
where the inequality in a) follows from using Jensen's inequality.
Depending on whether $\frac{\p}{2} \le 1$ or $\frac{\p}{2} \ge 1$ we bound \eqref{eq: bound on conditional expectation}
as follows:
\begin{subequations}
\begin{align}
&\text{ for $\p \ge 2$ : }  \E^{\frac{1}{\p}} \left[   \E^\frac{\p}{2} \left[ \Trc ( \Z \Z^T) | \Y\right] \right]  \stackrel{a)}{\le} \E^{\frac{1}{\p}} \left[   \E \left[ \Trc^\frac{\p}{2} ( \Z \Z^T) | \Y\right] \right] =\E^{\frac{1}{\p}} \left[   \Trc^\frac{\p}{2} ( \Z \Z^T)  \right] = n^{\frac{1}{\p}}  \|\Z \|_\p,\\
&\text{ for $1 \le \p < 2$ : }  \E^{\frac{1}{\p}} \left[   \E^\frac{\p}{2} \left[ \Trc ( \Z \Z^T) | \Y\right] \right]  \stackrel{b)}{\le} \E^{\frac{1}{2}} \left[   \E \left[ \Trc ( \Z \Z^T) | \Y\right] \right] =\E^{\frac{1}{2}} \left[   \Trc ( \Z \Z^T)  \right] = n^{\frac{1}{2}}  \|\Z \|_2,
\end{align}
\label{eq: for any p with Z}
\end{subequations}
where the inequalities follow from: a)   by using Jensen's inequality on a convex function $x^r$ for $r \ge 1$; and  b)  by using Jensen's inequality on a concave function $x^r$ for $r \le 1$.

By putting \eqref{eq: bound on high order error with Z}, \eqref{eq: bound on conditional expectation}  and \eqref{eq: for any p with Z} together we get 
\begin{subequations}
 \begin{align}
&\text{ for $\p \ge 2$ : } \| \X-\E[\X|\Y]  \|_\p \le  \frac{2}{ \sqrt{\snr}}  \|\Z \|_\p,\\
& \text{ for $1 \le \p < 2$ : } \| \X-\E[\X|\Y]  \|_\p \le  \frac{1}{ \sqrt{\snr}} \left(\|\Z \|_\p +n^{\frac{1}{2}-\frac{1}{\p}} \|\Z \|_2  \right).
 \end{align}
 \label{eq: bound via Z}
 \end{subequations}

The second term in the minimum of \eqref{eq: 1/snr bound for higher orders} and \eqref{eq: 1/snr bound for higher orders: 1/2<p<1} is shown by assuming that  $\| \X \|_\p$ is finite and  by  mimicking the steps leading to  the bound in  \eqref{eq: bound via Z}.  We have
\begin{subequations} 
 \begin{align}
&\text{ for $\p \ge 1$ : } \| \X-\E[\X|\Y]  \|_\p \le 2  \|\X \|_\p,\\
& \text{ for $1 \le \p < 1$ : } \| \X-\E[\X|\Y]  \|_\p \le  \left(\|\X \|_\p +n^{\frac{1}{2}-\frac{1}{\p}} \|\X \|_2  \right).
 \end{align}
 \label{eq: bound via X}
 \end{subequations}

Taking the minimum, between \eqref{eq: bound via Z} and \eqref{eq: bound via X}   concludes the proof.

\subsection{Proof of the bound in \eqref{eq: p-th error bounds} }
The first part of the  bound  in \eqref{eq: p-th error bounds} follows by choosing $f({\bf y})=\frac{{\bf y}}{\sqrt{\snr}}$ in the definition of then MMPE, and hence
\begin{align}
\mmpe(\X,\snr,\p) &\le \left \| \X-\frac{\Y}{\sqrt{\snr}} \right\|_{\p}^{\p} \notag\\
&=  \frac{1}{\snr^\frac{\p}{2}}\| \Z \|_{\p}^{\p}. \label{eq:bound1}
\end{align}
The bound holds as long as  $\| \Z \|_{\p}^{\p}=\E \left[  \left( \sum_{i=1}^n Z_i^2 \right)^\frac{\p}{2} \right] $ is finite which is  the case for $\p \ge 0$.

The second bound follows by choosing $f({\bf y})=0$ in the definition of then MMPE, and hence
\begin{align}
\inf _{f} \E[\Trc^\frac{\p}{2} (\X-f(\X|\Y)) (\X-f(\X|\Y))^T  ]  \le \E \left[ \Trc^\frac{\p}{2} \left(\X\X^T \right) \right],  \label{eq:bound2}
\end{align}
which holds for any $\p$ as long as $\E \left[ \Trc^\frac{\p}{2} \left(\X\X^T \right) \right] $ exists. 

The proof of the upper bound in \eqref{eq: p-th error bounds} is completed by taking the minimum of the bound in \eqref{eq:bound1} and \eqref{eq:bound2}. 
This concludes the proof.

\section{Proof of the bound in Proposition~\ref{prop:GaussianHardesToEstimate} }
\label{app: GaussianHardesToEstimate}

First we show that if $\| \X \|_{\p} \le \| \Z \|_{\p}$ then 
\begin{align}
\mmpe(\X,\snr,\p) \le \kappa_{\p,\snr} \frac{  \|\Z \|_{\p}^{\p} }{(1+\snr)^\frac{\p}{2}}. \label{eq: 1/(1+snr)  bound;temp}
\end{align} 
 Consider the following sub-optimal estimator  $ f(\Y)= \frac{ \sqrt{\snr}}{1+\snr} \Y$ 
\begin{align}
\mmpe(\X,\snr,\p) &\le \left\| \X - \frac{ \sqrt{\snr}}{1+\snr} \Y \right\|_{\p}^{\p} \notag\\
&= \left\| \frac{1}{1+\snr}  \X - \frac{ \sqrt{\snr}}{1+\snr} \Z \right\|_{\p}^{\p} \notag\\
&= \frac{ \left\| \X -  \sqrt{\snr}\Z \right\|_{\p}^{\p} }{(1+\snr)^{\p}} \notag\\
& \stackrel{a)}{\le} \frac{ \left(\| \X \|_\p  + \sqrt{\snr} \|\Z \|_\p\right)^{\p} }{(1+\snr)^{\p}} \notag\\
&  \stackrel{b)}{\le} \frac{ \left( 1+\sqrt{\snr} \right)^{\p} \|\Z \|_{\p}^{\p} }{(1+\snr)^{\p}},\\
&= \kappa_{\p,\snr}   \frac{  \|\Z \|_{\p}^{\p} }{(1+\snr)^\frac{\p}{2}}, \text{ where }  \kappa_{\p,\snr}^{\frac{1}{\p}}=\frac{1+\sqrt{\snr}}{\sqrt{1+\snr}},
\end{align}
where the (in)-equalities follow from: a) triangle inequality and scaling property of the norm; and b) by using the assumption that $  \| \X \|_{\p} \le   \|\Z \|_{\p}$.

Next, let $\X= \sigma {\bf U}$.  Then  $\| \X \|_\p = \|  \sigma {\bf U} \|_\p \le \sigma \| \Z \|_\p$ and  therefore $\| {\bf U} \|_\p \le  \| \Z \|_\p$, so by using the bound in \eqref{eq: 1/(1+snr)  bound;temp} we have that 
\begin{align}
\mmpe(\X,\snr,\p)&= \mmpe(\sigma {\bf U},\snr,\p)\\
&\stackrel{a)}{=}\sigma^{\p} \mmpe({\bf U},\sigma^2\snr,\p) \\
& \stackrel{b)}{\le}  \kappa_{\p, \sigma^2\snr} \sigma^{\p}   \frac{\| \Z \|_{\p}^{\p} }{ (1+ \snr \sigma^2)^\frac{\p}{2}},
\end{align} 
where the (in)-equalities follow from: a) by using the scaling property of the MMPE in Proposition~\ref{prop:Scaling}; and  b) by using the bound in \eqref{eq: 1/(1+snr)  bound;temp}.

Observe that the bound in \eqref{eq: 1/(1+snr)  bound;temp} is achieved asymptotically by using $ \X_G \sim \mathcal{N}( {\bf 0}, \sigma^2 \I)$  since by Proposition~\ref{prop:GaussianMMPE} and the scaling property in Proposition~\ref{prop:Scaling} we have that 
\begin{align}
\mmpe(\X_G,\snr,\p)= \frac{\sigma^{\p}  \| \Z \|_{\p}^{\p} }{ (1+ \snr \sigma^2)^\frac{\p}{2}}.  
\end{align}
This concludes the proof.

\section{Proof of Proposition~\ref{prop:bound on discrete inputs}}
\label{app:prop:bound on discrete inputs}

We use the approach of \cite{OptimalPowerAllocationOfParChan}. Suppose we use the following sub-optimal decoder:
\begin{align}
g(\Y)=  \sum_{i=1}^N{\bf x}_i  1_{  B_{{\bf x}_i}(r)} (\Y),
\end{align} 
where $B_{{\bf x}_i}(r)$ is the $n$-dimensional ball of radius $ r=\frac{ \sqrt{\snr  \ d^2({\bf x}_i)}}{2}$ centered at ${\bf x}_i$. Then
\begin{align}
  n \  \mmpe(\X,\snr,\p) & \le \E \left[ \Err^\frac{\p}{2} (\X, g(\Y) )\right]\\
  & =  \sum_{i=1}^N  p_i\E \left[ \Err^\frac{\p}{2} (\X, g(\Y) ) | {\bf x}_i\right] \\
  &= \sum_{i=1}^N  p_i\E \left[ \Err^\frac{\p}{2} (\X, g(\Y) ) | \X= {\bf x}_i, \Y \in B_{{\bf x}_i}(r) \right] \mathbb{P}[ \Y \in B_{{\bf x}_i}(r) |\X={\bf x}_i ]  \\
  &+ \sum_{i=1}^N  p_i\E \left[ \Err^p (\X, g(\Y) ) | \X={\bf x}_i, \Y \notin B_{{\bf x}_i}(r) \right] \mathbb{P}[ \Y \notin B_{{\bf x}_i}(r) |\X={\bf x}_i].
\end{align}
Since $\E \left[ \Err^\frac{\p}{2} (\X, g(\Y) ) | \X= {\bf x}_i, \Y \in B_{{\bf x}_i}(r)\right]=0$ we have that
\begin{align}
  n \  \mmpe(\X,\snr,\p)  \le \sum_{i=1}^N  p_i\E \left[ \Err^\frac{\p}{2} (\X, g(\Y) ) | \X={\bf x}_i, \Y \notin B_{{\bf x}_i}(r) \right] \mathbb{P}[ \Y \notin B_{{\bf x}_i}(r) | \X={\bf x}_i].
\end{align}

First, observe that
\begin{align*}
\mathbb{P}[ \Y \notin B_{{\bf x}_i}(r) ,\X={\bf x}_i]&= \mathbb{P}[ \Z \notin S_{{\bf 0},r} ] \\
&= 1 - \mathbb{P}[ \Z \in S_{{\bf 0},r} ] \\
&=1- \int_{S_{{\bf 0},r}} \frac{1}{\sqrt{(2 \pi)^n }} \eu^{-\frac{1}{2} {\bf x}^T {\bf x}} d{\bf x}\\
&= 1- \frac{1}{2^{\frac{n}{2}-1} \Gamma \left( \frac{n}{2}\right)} \int_{0}^r  \rho^{n-1}\eu^{-\frac{\rho^2}{2}} d \rho\\
&=  \frac{\int_{r}^\infty  \rho^{n-1}\eu^{-\frac{\rho^2}{2}} d \rho}{2^{\frac{n}{2}-1} \Gamma \left( \frac{n}{2}\right)}  \\
& =\frac{ \Gamma \left( \frac{n}{2}; \frac{r^2}{2}\right)}{\Gamma \left( \frac{n}{2}\right)} \\
&= \bar{Q}\left( \frac{n}{2}; \frac{r^2}{2}\right).
\end{align*}

Second, observe that 
\begin{align*}
\E \left[ \Err^\frac{\p}{2} (\X, g(\Y) ) | \X={\bf x}_i, \Y \notin B_{{\bf x}_i}(r) \right] 
&= \E \left[ \| {\bf x}_i -{ \bf x}_j \|^{\p} | \X={\bf x}_i, \Y \notin B_{{\bf x}_i}(r) \right]  \\
& \le d_{\max}^{\p} (\X_D).
\end{align*}

Therefore, we have that
\begin{align*}
\mmpe(\X_D,\snr,\p)& \le  d_{\max}^{\p}(\X_D)  \frac{ \sum_{i=1}^N p_i \bar{Q} \left( \frac{n}{2}; \frac{\snr \ d_{ {\bf x}_i}^2(\X_D) }{8}\right)}{ n }  \\
&  \le  d_{\max}^{\p}(\X_D)   \frac{ \bar{Q} \left( \frac{n}{2}; \frac{\snr \ d_{ \min }^2(\X_D) }{8}\right)\sum_{i=1}^N p_i }{ n } \\
&  =  d_{\max}^{\p}(\X_D)   \frac{ \bar{Q} \left( \frac{n}{2}; \frac{\snr \ d_{ \min }^2(\X_D) }{8}\right)}{ n },
\end{align*}
where the last inequality follows since $\bar{Q}(x,a)$ is decreasing in $a$. 
This concludes the proof.

\section{Proof of Theorem~\ref{prop:generalization of continuous Fano's inequality} }
\label{proof:generalization of continuous Fano's inequality}
Let $ \W_{\bf v}=\U_{\bf v}-g({\bf v})$ where $g(\cdot)$ is a deterministic function and $\U_{\bf v} \sim p_{\U | \V} (\cdot  | {\bf v})$. By \cite[Theorem 3]{MomentsEntropyInequality} we have 
\begin{align}
\frac{ n^{\frac{1}{\p}} \| \W_{\bf v} \|_\p }{e^{ \frac{1}{n}h_{\eu}(\W_{\bf v})}} \ge    \frac{1}{k_{n,\p}}, \
 k_{n,\p}:= \frac{ \sqrt{\pi}\left( \frac{\p}{n} \right)^{\frac{1}{\p}}\eu^{\frac{1}{\p}} \Gamma^{\frac{1}{n}} \left( \frac{n}{\p}+1\right)}{ \Gamma^{\frac{1}{n}}\left( \frac{n}{2}+1 \right)},  \label{eq: bound moments and enropies: any n} 
\end{align} 
where $h_{\eu}(\cdot)$ is the differential entropy measured in nats. 
Moreover, observe that  
$h_{\eu}(\W_{\bf v}) =h_{\eu}(\U_{\bf v} -g({\bf v}))=h_{\eu}(\U_{\bf v})$ 
due to the  translation invariance of the differential entropy. 
Therefore, by rearranging  \eqref{eq: bound moments and enropies: any n}  and by using the translation invariance of the differential entropy, we get 
\begin{align}
\frac{1}{n}h_{\eu}(\U_{\bf v}) \log(\eu) &\le  \log \left( k_{n,\p}  \cdot n^{\frac{1}{\p}} \| \W_{\bf v} \|_\p\right),
\label{eq: bound on marginal conditional entropy}
\end{align} 
where from \eqref{eq: defintion of the norm} we have  $n^{\frac{1}{\p}}  \| \W_{\bf v} \|_\p= \E^{\frac{1}{\p}} \left[ \Err^\frac{\p}{2}(\U, g(\V)) | \V={\bf v} \right]$.
By taking the expectation on both sides of \eqref{eq: bound on marginal conditional entropy} with respect to $p_{\V}({\bf v})$ we arrive at
\begin{align*}
&n^{-1}h_{\eu}(\U| \V) \log(\eu)=  n^{-1}h(\U| \V) \\
&  \le  \frac{1}{\p}\E \left[  \log \left( k_{n,\p}^{\p}  \cdot n \cdot \frac{1}{n} \cdot  \E \left[ \Err^\frac{\p}{2}(\U, g(\V)) | \V \right]\right) \right] \\
&  \stackrel{a)}{\le}   \frac{1}{\p} \log \left( k_{n,\p}^{\p}  \cdot n \cdot \frac{1}{n} \cdot \E \left[  \E \left[ \Err^\frac{\p}{2}(\U, g(\V)) | \V \right] \right]\right) \\
&  =   \frac{1}{\p} \log \left( k_{n,\p}^{\p}  \cdot n \cdot \frac{1}{n} \cdot  \E \left[ \Err^\frac{\p}{2}(\U, g(\V)) \right] \right) \\
&  =   \log \left( k_{n,\p}  \cdot n^{\frac{1}{\p}} \cdot  \| \U-g(\V) \|_\p \right),
\end{align*}
where the inequality in a) follows from Jensen's inequality. 
Finally, since this bound holds for any deterministic function $g(\cdot)$, to tighten this bound,  and due to the monotonicity of the $\log$ function, we may pick $g(\cdot)$ to be the optimal $\p$-th estimator of $\U$.
This concludes the proof.

\section{Proof of Theorem~\ref{prop:OWimproved}}
\label{app:prop:OWimproved}
Let $(\U,\X_D,\Z)$ be mutually independent. By the data processing inequality and the assumption in \eqref{eq: assumption on U} we have
\begin{align}
I( \X_D; \Y) & 
\geq I( \X_D+ \U;\Y) = h(\X_D+ \U)- h(\X_D+ \U| \Y) \notag\\
&= 
H(\X_D) +h(\U)- h(\X_D+ \U| \Y).
\label{eq: OWproof-decomp}
\end{align}
Next, by using Theorem~\ref{prop:generalization of continuous Fano's inequality}, 
we have that the last term of \eqref{eq: OWproof-decomp} can be bounded as
\begin{align}
&n^{-1} h(\X_D+ \U| \Y) \le  \log \left( k_{n,\p}  \cdot n^{\frac{1}{\p}} \cdot  \|  \X_D+\U-g(\Y) \|_\p \right) .
\label{eq: third term after triangular inequality}
\end{align} 
Next,  by  combining  \eqref{eq: OWproof-decomp} and \eqref{eq: third term after triangular inequality} and taking $g(\Y)=f_p(\X| \Y)$ we have that 
\begin{align}
I( \X_D; \Y) & \ge H(\X_D)- \gap_p, \\
 n^{-1} \mathsf{gap}_\p& \le  \inf_{\U \in \mathcal{K} }  \left( L_{1,\p}(\U,\X_D) + L_{2,\p}(\U)  \right),  \notag\\
G_{1,\p}(\U,\X_D)&=  \log \left(   \frac{  \|  \U+\X_D-f_\p(\X| \Y) \|_\p}{ \| \U\|_\p} \right)   \stackrel{ \text{ for  $\p \ge 1$} }{\le}     \log \left( 1+\frac{ \mmpe^{\frac{1}{\p}}(\X_D,\snr,\p) }{\| \U \|_\p}  \right), \label{eq: bound on L1p term} \\
G_{2,\p}(\U)&= \log \left(  \frac{  k_{n,\p}  \cdot n^{\frac{1}{\p}} \cdot  \| \U \|_\p}{\eu^{ \frac{1}{n} h_{\eu}(\U)} }\right),
\end{align} 
where inequality in \eqref{eq: bound on L1p term} follows by the triangle inequality which holds for $\p \ge 1$.

Finally, the proof  concludes by taking $g(\Y)=f_p(\X| \Y)$.

\section{ Proof of Theorem~\ref{prop: Gap in Ozarow-Wyner at n infinity}} 
\label{app:prop: Gap in Ozarow-Wyner at n infinity}
To show that $\lim_{ n \to \infty} G_{2,\p} (\U) = 0$ we show that 
\begin{align}
 \lim_{n \to \infty} \frac{  k_{n,\p}  \cdot n^{\frac{1}{\p}} \cdot  \| \U \|_\p}{\eu^{ \frac{1}{n} h_{\eu}(\U)} }=1.
\end{align} 
First of all observe that  using \eqref{eq: bound moments and enropies: any n} in Appendix~\ref{proof:generalization of continuous Fano's inequality} 
\begin{align}
1 \le \frac{  k_{n,\p}  \cdot n^{\frac{1}{\p}} \cdot  \| \U \|_\p}{\eu^{ \frac{1}{n} h_{\eu}(\U)} }.
\end{align} 
Next, we show an upper bound. Note that if $\U$ is uniform over a ball $B_0(r)$ of radius $r=d_{\min}(\X_D)/2$ then 
\begin{align}
h(\U)&=\log \left({\rm Vol}(B_0(r)) \right), \\
\text{ where } {\rm Vol}(B_0(r))&=\frac{\pi^{n/2}}{\Gamma \left( \frac{n}{2}+1\right)}  r^n. \label{eq:entropyOfU}
\end{align}
Moreover, the norm $\U$ can be upper bounded by
\begin{align}
\| \U \|_{\p}^{\p}&= \frac{1}{n}\frac{1}{{\rm Vol}(B_0(r))} \int_{B_0(r)}  \left( \sum_{i=1}^n  u_i^2 \right)^{\frac{\p}{2}} du_1 du_2 \cdot \cdot \cdot du_n \notag\\
& \le  \frac{1}{n}\frac{1}{{\rm Vol}(B_0(r))}     \int_{B_0(r)}  \left(  r^2 \right)^{\frac{\p}{2}}  du_1 du_2 \cdot \cdot \cdot du_n = \frac{r^{\p}}{n}. \label{eq:boundOnNormOfU}
\end{align}

Therefore, by using \eqref{eq:boundOnNormOfU}  and  \eqref{eq:entropyOfU}
\begin{align}
\frac{  k_{n,\p}  \cdot n^{\frac{1}{\p}} \cdot  \| \U \|_\p}{\eu^{ \frac{1}{n} h_{\eu}(\U)} }  &\le   \frac{ k_{n,\p}  \cdot  \Gamma^{\frac{1}{n}} \left( \frac{n}{2}+1\right) }{ \sqrt{\pi}} \\
&= \frac{ \sqrt{\pi}\left( \frac{\p}{n} \right)^{\frac{1}{\p}}\eu^{\frac{1}{\p}} \Gamma^{\frac{1}{n}} \left( \frac{n}{\p}+1\right)  \Gamma^{\frac{1}{n}}\left( \frac{n}{2}+1 \right)}{ \Gamma^{\frac{1}{n}}\left( \frac{n}{2}+1 \right) \sqrt{\pi}} \\
&=\left( \p \eu  \right)^{\frac{1}{\p}} \left( \frac{1}{n} \right)^{\frac{1}{\p}} \Gamma^{\frac{1}{n}} \left( \frac{n}{\p}+1\right).
\end{align}

Next by using  the Stirling's approximation $\Gamma(x+1)=  \sqrt{2 \pi x} \left(  \frac{x}{ \eu}\right)^x + o(x)$ we have that
\begin{align}
\left( \frac{1}{n} \right)^{\frac{1}{\p}} \Gamma^{\frac{1}{n}} \left( \frac{n}{\p}+1\right) \le   \left( \frac{2 \pi n}{\p}\right)^{\frac{1}{n}} \left(  \frac{1}{\p \eu } \right)^\frac{1}{\p} + o \left( \frac{n}{\p} \right),
\end{align}
and therefore
\begin{align}
\frac{  k_{n,\p}  \cdot n^{\frac{1}{\p}} \cdot  \| \U \|_\p}{\eu^{ \frac{1}{n} h_{\eu}(\U)} } \le  \left( \frac{2 \pi n}{\p}\right)^{\frac{1}{n}}  + o \left( \frac{n}{\p} \right) \stackrel{ \text{ as $n \to \infty$} }{ \to } 1.
\end{align} 

This shows that $\lim_{ n \to \infty} G_{2,\p} (\U) = 0$. 

Next, we show that  $\lim_{ n \to \infty} G_{1,\p} (\X_D,\U) = 0$ by showing that  $\lim_{ n \to \infty}  \frac{ \mmpe(\X_D,\snr,\p)}{ \| \U \|_\p} = 0$. 
First, observe that  by using the bound in Proposition~\ref{prop:bound on discrete inputs}
\begin{align}
\mmpe^{\frac{1}{\p}}(\X_D,\snr,\p) \le  \frac{ d_{\max}  Q^{\frac{1}{\p}} \left( \frac{n}{2}; \frac{\snr d^2_{ \min} }{8}\right)}{n^{\frac{1}{\p}} }, 
\end{align}
and  by using  \eqref{eq: mommentsUniform} we have that

\begin{align}
  \frac{ \mmpe^{\frac{1}{\p}}(\X_D,\snr,\p)}{ \| \U \|_p}  &\le   \frac{    \frac{ d_{\max}(\X_D)  Q^{\frac{1}{\p}} \left( \frac{n}{2}; \frac{\snr d^2_{ \min}(\X_D) }{8}\right)}{n^{\frac{1}{\p}} } }{  \frac{d_{\min}(\X_D)}{2(\p+n)^{\frac{1}{\p}}}} \\
      &=2 \frac{ d_{\max}(\X_D)}{ d_{\min}(\X_D)} \sqrt[\p]{   \frac{ (\p+n)   \bar{Q} \left( \frac{n}{2}; \frac{\snr d^2_{ \min}(\X_D) }{8}\right) }{ n }} .\label{eq:conditionPakinng}
\end{align}

 This concludes the proof.

\section{On Finding the Optimal $r$ in the proof of Theorem~\ref{prop: bound through MMPE} }
\label{app: approximate r}

We must solve the following optimization problem:
 \begin{align}
 \min_{r> \frac{2}{\gamma}}g(r)&=M^{ \frac{ \gamma r-2}{r-2}}   \frac{G^\frac{r (1-\gamma)}{r-2}}{N^{\frac{2(1-\gamma)}{r-2}}} \Gamma^{\frac{2(1-\gamma)}{r-1}}(n/2+r/2), \\
 M&=\mmse(\X,\snr_0),\\
 G&=\frac{8}{\snr_0},\\
 N&=n  \Gamma \left( \frac{n}{2} \right)=2 \Gamma \left( \frac{n}{2}+1 \right).
 \end{align}
Instead of optimizing $g(r)$ we will focus on  optimizing  $h(r)=\ln(g(r))$ where
 \begin{align}
 h(r)=\frac{ \gamma r-2}{r-2} \ln(M)+\frac{r (1-\gamma)}{r-2} \ln(G)-\frac{2(1-\gamma)}{(r-2)} \ln(N)+ \frac{2(1-\gamma)}{r-2} \ln( \Gamma(n/2+r/2)). \label{eq: h(r) objectibe functions}
 \end{align}

Unfortunately,  a closed form solution for the optimum of \eqref{eq: h(r) objectibe functions}  is difficult to find and instead  we look for an approximate solution.  This is done by using Stirling's formula $\Gamma(x+1) \approx \sqrt{2 \pi x} \left( \frac{x}{\eu} \right)^x$. We have 
\begin{align}
\Gamma(n/2+r/2)=\Gamma(n/2+r/2-1+1)\approx \sqrt{2 \pi \left(r/2+\frac{n-2}{2}\right)} \left( \frac{r/2+\frac{n-2}{2}}{\eu} \right)^{r/2+\frac{n-2}{2}}.
\end{align}
Now, we seek to optimize the following expression:
\begin{align}
g(r) \approx \frac{M^{\frac{\gamma\, r - 2}{r - 2}}\, N^{\frac{2(\gamma - 1)}{ \left(r - 2\right)}}}{G^{\frac{ r\, \left(\gamma - 1\right)}{r - 2}} } \left( \sqrt{2 \pi (\frac{r}{2}+\frac{n-2}{2})} \left( \frac{\frac{r}{2}+\frac{n-2}{2}}{\eu} \right)^{\frac{r}{2}+\frac{n-2}{2}}
\right)^{\frac{2(1-\gamma)}{r-2}},
\end{align}

that is 

\begin{align}
h(r)& \approx \frac{ \gamma r-2}{r-2} \ln(M)+\frac{r (1-\gamma)}{r-2} \ln(G)-\frac{2(1-\gamma)}{r-2}\ln(N)+\frac{1-\gamma}{r-2}\ln \left(2 \pi \left(\frac{r}{2}+\frac{n-2}{2}\right)\right)
\notag\\&+ \frac{2(1-\gamma)(\frac{r}{2}+\frac{n-2}{2})}{r-2} \ln \left(\frac{r}{2}+\frac{n-2}{2}\right)-\frac{2(1-\gamma)(\frac{r}{2}+\frac{n-2}{2})}{r-2}. \label{eq:approx h(r)}
\end{align}

By taking the derivative of \eqref{eq:approx h(r)} with respect to $r$ we get
\begin{align}
h'(r)
&=\frac{1}{2} \frac{1-\gamma}{(\frac{r}{2}-1)^2}   f(r) \\
  f(r)
&=\ln(M)-\ln(\sqrt{2 \pi} G)+\log(N)+\frac{r-2}{n-2+r} -\frac{n+1}{2} \ln \left(\frac{n-2}{2}+\frac{r}{2}\right)  +\frac{r}{2}+\frac{n-2}{2}\\
& \approx  \ln(M)-\ln(\sqrt{2 \pi}G)+\log(N) -\frac{n+1}{2} \ln \left(\frac{n-2}{2}+1\right)  +\frac{r}{2}+\frac{n-2}{2} \label{eq: approx Hr}
 \end{align}
 where in the last step we used the approximation $\frac{n+1}{2} \ln \left(\frac{n-2}{2}+\frac{r}{2}\right) \approx \frac{n+1}{2} \ln \left(\frac{n-2}{2}+1\right)$ and $\frac{r-2}{n-2+r}\approx 0$ which is reasonable as $n$ becomes large.
 
 Solving $f(r)=0$ in \eqref{eq: approx Hr} we get that  the approximate solution is 
 \begin{align*}
 \frac{r}{2}&= \ln \left( \frac{\sqrt{2 \pi} G  \left(\frac{n}{2} \right)^{\frac{n+1}{2}}}{M N \eu^{\frac{n-2}{2}}} \right) \\
 &=  \ln \left( \frac{ 8 \sqrt{2 \pi} \left(\frac{n}{2} \right)^{\frac{n+1}{2}}}{ \snr_0 \mmse(X,\snr_0)  2 \Gamma \left( \frac{n}{2} +1\right) \eu^{\frac{n-2}{2}}} \right)\\
  &\approx  \ln \left( \frac{ 8 \sqrt{2 \pi} \left(\frac{n}{2} \right)^{\frac{n+1}{2}}}{ \snr_0 \mmse(X,\snr_0)  2  \sqrt{  2\pi \frac{n}{2} } \left( \frac{n}{2 \eu}\right)^{\frac{n}{2}} \eu^{\frac{n-2}{2}}} \right)\\
    &= \ln \left( \frac{ 4 \eu }{ \snr_0 \mmse(X,\snr_0)   } \right),
  \end{align*}
  where in the last approximation we have used Stirling's formula. 
 
 Since, we have a constraint that $r >\frac{2}{\gamma}$ we set $r$ to be 
 \begin{align}
 r \approx  \left \{\begin{array}{ll}  2 \ln \left( \frac{ 4 \eu }{ \snr_0 \mmse(X,\snr_0)   } \right), & \frac{2}{\gamma} \le \ln \left( \frac{ 4 \eu }{ \snr_0 \mmse(X,\snr_0)   } \right)\\  
 \frac{2}{\gamma}, & \frac{2}{\gamma} > \ln \left( \frac{ 4 \eu }{ \snr_0 \mmse(X,\snr_0)   } \right)\\ 
  \end{array} \right. .
 \end{align}
This concludes the proof.

\section{Proof of Proposition~\ref{prop:bounds on derivative of MMSE via MMPE}}
\label{app:prop:bounds on derivative of MMSE via MMPE}
First observe that 
\begin{align}
\inf_f \E[ \Err(\X,f(\Y)) | \Y={\bf y}] =  \E \left[  \Err(\X,\E[\X| \Y] ) |\Y={\bf y} \right].
\label{eq: coditional MMSE}
\end{align}
We  will need the following bounds on trace of ${\bf A} \succeq 0$ where ${\bf A} \in \mathbb{R}^{ n\times n}$
\begin{align}
\frac{1}{n} \Trc({\bf A} )^2 \le \Trc({\bf A} ^2) \le n \Trc({\bf A} )^2 \label{eq: trace matrix bounds}.
\end{align}

For the upper bound we have that 
\begin{align*}
 \Trc \left( \E \left[\cov^2(\X|\Y) \right] \right) &=   \E \left[  \Trc \left(\cov^2(\X|\Y)  \right) \right] \\
 & \stackrel{a)}{\le}  \E \left[ n \Trc^2 \left(\cov(\X| \Y)  \right) \right] \\
  & =  \E \left[ n \Trc^2 \left(\E \left[ (\X -\E[\X| \Y])(\X -\E[\X| \Y])^T| \Y\right]  \right) \right] \\
    & =  \E \left[ n  \E^2 \left[  \Trc (\X -\E[\X| \Y])(\X -\E[\X| \Y])^T |\Y \right]   \right] \\
      & =  \E \left[ n  \E^2 \left[  \Err(\X,\E[\X| \Y] ) |\Y \right]   \right] \\
            & \stackrel{b)}{=}  \E \left[ n   \left(\inf_{f} \E \left[  \Err(\X,f(\Y)) |\Y \right]  \right)^2  \right] \\
            & =  \E \left[ n   \inf_{f} \E^2 \left[  \Err(\X,f(\Y)) |\Y \right]  \right] \\
                  & \stackrel{c)}{\le}  \E \left[ n   \inf_{f} \E \left[  \Err^2(\X,f(\Y)) |\Y \right]  \right] \\
                   & \stackrel{d)}{\le}   n   \inf_{f}  \E \left[  \E \left[  \Err^2(\X,f(\Y)) |\Y \right]  \right] \\
                   & \stackrel{e)}{=}   n   \inf_{f}   \E \left[  \Err^2(\X,f(\Y))  \right]  \\
                   &= n^2 \mmpe(\X,\snr,4),
 \end{align*}
where the (in)-equalities follow from: a) since $\cov(\X| \Y) \succeq 0$ and using the inequality in \eqref{eq: trace matrix bounds}; and b) by using \eqref{eq: coditional MMSE}; c) Jensen's inequality;  d) by using $\E[X_1] \le \E[X_2]$ if  $X_1 \le X_2$; and e) law of total expectation. 

For the lower bound 

\begin{align*}
\frac{1}{n} \Trc \left( \E \left[\cov^2(\X|\Y) \right] \right) &=  \frac{1}{n} \E \left[  \Trc \left(\cov^2(\X|\Y)  \right) \right] \\
 & \stackrel{a)}{\ge}  \frac{1}{n}  \E \left[ \frac{1}{n} \Trc^2 \left(\cov(\X| \Y)  \right) \right]\\
  & \stackrel{b)}{\ge}   \frac{1}{n^2}\E^2 \left[  \Trc \left(\cov(\X| \Y)  \right) \right]\\
    &=  \mmse^2(\X,\snr),
\end{align*}
where the inequalities follow from: a) since $\cov(\X| \Y) \succeq 0$  and by using the  inequality in \eqref{eq: trace matrix bounds}; and b) Jensen's inequality. 

\end{appendices} 
\bibliography{refs}
\bibliographystyle{IEEEtran}
\end{document}